\documentclass[12pt]{amsart}

\usepackage{amsmath, amsthm, amssymb,amscd}
\usepackage{fullpage, paralist,enumerate}
\usepackage{verbatim}
\usepackage[all]{xy}
\usepackage[parfill]{parskip}
\usepackage{graphicx}

\usepackage{afterpage}

    \usepackage[utf8]{inputenc}
\usepackage{pgfplots}

\usepackage{amsmath}
\usepackage{amsfonts}
\usepackage{mathcomp}
\usepackage{textcomp}
\usepackage{amssymb}
\usepackage{latexsym}
\usepackage{graphicx}
\usepackage{courier}
\usepackage[english]{babel}
\usepackage{mathbbol}
\usepackage{multirow}
\usepackage{latexsym}
\usepackage{epsfig}
\usepackage{subfigure}
\usepackage{dcolumn}% Align table columns on decimal point
\usepackage{bm}% bold math
\usepackage{colordvi}
\usepackage{color}
\usepackage{booktabs}
\usepackage{stmaryrd}
\usepackage{cite}
\usepackage[linesnumbered,commentsnumbered,ruled]{algorithm2e}
\usepackage[noend]{algpseudocode}
\usepackage{epstopdf}
\input xy
\xyoption{all}
\newcommand{\commentout}[1]{}

\newtheorem{thm}{Theorem}[section]

\newtheorem{lem}[thm]{Lemma}
\newtheorem{prop}[thm]{Proposition}
\newtheorem{cor}[thm]{Corollary}
\newtheorem{ex}[thm]{Example}
\newtheorem{rmk}[thm]{Remark}
\newcommand{\nwc}{\newcommand*}

\nwc{\ben}{\begin{equation*}}
\nwc{\bea}{\begin{eqnarray}}
\nwc{\beq}{\begin{eqnarray}}
\nwc{\bean}{\begin{eqnarray*}}
\nwc{\beqn}{\begin{eqnarray*}}
\nwc{\beqast}{\begin{eqnarray*}}

\nwc{\eal}{\end{align}}
\nwc{\een}{\end{equation*}}
\nwc{\eea}{\end{eqnarray}}
\nwc{\eeq}{\end{eqnarray}}
\nwc{\eean}{\end{eqnarray*}}
\nwc{\eeqn}{\end{eqnarray*}}

\CompileMatrices

\theoremstyle{remark}

\nwc{\nn}{\nonumber}
\nwc{\mb}{\mathbf}
\nwc{\ml}{\mathcal}

\newcommand{\lt}{\left}
\newcommand{\rt}{\right}

\nwc{\cle}{\preccurlyeq}

\nwc{\lb}{\llbracket}
\nwc{\rb}{\rrbracket}
\nwc{\modpi}{{{\rm mod}\,2\pi}}
\nwc{\tphi}{{{\phi}_0}}
\nwc{\mpc}{\,\mbox{MPC($\gamma$)}\,\,}

\nwc{\vep}{\varepsilon}
\nwc{\ep}{\epsilon}
\nwc{\vrho}{\varrho}
\nwc{\orho}{\bar\varrho}
\nwc{\vpsi}{\varpsi}
\nwc{\lamb}{\lambda}
\nwc{\om}{\omega}
\nwc{\Om}{\Omega}
\nwc{\al}{\alpha}
\nwc{\sgn}{\mbox{\rm sgn}}

\nwc{\IA}{\mathbb{A}} %algebraic
\nwc{\bi}{\mathbf{i}}
\nwc{\ba}{\mathbf{a}}
\nwc{\bc}{\mathbf{c}}
\nwc{\bmb}{\mathbf{b}}
\nwc{\bo}{\mathbf{o}}
\nwc{\IB}{\mathbb{B}}
\nwc{\IC}{\mathbb{C}} %complex
\nwc{\ID}{\mathbb{D}} %Dedekind
\nwc{\IM}{\mathbb{M}} %Dedekind
\nwc{\IP}{\mathbb{P}} %Dedekind
\nwc{\II}{\mathbb{I}} %Dedekind
\nwc{\IE}{\mathbb{E}} %Euklides
\nwc{\IF}{\mathbb{F}} %finite field
\nwc{\IG}{\mathbb{G}} %Gauss
\nwc{\IN}{\mathbb{N}} %natural
\nwc{\IQ}{\mathbb{Q}} %rational
\nwc{\IR}{\mathbb{R}} %real
\nwc{\IT}{\mathbb{T}} %torus
\nwc{\IZ}{\mathbb{Z}} %integers

\nwc{\cE}{{\ml E}}
\nwc{\cP}{{\ml P}}
\nwc{\cQ}{{\ml Q}}
\nwc{\cL}{{\ml L}}
\nwc{\cX}{{\ml X}}
\nwc{\cW}{{\ml W}}
\nwc{\cZ}{{\ml Z}}
\nwc{\cR}{{\ml R}}
\nwc{\cV}{{\ml V}}
\nwc{\cT}{{\ml T}}
\nwc{\crV}{{\ml L}_{(\delta,\rho)}}
\nwc{\cC}{{\ml C}}
\nwc{\cO}{{\ml O}}
\nwc{\cA}{{\ml A}}
\nwc{\cK}{{\ml K}}
\nwc{\cB}{{\ml B}}
\nwc{\cD}{{\ml D}}
\nwc{\cF}{{\ml F}}
\nwc{\cS}{{\ml S}}
\nwc{\cM}{{\ml M}}
\nwc{\cG}{{\ml G}}
\nwc{\cH}{{\ml H}}
\nwc{\bk}{{\mb k}}
\nwc{\bn}{{\mb n}}
\nwc{\bp}{{\mb p}}
\nwc{\bq}{{\mb q}}
\nwc{\bz}{\mb z}
\nwc{\bl}{{\mb l}}
\nwc{\bj}{{\mb j}}
\nwc{\bs}{{\mb s}}
\nwc{\by}{\mathbf{h}}
\nwc{\bZ}{\mathbf{Z}}
\nwc{\bF}{\mathbf{F}}
\nwc{\bE}{\mathbf{E}}
\nwc{\bV}{\mathbf{V}}
\nwc{\bY}{\mathbf Y}
\nwc{\br}{\mb r}
\nwc{\pft}{\cF^{-1}_2}
\nwc{\bU}{{\mb U}}
\nwc{\bG}{{\mb G}}
\nwc{\bg}{\mathbf{g}}
\nwc{\mbf}{\mathbf{f}}
\nwc{\mbe}{\mathbf{e}}
\nwc{\be}{\mathbf{e}}
\nwc{\ind}{\operatorname{I}}
\nwc{\mbx}{\mathbf{f}}
\nwc{\bb}{\mathbf{g}}
\nwc{\xmax}{f_{\rm max}}
\nwc{\xmin}{f_{\rm min}}
\nwc{\suppx}{\hbox{\rm supp} (\mbf)}
\nwc{\cI}{\IZ^2_N}
\nwc{\chis}{{\chi^{\rm s}}}
\nwc{\chii}{{\chi^{\rm i}}}
\nwc{\pdfi}{{f^{\rm i}}}
\nwc{\pdfs}{{f^{\rm s}}}
\nwc{\pdfii}{{f_1^{\rm i}}}
\nwc{\pdfsi}{{f_1^{\rm s}}}
\nwc{\thetatil}{{\tilde\theta}}
\nwc{\red}{\color{red}}
\nwc{\blue}{\color{blue}}

\nwc{\prox}{\hbox{prox}}
\nwc{\diag}{\hbox{\rm diag}}
\nwc{\supp}{{\hbox{\rm supp}}}

\nwc{\sloc}{J_{\rm f}}
\nwc{\bu}{{\mb u}}
\nwc{\bv}{{\mb v}}
\nwc{\cU}{\mathcal{U}}
\nwc{\cN}{\mathcal{N}}
\nwc{\bN}{\mathbf{N}}
\nwc{\mbm}{\mathbf{m}}
\nwc{\bw}{\mathbf{w}}
\nwc{\bom}{\mathbf{w}}
\nwc{\bt}{\mathbf{t}}
\nwc{\z}{y}
\nwc{\cY}{\mathcal{Y}}
\nwc{\bM}{\mathbf{M}}
\nwc{\half}{{1\over 2}}
\nwc{\Sf}{S_{\rm f}}
\nwc{\Jf}{J_{\rm f}}
\nwc{\nul}{\hbox{\rm null}_\IR}
\nwc{\spanR}{\hbox{\rm span}_\IR}
\nwc{\Arg}{\hbox{\rm Arg~}}
\nwc{\fdr}{S_{\rm f}}
\nwc{\phase}[1]{\exp\lt[i\measured #1\rt]}

\nwc{\im}{{\rm i}}

\usepackage[utf8]{inputenc}
\usepackage{pgfplots}

\begin{document}

%\centerline{\tt \blue arXiv:2208.04798}

 \title{
 3D Tomographic Phase  Retrieval and Unwrapping}

\author{Albert Fannjiang 
 \address{
Department of Mathematics, University of California, Davis, California  95616, USA. Email:  {\tt fannjiang@math.ucdavis.edu}
} 
}

\maketitle 

\begin{abstract} {  This paper develops uniqueness theory for 
3D  phase retrieval with finite, discrete measurement data for {strong phase objects} and { weak phase objects}, including:

(i) {\em Unique determination of (phase) projections from diffraction patterns} -- General measurement schemes with coded and uncoded apertures are proposed and shown to ensure unique reduction of diffraction patterns to the phase projection for a strong phase object (respectively, the projection for a weak phase object) in each direction separately without the knowledge of relative orientations and locations.    (ii) {\em Uniqueness for  3D phase unwrapping} -- General conditions for unique determination of a 3D strong phase object from its phase projection data are established, including, but not limited to, random tilt  schemes densely sampled from a spherical triangle of vertexes  in three orthogonal directions  and other deterministic tilt schemes. (iii) {\em Uniqueness for projection tomography} --   Unique determination of an object of $n^3$ voxels  from generic $n$ projections or $n+1$ coded diffraction patterns is proved. 

This approach of reducing 3D phase retrieval to the problem of (phase) projection tomography has the practical implication  of  enabling classification and alignment, when relative orientations are unknown,  to be carried out in terms of (phase) projections, instead of diffraction patterns.

The applications with the measurement schemes such as single-axis tilt, conical tilt, dual-axis tilt, random conical tilt and general random tilt are discussed. 
}

 \end{abstract}

%\begin{keywords} Ptychography, phase retrieval, affine phase ambiguity, uniqueness.
%\end{keywords}
% \begin{AMS}49K35, 05C70,  90C08\end{AMS}

\section{Introduction}

\commentout{%08/03/2023 before edited by ChatGPT
Diffraction plays a central role in structure determination by high resolution X-ray and electron microscopies due to high sensitivity of phase contrast mechanism \cite{Henderson95, fel,3DED19,femto-pulse15}. Compared to the real-space imaging with lenses such as transmission electron microcopy, lensless diffraction methods are  aberration free and have the potential for delivering equivalent resolution with fewer photons/electrons \cite{X-noise09,ED18}.

While single crystal X-ray diffraction remains as the most widely used technique for 3D structure determination, 
 limited crystallinity of many materials often makes it challenging  to obtain sufficiently large and well-ordered crystals for X-ray diffraction  \cite{min-crystal}. This obstacle motivates the development of coherent diffractive imaging of non-periodic structures. 
 
X-ray and electron diffractions for non-periodic objects  can be realized in  two imaging modalities:
 diffraction tomography and single-particle imaging/reconstruction (Figure \ref{fig0}).  The former employs a  sizable object capable of sustaining illuminations from multiple  directions while the latter deals with multiple copies of the particle, such as a biomolecule, in various orientations \cite{serial-ED20,ED-support19,micro-ED19, serial-X11,XFEL17,SPI18}. The two modalities are mathematically equivalent, except that in single-particle reconstruction there are various levels of uncertainty in the relative orientations and locations between the object and the measurement set-up depending on the sample delivery methods
  \cite{serial15,chip12,sponge13,micro-X14, fixed17,sample-delivery18}. As a result, single-particle imaging is sometimes referred to
  as {\em crypto-tomography} \cite{Elser09}. 

Due to the extremely short wavelengths of X-ray and electron waves, only intensity measurement data can be collected. To emphasize this aspect, we collectively refer to the two imaging modalities with X-rays and electrons as {\em 3D  phase retrieval}.

Phase retrieval is the process of estimating the phase of a wave from intensity measurements. As phase information is not directly measurable in most imaging systems, it must be inferred from intensity measurements.
\commentout{ which are typically squared magnitudes of complex-valued wavefronts. These measurements capture the intensity but lose the phase information of the wave, which often contains essential details about an object's internal structure or the refractive index variations in optical tomography. To retrieve this lost phase information, various algorithms, including iterative or optimization-based methods, are used. These algorithms often involve alternating between enforcing constraints in the spatial and Fourier domains, or they might utilize prior knowledge about the object.
} On the other hand, phase unwrapping is a process required when the phase of a wave is already known but is 'wrapped' due to its cyclical nature. The phase of a wave is typically measured modulo $2\pi$, which means that it repeats every $2\pi$ and only values between $-\pi$  and $\pi$ or between 0 and $2\pi$, are recorded. When the actual phase exceeds this range, it 'wraps' around, leading to ambiguities in the phase data. Phase unwrapping involves resolving these ambiguities to recover the true phase map. 
\commentout{Various algorithms, ranging from simple path-following to more advanced statistical or regularization-based methods, are used for phase unwrapping.

In summary, phase retrieval aims to estimate the phase information from intensity measurements, while phase unwrapping resolves the ambiguities in the measured phase values due to the inherent 'wrapping' of phase values within a specific range. Both processes are essential for recovering detailed, high-resolution images in applications like diffraction tomography, optical interferometry, and more.
}

Our objective  is a theory of uniqueness for 3D phase retrieval with finite, discrete measurement data for both strong phase and weak phase objects. To this end,  we propose pairwise diffraction measurement schemes and analyze the conditions for unique determination of  the phase projection for a strong phase object (respectively, the projection for a weak phase object) in each direction; In the case of a  strong phase object,  the rendered phase projection data  contain only the wrapped phase information and  we propose a framework and tilt schemes for the resulting 3D phase unwrapping problem; 
In the case of a weak phase object, we analyze the resulting problem of projection
tomography and derive explicit conditions for unique determination of the object  of $n^3$ voxels from $n$ projection data or $n+1$ coded diffraction patterns.
}

Diffraction is crucial in structure determination via high-resolution X-ray and electron microscopies due to the high sensitivity of the phase contrast mechanism  \cite{Henderson95, fel,3DED19,femto-pulse15}. Compared to real-space imaging with lenses, like that in transmission electron microscopy, lensless diffraction methods are aberration-free and have the potential to deliver equivalent resolution using fewer photons/electrons  \cite{X-noise09,ED18}.

Although single crystal X-ray diffraction is the most commonly used technique for 3D structure determination, the limited crystallinity of many materials often makes obtaining sufficiently large and well-ordered crystals for X-ray diffraction challenging   \cite{min-crystal}. This obstacle has inspired the development of coherent diffractive imaging for non-periodic structures.

X-ray and electron diffractions for non-periodic objects can be realized in two imaging modalities: diffraction tomography and single-particle imaging/reconstruction  (Figure \ref{fig0}).  The former involves a sizable object capable of enduring illuminations from various directions, while the latter handles multiple copies of a particle, such as a biomolecule, in different orientations  \cite{serial-ED20,ED-support19,micro-ED19, serial-X11,XFEL17,SPI18}. These two modalities are mathematically equivalent, except that in single-particle reconstruction, the uncertainty levels vary concerning the relative orientations and locations between the object and the measurement set-up and depend on the sample delivery methods
  \cite{serial15,chip12,sponge13,micro-X14, fixed17,sample-delivery18}. 

Since the wavelengths of X-ray and electron waves are extremely short, only intensity measurement data can be collected. To emphasize this aspect we refer to the two imaging modalities  as {\em 3D phase retrieval}.

Phase retrieval is the process of estimating the phase of a wave from intensity measurements because phase information cannot be directly measured in most imaging systems. In contrast, phase unwrapping is necessary when the phase of a wave is already known but is 'wrapped' due to its cyclical nature. The phase of a wave, typically measured modulo $2\pi$, repeats every $2\pi$ with values recorded between $-\pi$  and $\pi$, or between 0 and $2\pi$. When the actual phase exceeds this range, it 'wraps' around, creating ambiguities in the phase data. Phase unwrapping resolves these ambiguities to recover the true phase map.

Our goal is to develop a theory of uniqueness for 3D phase retrieval with finite, discrete measurement data for both strong phase and weak phase objects. To accomplish this, we propose pairwise diffraction measurement schemes and analyze the conditions necessary for the unique determination of the phase projection for a strong phase object, and the projection for a weak phase object in each direction. For a strong phase object, the provided phase projection data contain only the wrapped phase information, so we propose a framework and tilt schemes to address the resulting 3D phase unwrapping problem. For a weak phase object, we analyze the resulting problem of projection tomography and derive explicit conditions for the unique determination of the object of $n^3$ voxels from $n$ projection data or $n+1$ coded diffraction patterns.

\begin{figure}[t]
\centering
{\includegraphics[width=12cm]{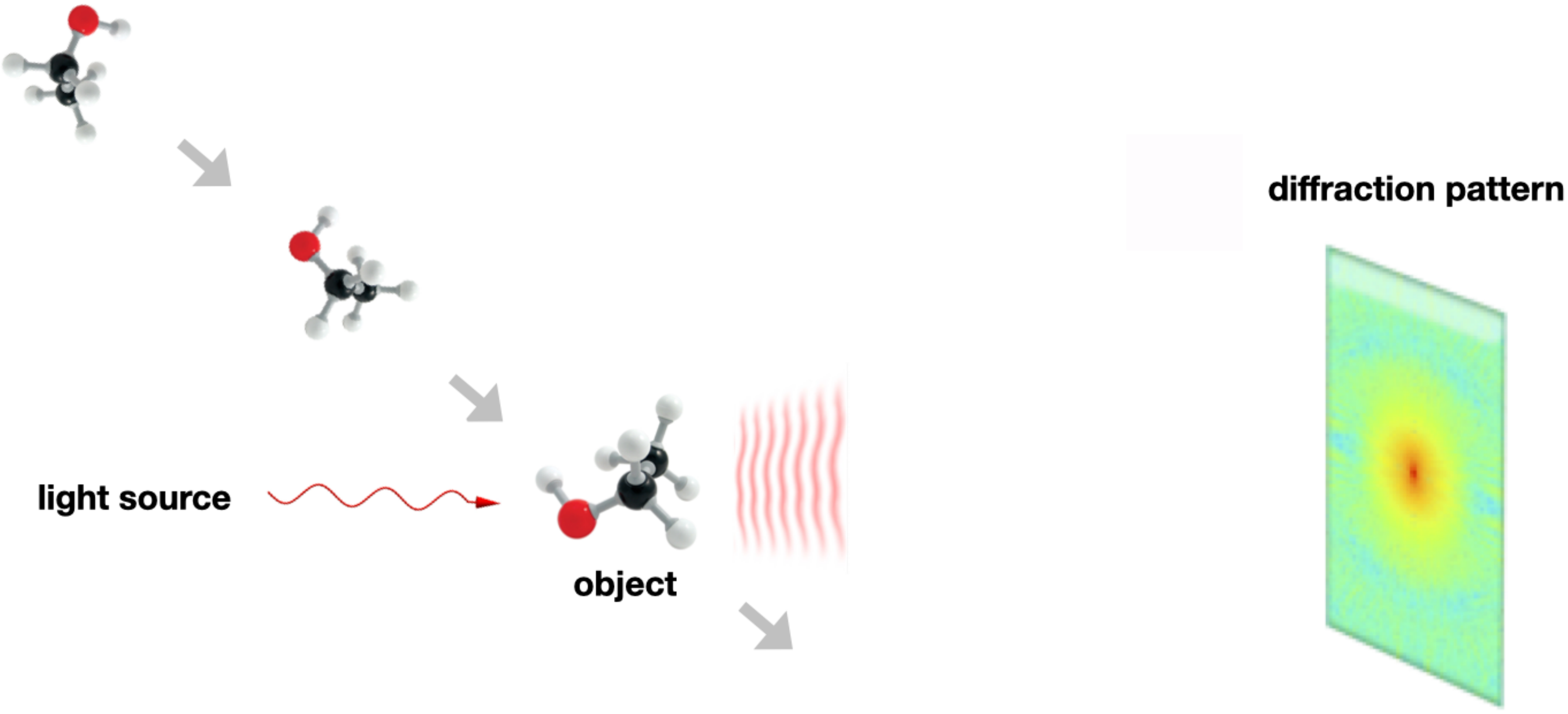}}
%{\includegraphics[width=14cm]{exp}}
\caption{Serial crystallography: A stream of identical particles  in various orientations scatter the incident wave with diffraction patterns measured in far field.}
\label{fig0}
\end{figure}

\subsection{Forward model}\label{sec:model}

Let  $n(\br)\in \IC$ denote the complex refractive index at the point $\br\in \IR^3$. 
The real and imaginary components of $n(\br)$ describe the dispersive and absorptive aspects of the wave-matter interaction. The real part is related to electron density in the case of X-rays and Coulomb field in the case of electron waves. 

Suppose that  $z$ is  the optical axis in which the incident plane wave $e^{\im \kappa z}$ propagates. 
 For a quasi-monochromatic wave field $u$ such as coherent X-rays and electron waves, it is useful to write  $u=e^{\im\kappa z} v$ to factor out  the incident wave and focus on the modulation field (i.e. the envelope), described by $v$.
 
The modulation field $v$ satisfies 
 the paraxial wave equation   \cite{Paganin}
\beq
\label{para}
\im \kappa {\partial\over \partial z}v+\half \Delta_\perp v+\kappa^2 f v=0,\quad f:=(n^2-1)/2
\eeq
where $\Delta_\perp=\nabla_\perp^2,  \nabla_\perp=(\partial_x,\partial_y)$,  derived from the fundamental wave equation by the so called small-angle approximation (hence the term ``paraxial wave")  which requires the wavelength $\lambda$ to be smaller than the maximal distance $d$  over which the fractional variation of $f$ is negligible   \cite{strong-phase}. 

In view of  different scaling regimes involved in the set-up  (Figure \ref{fig0} and \ref{fig1}),  we  now break  up the forward model into two components: First, a large Fresnel number regime from the entrance pupil to  the exit plane; Second, a small Fresnel number regime  from  the exit plane  to the detector  plane.

 For the exit wave, consider the large Fresnel number regime
 \beq
 \label{lim2}
 N_{\rm F}={d^2\over \lambda \ell}\gg 1
 \eeq
 where $\ell$ is the linear size of the object. 
 By  rescaling the coordinates 
\beqn
%\label{scale}
z \longrightarrow \ell z,\quad (x,y)\longrightarrow d (x,y)
\eeqn
we  non-dimensionalize \eqref{para} as 
\beqn
\im {\partial\over \partial z}v+{1\over 4\pi} N^{-1}_{\rm F}\Delta_\perp v+\kappa \ell f v=0, %\label{Fresnel}
\eeqn
which has a diminishing diffraction term $\Delta_\perp$ under \eqref{lim2}.

After dropping the $\Delta_\perp$ term, the reduced equation  in terms of the original coordinates before rescaling  is
\[
\im {\partial\over \partial z}v+\kappa f v=0,%\quad f:=(n^2-1)/2
\]
which can be solved by  integrating along the optical axis as
 \beq
v(\br)&=&e^{\im\kappa\psi(\br)} \label{rytov},\\
\label{wrap}\psi(x,y,z)&=&\int_{-\infty}^zf(x,y,z')dz'.  %{1\over 2} \int_{-\infty}^z (n^2(x,y,z')-1)dz'=\int_{-\infty}^zf(x,y,z')dz',\quad f:=(n^2-1)/2. 
\eeq
Alternatively, \eqref{rytov}-\eqref{wrap} can be derived by  stationary phase analysis  \cite{strong-phase} or
the {high-frequency  Rytov} approximation \cite{CT}.

The exit wave is given by $u=e^{\im\kappa z} v$  evaluated at the object's rear boundary (say,  $z=0$). At $z=+\infty$,  \eqref{wrap} is called  the {\em ray transform}, or simply the {\em projection},  of the object 
$f$ in the $z$ direction and  \eqref{rytov} will be called the {\em phase projection}  which equals the exit wave, up to a constant phase factor  \cite{Natterer}.

By allowing significant phase fluctuations with arbitrary $\kappa |\psi|$,  \eqref{rytov}-\eqref{wrap} is an  improvement over the {\em weak-phase-object approximation}  
\beq
\label{Born}
e^{\im\kappa\psi}&\approx &1+\im\kappa\psi 
\eeq
often used in cryo-electron microscopy (cryo-EM)  \cite{Frank06, Glaeser,strong-phase10}.  Following the nomenclature in \cite{Glaeser} and \cite{Spence}, 
we call  \eqref{rytov}-\eqref{wrap} the {\em strong-phase-object} approximation, noting, however, that
$f$ is a complex-valued function in general (thus $e^{\im \kappa\psi}$  not a ``phase" object {\em per se}). 
The strong-phase-object approximation is often invoked the high-frequency forward scattering problems such  as diffraction tomography and X-ray diffractive imaging (see, e.g. \cite{Wolf01,ptycho-tomo10, strong-phase11, ptycho-tomo17,Xray}). In view of
\eqref{Born},  the weak-phase-object approximation is the first order Born approximation of the strong-phase-object approximation. The applicability and accuracy of the Born and Rytov approximations  has been well studied \cite{Keller, Rytov89, Rytov92, Rytov98,RG}.

However,  the exit wave described by the phase projection \eqref{rytov} yields only the information of the projection of $f$ modulo $2\pi/\kappa$, and therein arises  the
problem of {\em phase unwrapping}, which is fundamentally unsolvable unless additional prior information is known (see Section \ref{sec:unwrap}) and poses a major road block to the implementation 
of  diffraction tomography.
The solution for phase unwrapping is critical  in revealing the depth dimension of the object. In contrast, phase unwrapping problem is not present in computed tomography \cite{CT}, which neglects diffraction, or , cryo-EM which operates under  the weak-phase-object approximation \eqref{Born}. 

{Between the strong- and weak-phase-object approximations, there is a family of hybrid approximations which take the form
\beq\label{hybrid}
v_q= \Big(1+{\im\kappa\psi\over q}\Big)^q,\quad q\in \IR, 
\eeq
\cite{Lu85,Born-Rytov}. Clearly, the weak-phase-object approximation corresponds to $q=1$ while $v_q$ tends to the phase projection \eqref{rytov}  as $q\to\infty.$ In the hybrid approximation, the complex phase $\psi$ is given by
\beqn
\psi_q={q\over\im \kappa} \lt(v_q^{1/q}-1\rt) 
\eeqn
which, for  $q>1$, has multiple values due to the $q$-th root of a complex number. This leads to a phase unwrapping problem similar to that for the strong-phase-object approximation.  The hybrid approximation with an intermediate value of $q$  can be used to incorporate the different features of the Born and Rytov models. Although not a focus of the present work, we will further elaborate the phase unwrapping problem associated with \eqref{hybrid}  in Section \ref{sec:hybrid}. }

\begin{figure}
\centering
{\includegraphics[width=14cm]{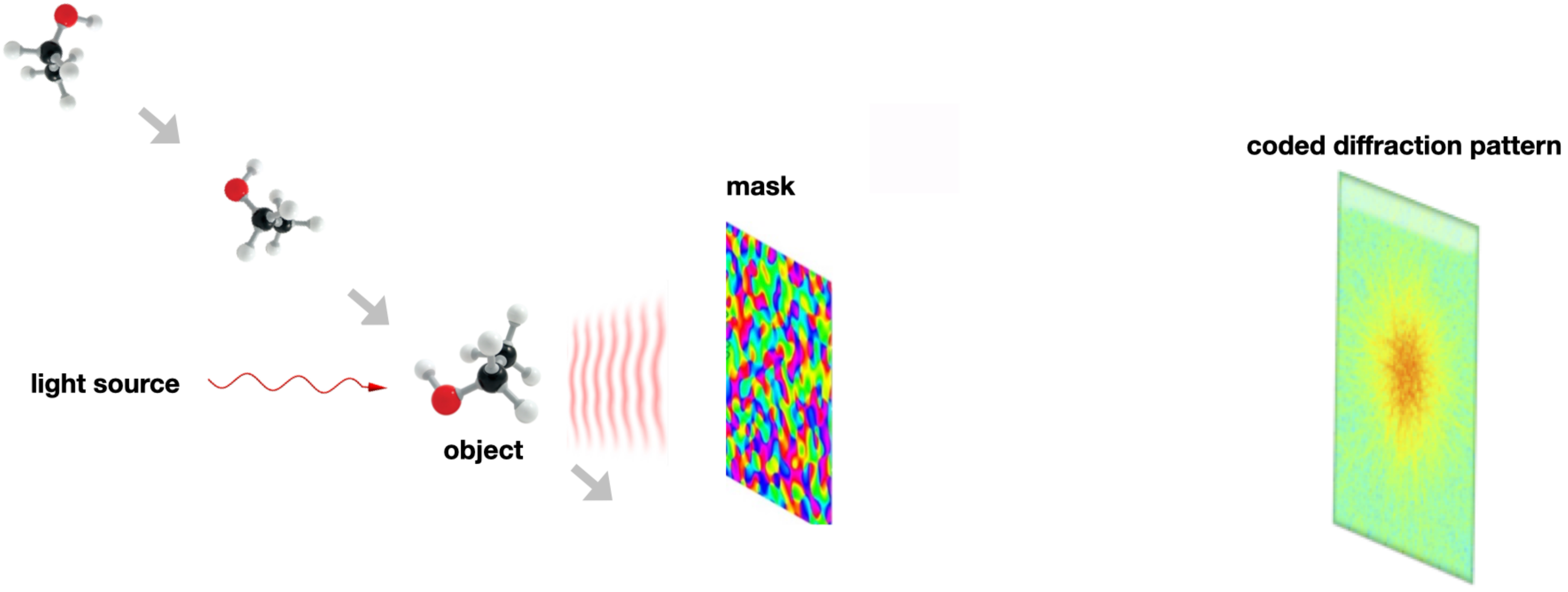}}
\caption{ Serial crystallography with a coded aperture}
\label{fig1}
\end{figure}

After passing through the object and the mask $\mu$ immediately behind, the exit wave  \eqref{rytov}-\eqref{wrap}
becomes  the masked exit wave $\mu e^{\im  \kappa \psi}$ at the exit plane $z=0$ and then undergoes the free space propagation (with $n=1, f=0$) for $z>0$ described by
\beqn
%\label{para2}
\im {\partial\over \partial z}v+{1\over 2\kappa} \Delta_\perp v=0,\quad v(x,y,0)=\mu e^{\im  \kappa \psi}.
\eeqn
The solution is given  by convolution with the Fresnel kernel as
\beqn
\nn v(x,y,z)&=&{1\over \im \lambda z}\int_{\IR^2} e^{{\im\kappa\over 2z} (|x-x'|^2+|y-y'|^2)}\mu(x',y') e^{\im  \kappa \psi(x',y',0)}dx'dy'
\eeqn
and hence,  after writing out the quadratic phase term, 
\beq
u(x,y,z)&=& {e^{\im \kappa z}\over \im \lambda z}e^{{\im\kappa \over 2z} (x^2+y^2)} \int_{\IR^2} e^{-{\im\kappa\over z} (xx'+yy')}e^{{\im\kappa\over 2z} (|x'|^2+|y'|^2)}\mu(x',y') e^{\im  \kappa \psi(x',y',0)}dx'dy'.\label{Fresnel}
\eeq

Let the detector plane $z=L$ to be sufficiently far away from the exit plane $z=0$ so that  the Fresnel number is small:
\beq
\label{lim3} N_F={\ell^2\over\lambda L} \ll 1. 
\eeq  
Then the second integrand (the quadratic phase factor)  in \eqref{Fresnel} is approximately 1  because  the integration  is carried out in the support of $\mu$ which is taken to be a square of size $\ell$ around the origin. On the other hand, the first integrand in \eqref{Fresnel} (the cross phase factor)  has the effect of the Fourier transform $\cF$  if the coordinates are properly rescaled to reflect the fact that the detector area is usually much larger than $\ell^2$. 

Since only the intensities of $u$ are measured, 
the phase factors  $e^{\im \kappa z}e^{{\im\kappa \over 2z} (x^2+y^2)}$ in \eqref{Fresnel} drop out and
 \beq
|u(x,y,L)|^2=|v(x,y,L)|^2\sim |\cF[\mu e^{\im  \kappa \psi} ]|^2,\label{data}
\eeq
up to  a scale factor  $(\lambda L)^{-1}$, which can be neglected in our analysis. 

 Depending on the context we shall refer to either a diffraction pattern or a projection as a  ``snapshot".

\subsection{Contributions}\label{outline1}

This paper presents a {\em discrete} framework for analyzing {\em discrete, finite} measurement data analogous to \eqref{data} and develops a uniqueness theory for 3D phase retrieval and unwrapping. 

Our main contribution in this paper is as follows:
\begin{itemize}
\item[\bf  1) (Phase) projection recovery.] 
We introduce pairwise measurement schemes with both coded and un-coded apertures and derive precise conditions
for unique determination of  (i) the phase  projection  for strong phase objects (Theorem  \ref{thm1} \& \ref{thm2}) and (ii) the projection, up to a phase factor, for  weak phase objects (Theorem \ref{thm3} \& \ref{thm4}). 

\item[\bf 2) Phase unwrapping.]  We propose a framework for analyzing 
the phase unwrapping problem when the given data are phase projections and derive generic conditions for unique phase unwrapping (Theorem \ref{tom2}).
We provide explicit  tilt schemes for phase unwrapping, including, but not limited to, random tilt schemes densely sampled from a spherical triangle with vertexes  in three orthogonal directions   (Section  \ref{sec:random-tilt})
and other deterministic tilt schemes (Section \ref{sec:deterministic-tilt})

\item[\bf 3) Projection tomography.] We show that any generic set of $n$ projections or generic set of $n+1$ coded diffraction patterns uniquely determine the object (Theorem  \ref{tom2-weak} \& \ref{thm:born}).

\end{itemize}
Our numerical simulation shows that randomly initialized alternating projection algorithm with random tilt schemes 
can robustly reconstruct 3D objects at high noise levels   (Section \ref{sec:num}).

\subsection{Outline}\label{outline2}
{The rest of the paper can be outlined  as follows.

In Section \ref{sec:discrete}, we set up a discrete framework   for tomography which is needed for the uniqueness question with finite, discrete measurements.

In Section \ref{sec4},  we prove  that with pairwise measurement schemes
  \beq
  %|\cF(\mu\odot e^{\im\kappa  g^*_\bt})|&=&|\cF(\mu\odot e^{\im\kappa  f^*_\bt})|,\quad\forall\bt\in \cT;\label{PR}\\
 e^{\im \kappa  g_\bt}&=&e^{\im \kappa  f_\bt}\quad  \hbox{in the case of  strong phase objects} \label{1.18'}
 \eeq
 and, for a constant $\theta_0\in \IR$,
 \beq\label{PT}
 g_\bt &=&e^{\im \theta_0} f_\bt \quad  \hbox{ in the case of  weak phase objects} 
 \eeq
 if $f, g\in O_n$  produce the same diffraction patterns   for $\bt\in \cT$, thus reducing 3D phase retrieval to the problem of (phase) projection tomography.  In Section \ref{sec:schemes},
we introduce pairwise measurement schemes for single-particle imaging where each particle is destroyed after one illumination.
 
In  Section \ref{sec:unwrap}  we find conditions on $\cT$  that ensure $g=f$  if \eqref{1.18'} holds. This is a uniqueness theory for phase unwrapping
and in Section \ref{sec:tilt} we propose various concrete measurement schemes that conveniently realize the uniqueness conditions.

In Section \ref{sec:weak}, we find  conditions on $\cT$  that ensure $g=e^{\im\theta_0} f$ for some constant $\theta_0\in \IR$  independent of $\bt$ if  \eqref{PT} holds. This is a unique theory for projection tomography. 

In Section \ref{sec:num}, we show numerically that with the proposed random tilt measurement scheme  alternating projection algorithm can robustly reconstruct a 3D weak-phase object from noisy data of a similar count to that required by the uniqueness theorem. 

We discuss  the design of 3D tomographic phase unwrapping algorithms  and  applications to single-particle imaging with unknown orientations in Section \ref{sec:final}.
}

\section{Discrete set-up}\label{sec:discrete}
Imagine an object defined on a cube of size $\ell$ in $\IR^3$. If we want to discretize the object, what would be a proper grid spacing? Obviously, the finer the grid  the higher the fidelity of the discretization. The grid system, however,  would be unnecessarily large if the grid spacing is  much smaller  than the resolution length which is the smallest feature size resolvable by the imaging system. 

In a diffraction-limited imaging system such as X-ray crystallography  the resolution length is roughly $\lambda/2$. In a radiation-dose-limited system such as electron diffraction, the resolution length can be considerably larger than $\lambda/2$.

Now suppose we choose $\lambda/2$ to be the grid spacing (the voxel size). 
For this grid system to represent accurately the object continuum, it is necessary that the grid spacing is equal to or smaller than the maximal distance $d$ over which the fractional variation of the object is negligible. On the other hand,
the underlying assumption for the paraxial wave equation \eqref{para} is exactly $\lambda \le d$ \cite{strong-phase}.
Hence  the fractional variation of the object  within a voxel  as well as  between adjacent grid points is negligible, two desirable properties of a grid system.
In other words, the grid system with spacing $\lambda/2$ gives an accurate representation of the object continuum under the assumption of the paraxial wave equation without resulting in  unnecessary complexity.  

We will, however,  let the discrete object to take independent, arbitrary value in each voxel,
except for Section \ref{sec:unwrap} where  we assume  the so called Itoh condition that the difference in the object between two adjacent grid points is less than $\pi/\kappa$ in order to obtain uniqueness for phase unwrapping.

Let $\lambda/2$  be the unit of length and
the grids  (the voxels) be labelled by  integer triplets  $(i,j,k)$. 
 In  this length unit, the wavenumber $\kappa$ has the value $\pi$ (i.e.  $\pi/\kappa=1$). The number $n$  of grid points in each dimension is about $2\ell/\lambda$ which may be  large for a strong phase object.

Let $\lb k,l\rb$ denote the integers
between and including the integers $k$ and $l$. 
Let  $O_n$ denote the class of discrete complex-valued objects
\beq
\label{1.1}
O_n:=\{f: f(i,j,k)\in \IC,  i,j,k\in \IZ_n\}
 \eeq
 where
\beq
\label{1.2}
\IZ_n&=&\left\{\begin{array}{lll}
\lb-n/2, n/2-1\rb && \mbox{if $n$ is an even integer;}\\
 \lb-(n-1)/2, (n-1)/2\rb && \mbox{if $n$ is an odd integer.}
 \end{array}\rt.
\eeq
To fix the idea, we consider the case of odd $n$ in the paper.

Following  the framework in \cite{discrete-X} we discretize the projection geometry given in Section \ref{sec:model}. 

%our main goal in discretizing the projection geometry is to establish the discrete version of the Fourier slice theorem (Proposition \ref{thm:slice}) and the common line property. 

We define three families of line segments, the $x$-lines, $y$-lines, and $z$-lines. 
The $x$-lines, denoted by $\ell_{(1,\alpha,\beta)}(c_1,c_2)$ with $ |\alpha|, |\beta|<1$,  are  defined by
\beq
\label{1.3'}
\ell_{(1,\alpha,\beta)}(c_1,c_2): \lt[\begin{matrix}
y\\
z\end{matrix}\rt]=\lt[\begin{matrix}\alpha x+c_1\\  \beta x+c_2
\end{matrix}\rt] && c_1, c_2\in \IZ_{2n-1}, \quad x\in \IZ_n%\quad  |\alpha|,\,\, |\beta|\le  1
\eeq
 To avoid wraparound of $x$-lines with , we can zero-pad $f$ in a larger lattice $\IZ^3_p$ with $p\ge 2n-1.$ This is particularly important when it comes to define the ray transform by a line sum (cf. \eqref{2.8}-\eqref{2.10}) since wrap-around is unphysical.

Similarly, a $y$-line and a $z$-line are defined as
\beq
\ell_{(\alpha,1,\beta)}(c_1,c_2):  \lt[\begin{matrix}
x\\
z\end{matrix}\rt]=\lt[\begin{matrix}\alpha y+c_1\\  \beta y+c_2
\end{matrix}\rt]
 && c_1, c_2\in \IZ_{2n-1}, \quad y\in \IZ_n,\\
\ell_{(\alpha,\beta,1)}(c_1,c_2): 
 \lt[\begin{matrix}
x\\
y\end{matrix}\rt]=\lt[\begin{matrix}\alpha z+c_1\\  \beta z+c_2
\end{matrix}\rt] && c_1, c_2\in \IZ_{2n-1}, \quad z\in \IZ_n, 
\eeq
with $ |\alpha|, |\beta|<1$.

Let $\widetilde f$ be the continuous interpolation of $f$
given by
\beq\label{1.5}
  \widetilde f(x,y,z)&=&\sum_{i\in \IZ_n} \sum_{j\in \IZ_n}\sum_{k\in \IZ_n} f(i,j,k) D_p(x-i)D_p(y-j)D_p(z-k),\quad x,y,z,\in \IR,
%&& \nn \forall  x, y,z\in [-(n-1)/2, (n-1)/2]^3 %\hbox{\rm Hull} (\IZ_n^3)% |y|,|z|<(p-1)/2
\eeq
where %$\hbox{\rm Hull} (\IZ_n^3)$ is the convex hull of  $\IZ_n^3$ and
 $D_p$ is the $p$-periodic Dirichlet kernel given by
\beq
D_p(t)={1\over p} \sum_{l\in \IZ_{p}} e^{\im 2\pi l t/p}
&=&\lt\{\begin{matrix}1,& t=mp,\quad m\in \IZ\\
{\sin{(\pi t)}\over p\sin{(\pi t/p)}},& \mbox{else}.% |t|<(p-1)/2\\
\end{matrix}\rt.\label{Dir}
\eeq
In particular, $[D_p(i-j)]_{i,j\in \IZ_p}$ is the $p\times p$ identity matrix. 
Because $D_p$ is a continuous $p$-periodic function, so is $\widetilde f$. 
However, we will only use the restriction of $\widetilde f$ to one period cell  $[-(p-1)/2, (p-1)/2]^3$ to define the discrete projections and  avoid the wraparound effect. 

We define the discrete projections as the following line sums%$P_{(1,\alpha,\beta)} f(c_1,c_2)$ as
\beq
 \label{2.8} %P_{(1,\alpha,\beta)}f(c_1,c_2)
 f_{(1,\alpha,\beta)}(c_1,c_2)&=& \sum _{i\in \IZ_n} \widetilde f(i,\alpha i+c_1,\beta i+c_2),\\%\quad |\alpha|\le 1,|\beta|\le 1, \quad  c_1, c_2\in\lb-n, n\rb.
%P_{(\alpha,1,\beta)} f(c_1,c_2)
f_{(\alpha,1,\beta)}(c_1,c_2)&=& \sum _{j\in\IZ_n} \widetilde f (\alpha j+c_1,j, \beta j+c_2)\label{2.9}\\
%P_{(\alpha,\beta,1)} f(c_1,c_2)
f_{(\alpha,\beta,1)}(c_1,c_2)&=& \sum _{k\in \IZ_n} \widetilde f(\alpha k+c_1,\beta k+c_2, k)\label{2.10}
\eeq
with $ c_1, c_2\in \IZ_{2n-1}$.

The 3D discrete Fourier transform $\widehat f$ of the object $f\in O_n$,  is given by
\beq
\widehat f(\xi,\eta,\zeta)&=&\sum_{i,j,k\in \IZ_n} f(i,j,k)e^{-\im 2\pi(\xi i+\eta j+ \zeta k)/p}= \sum_{i,j,k\in \IZ_p} f(i,j,k)e^{-\im 2\pi(\xi i+\eta j+ \zeta k)/p}\label{1.10'}
\eeq
where
the range of the Fourier variables $\xi,\eta,\zeta$ can be extended from the discrete interval $\IZ_p$ to the continuum  $[-(p-1)/2, (p-1)/2]$. Note that by definition, $\widehat f$ is a $p$-periodic band-limited function. 
The associated 1D and 2D (partial) Fourier transforms are similarly defined $p$-periodic band-limited functions.

\subsection{Fourier slices and common lines}

The Fourier slice theorem concerns the 2D discrete Fourier transform $\widehat f_{(1,\alpha,\beta)}$ defined as
\beq
\widehat f_{(1,\alpha,\beta)}(\eta,\zeta)&=& \sum_{j,k\in \IZ_{2n-1}} f_{(1,\alpha,\beta)}(j,k)e^{-\im 2\pi(\eta j+ \zeta k)/p},
\eeq
 and the 3D discrete Fourier transform given in \eqref{1.10'}. 

It is straightforward, albeit somewhat tedious, to derive the discrete Fourier slice theorem which plays an important role in our analysis.

\begin{prop}\label{thm:slice}\cite{discrete-X}
(Discrete Fourier slice theorem) For a given family of $x$-lines $\ell_{(1,\alpha,\beta)}$ with fixed slopes $(\alpha,\beta)$ and variable intercepts $(c_1,c_2)$. Then the 2D discrete Fourier transform  $\widehat f_{(1,\alpha,\beta)}$ of the $x$-projection $f_{(1,\alpha,\beta)},$ given in \eqref{2.8},   and the 3D discrete Fourier transform $\widehat f$ of the object $f$ satisfy the equation
\beq\label{proj-x}
\widehat f_{(1,\alpha,\beta)}(\eta,\zeta)&=& \widehat f(-\alpha \eta-\beta \zeta, \eta,\zeta),\quad \eta,\zeta\in \IZ.
\eeq
Likewise, we have
\beq
\widehat f_{(\alpha,1,\beta)}(\xi,\zeta)&=& \widehat f(\xi, -\alpha \xi-\beta \zeta, \zeta),\quad  \xi,\zeta\in \IZ;\\
\widehat f_{(\alpha,\beta,1)}(\xi,\eta)&=& \widehat f(\xi,\eta,  -\alpha \xi-\beta \eta).\label{proj-z},\quad \xi, \eta\in \IZ. 
\eeq
\end{prop} 

\begin{rmk}\label{rmk:common}
For the general domain $\IR^2$, it is not hard to derive 
the following results 
\beq\label{proj-x'}
\widehat f_{(1,\alpha,\beta)}(\eta,\zeta)&=&\sum_{j,k\in \IZ_p} \widehat f(-\alpha j-\beta k, j,k)D_p(\eta-j)D_p(\zeta-k),\quad\eta,\zeta\in \IR;\\
\widehat f_{(\alpha,1,\beta)}(\xi,\zeta)&=&\sum_{j,k\in \IZ_p} \widehat f(\xi, -\alpha j-\beta k, \zeta)D_p(\xi-j)D_p(\zeta-k),\quad\xi,\zeta\in \IR\label{proj-y'};\\
\widehat f_{(\alpha,\beta,1)}(\xi,\eta)&=& \sum_{j,k\in \IZ_p} \widehat f(j,k,  -\alpha j-\beta k)D_p(\xi-j)D_p(\eta-k), \quad\eta,\xi\in \IR,\label{proj-z'}
\eeq
in  the form of interpolation by the grids in the respective Fourier slices. 
%which may be referred to as the {\em generalized discrete Fourier slice theorem}. 
From  \eqref{Dir} it follows that the right hand side of \eqref{proj-x'}-\eqref{proj-z'} are Laurent polynomials of 2 trigonometric variables (e.g. $e^{\im 2\pi \eta/p}, e^{\im 2\pi \zeta/p}$ for \eqref{proj-x'}),
 and that  
\eqref{proj-x'}-\eqref{proj-z'} reduce to \eqref{proj-x}-\eqref{proj-z} when the trigonometric variables are integer powers of $e^{\im 2\pi/p}$.

Recalling the view of discretization espoused at the beginning of this section and returning  to the original scale in the continuous setting, we note that
\beq
\label{large}
\lim_{p\to\infty}pD_p(pt)=\delta(t),\quad t\in \IR,
\eeq
the Dirac delta function. By \eqref{large}  and rescaling the standard, continuous version of Fourier slice theorem is recovered from \eqref{proj-x'}-\eqref{proj-z'}. 
\end{rmk}

For ease of notation, we denote  by $\bt$ the direction of projection, $(1,\alpha,\beta), (\alpha,1,\beta)$ or $(\alpha,\beta,1)$ in the reference frame attached to the object. Let $P_\bt$ denote the origin-containing (continuous) plane orthogonal to $\bt$ in the Fourier space. The standard common line is defined by $L_{\bt,\bt'}:= P_{\bt}\cap P_{\bt'}$ for $\bt, \bt'$ not parallel to each other. 

By a slight abuse of notation, the common-line property implied by  Proposition \ref{thm:slice} can be succinctly  stated as 
\beq\label{common}
\widehat f_\bt(\bk)&=&\widehat f_{\bt'} (\bk'), \quad \bk \in P_{\bt}\cap\IZ^2,\quad \bk'\in P_{\bt'}\cap \IZ^2,
\eeq
where $\bk,\bk'$ are the corresponding coordinates and may differ from each other depending on the parametrization of $P_\bt$ and $P_{\bt'}$. 

For example, let $\bt\sim (\alpha,\beta,1)$ and $\bt'\sim (\alpha',\beta',1)$. The Fourier slices  are given by
\beq
\alpha\xi+\beta\eta+\zeta=0,\quad \alpha'\xi+\beta'\eta+\zeta=0,\label{same}
\eeq
with the correspondence  $\bk=\bk'=(\xi,\eta)\in \IZ^2$. 

For a different configuration, let $\bt\sim (\alpha,\beta,1)$ and $\bt'\sim (1,\beta',\gamma')$. The Fourier planes are given  by
\beq\label{diff}
\alpha\xi+\beta\eta+\zeta=0,\quad \xi+\beta'\eta+\gamma'\zeta=0,
\eeq
with  the correspondence  $\bk=(\xi,\eta)\in \IZ^2$ and $\bk'=(\eta,\zeta)\in \IZ^2$.

For non-integral points, however, the common lines are perturbed  by interpolation \eqref{proj-x'}-\eqref{proj-z'}. For  \eqref{same} and \eqref{diff}, the ``common lines" can be generalized  respectively as the traces of
the two-dimensional surfaces defined by 
\beq
\label{common1}
L_{\bt,\bt'}(f)&:=&\Big\{\widehat f_{(\alpha,\beta,1)}(\xi,\eta)=\widehat f_{(\alpha',\beta',1)}(\xi',\eta')\Big\}\subseteq \IR^2\times \IR^2\\
\label{common2}
L_{\bt,\bt'}(f)&:=& \Big\{\widehat f_{(\alpha,\beta,1)}(\xi,\eta)=\widehat f_{(1,\beta',\gamma')}(\eta',\zeta')\Big\} \subseteq \IR^2\times \IR^2.
\eeq
By the common line property \eqref{common},  $L_{\bt,\bt'}$ in \eqref{common1} contains  the set defined by $(\xi',\eta')=(\xi,\eta)\in \IZ^2$ while $L_{\bt,\bt'}$ in \eqref{common2} contains  the set defined by
\beq
\eta'=\eta,\quad \zeta'=-\alpha\xi-\beta\eta\in \IZ, \quad (1-\alpha\gamma')\xi=(\beta\gamma'-\beta')\eta,\quad (\xi,\eta)\in \IZ^2.%(1-\alpha\gamma')\xi=\eta(\beta\gamma'-\beta'),. 
\eeq
By \eqref{common}, $0\in L_{\bt,\bt'}(f)$ for any $\bt,\bt'$.  In view of \eqref{large}, the trace of $L_{\bt,\bt'}$ on either Fourier slice is
near $P_{\bt}\cap P_{\bt'}$ for sufficiently large $n$. 

We shall refer to $L_{\bt,\bt'}(f)$ as the {\em common set}  for the Fourier slices of $f$  orthogonal to $\bt,\bt'$. The notion of common sets will be used to formulate a technical condition for Theorem \ref{thm2}.

\subsection{Diffraction pattern}\label{sec:coded}

Let $\cT$ denote the set of directions $\bt$ employed in the 3D diffraction measurement, which can be coded (as in Figure \ref{fig1}) or uncoded (as in Figure \ref{fig0}). To fix the idea, {  let $p=2n-1$} in \eqref{Dir}. 

Let the Fourier transform $\cF$ of the projection $e^{\im\kappa f_\bt(\bn)}$ be written as 
\beq
\label{over-sampling}
F_\bt(e^{-\im 2\pi\bw})=\sum_{\bn\in \IZ_p^2} e^{-\im 2\pi \bn\cdot\bw} e^{\im\kappa f_\bt(\bn)},\quad \bw\in \Big[-\half,\half \Big]^2.
\eeq
 {In the absence of a random mask ($\mu\equiv 1$)}, the intensities of the Fourier transform  can be written as  \beq
|F_\bt(e^{-\im 2\pi\bw})|^2= \sum_{\bn\in   \IZ_{2p-1}^2}\lt\{\sum_{\bn'\in \IZ_p^2} e^{\im\kappa f_\bt(\bn'+\bn)}e^{-\im\kappa \overline{f_\bt(\bn')}}\rt\}
   e^{-\im 2\pi \bn\cdot \bom}, \quad \bom\in  \Big[-\half,\half \Big]^2,
   \label{auto}
   \eeq
which is called  the uncoded diffraction pattern in the direction $\bt$. 
 Here and below the over-line notation means
complex conjugacy. The expression in the brackets  in \eqref{auto} is the autocorrelation function of $e^{\im\kappa f_\bt}$.

The diffraction patterns are then uniquely determined by  sampling on the grid
\beq
\label{nyquist1}
\bw\in {1\over 2p-1} \IZ_{2p-1}^2
\eeq
or by Kadec's $1/4$-theorem  on any following  irregular grid \cite{Young}
\beq
\label{nyquist2}
\{\bw_{jk},\,\, j,k\in \IZ_{2p-1}: |(2p-1)\bw_{jk}-(j,k)|<1/4\}.  
\eeq
With the regular \eqref{nyquist1} or irregular \eqref{nyquist2} sampling, the diffraction pattern contains the same information as does  the autocorrelation function of $f_\bt$.

% \section{Reduction with strong phase objects} \label{sec:pair}
  \section{Pairwise reduction to (phase) projection}\label{sec4}
  
  We assume the following on the coded aperture
 
 \begin{quote}
{\bf Mask Assumption:} The mask function is given by $\mu(\bn)=\exp[\im \phi(\bn)]$ with independent, continuous random variables $\phi(\bn)\in \IR$. 
 
 \end{quote}

Let $f^*_\bt$ denote the projection $f_\bt$ translated by some $\bl_\bt \in \IZ^2$, i.e.
\beq
\label{loc}
f^*_\bt (\bn):= f_\bt (\bn+\bl_\bt),\quad   \quad\hbox{subject to}\,\,\supp ( f^*_\bt) \subseteq \IZ_n^2.
\eeq

We assume that each snapshot  is taken for $ f^*_\bt$ (not $f_\bt$) with a shift $\bl_\bt$ due to variability in sample delivery. 

In Section \ref{sec:pair} and \ref{sec:born}, we show for two different imaging geometries how a pair of diffraction patterns can uniquely determine
the respective (phase) projections. 
  
  \subsection{Strong phase objects} \label{sec:pair}

The following theorem says that  the two diffraction patterns, one coded and one uncoded,  uniquely determine the underlying phase projection.

\begin{thm}\label{thm1} Let $f, g\in O_n$ and assume the Mask Assumption.  
  Suppose that for any  $ \bt$ 
\beq
\label{43}
|\cF(e^{\im\kappa  g^*_{\bt}})|^2&=&|\cF(e^{\im\kappa  f^*_{\bt}})|^2\\
|\cF(\mu\odot e^{\im\kappa  g^*_\bt})|^2&=&|\cF(\mu\odot e^{\im\kappa  f^*_\bt})|^2\label{44}
\eeq
Then $
e^{\im\kappa  g_\bt}=e^{\im\kappa  f_\bt}$ almost surely. 
 \end{thm}
 The proof of Theorem \ref{thm1} is given in Appendix \ref{appA}. 
 
The next theorem says that  the two coded diffraction patterns in two different directions uniquely determine the two corresponding phase projections.

 \begin{thm}\label{thm2}  Let $f, g\in O_n$  such that for $\bt,\bt'$ not parallel to each other, the intersection 
$C_{\bt,\bt'}:=L_{\bt,\bt'}(f)\cap L_{\bt,\bt'}(g)$ contains some  $(\bk,\bk')\neq 0$ such that 
the slope of either $\bk$ or $\bk'$ is not a fraction over $\IZ_p$. Let the Mask Assumption hold true. 
 Suppose
 that 
\beq\label{43-2}
|\cF(\mu \odot e^{\im\kappa  g^*_{\bt'}})|^2&=&|\cF(\mu \odot e^{\im\kappa  f^*_{\bt'}})|^2\\
|\cF(\mu\odot e^{\im\kappa  g^*_\bt})|^2&=&|\cF(\mu\odot e^{\im\kappa  f^*_\bt})|^2.\label{44-2}
\eeq
Then $
e^{\im\kappa  g_\bt}=e^{\im\kappa  f_\bt}$ and $
e^{\im\kappa  g_{\bt'}}=e^{\im\kappa  f_{\bt'}}$ almost surely. 
 \end{thm}
 The proof of Theorem \ref{thm2} is given in Appendix \ref{appB}.

\begin{cor}\label{cor2}
If for each $\bt\in \cT$ there is a $\bt'\in \cT$ to satisfy Theorem \ref{thm2}, then  $
e^{\im\kappa  g_\bt}=e^{\im\kappa  f_\bt}$ for all $\bt\in \cT$. 
\end{cor}

 Note that Theorem \ref{thm1}, \ref{thm2} and Corollary \ref{cor2} do not hold
 for a uniform  mask ($\mu=$ cost.) because the chiral ambiguity and the shift ambiguity are present, i.e. both 
 $g(\cdot) =f(-\cdot)$ and $g(\cdot)=f(\cdot+\bl),\bl\in \IR^3, $ satisfy all the assumptions therein but 
 $
e^{\im\kappa  g_\bt}\not \equiv e^{\im\kappa  f_\bt}$ in general.

\commentout{% The extrinsic rotation should be counted too
Notably, the {\em relative direction} of two projections determines the {\em relative position} of the common line
to the two Fourier slices, which individually are not uniquely identifiable due to extrinsic rotations. This important observation is a key to the reduction of phase retrieval to PT in each (unknown) direction of the coded aperture. 

\begin{thm}\label{thm3}
Consider a random phase mask $\mu(\bn)=\exp[\im \phi(\bn)]$ with independent, continuous random variables $\phi(\bn)\in \IR$. Let $\cT=\{\bt_0\}$ and $\cA=\{\bt_1\}$ with $\bt_1-\bt_0$ precisely known (but $\bt_0,\bt_1$ unknown). Then for $f\in O_n$, under the condition $C_\cA$ and location uncertainty,  
3D phase retrieval and PT are equivalent with respect to $\cT$ almost surely. 
\end{thm}
The proof of Theorem \ref{thm3} is  given in Appendix \ref{appA}. 

It is straightforward to extend Theorem \ref{thm3} to the case with arbitrary number of pairwise measurements. 
\begin{cor}\label{cor3}
Let $\cT$ be nonempty and $\cA=\cT-\ba$ where $\ba\in \IR^3$ is a known constant. Then for $f\in O_n$, under the condition $C_\cA$ and location uncertainty,  
3D phase retrieval and PT are equivalent with respect to $\cT$ almost surely. 
\end{cor}
}

Analogous results can be formulated for the hybrid approximation \eqref{hybrid} but we will omit the details here. 
Instead, we will present  the dark-field imaging under the weak-phase-object approximation next. 

\subsection{Weak phase objects}\label{sec:born}
Under the weak-phase-object assumption \eqref{Born} the exit wave is given by
\beq
\label{Fresnel2}
v_B(x,y)= 1-{\im\over 2\kappa}  \int d z'{f(x,y,z')}.
\eeq
 The coded diffraction pattern of the exit wave is given by
\beq
|\cF (\mu \odot v_B)|^2&=& |\cF\mu|^2+{1\over\kappa}\Im\{{ \overline{\cF\mu}\cdot \cF(\mu \int f dz')}\}+{1\over 4\kappa^2} |\cF(\mu\int fdz')|^2\label{fpt}
\eeq
where $\Im$ denotes the imaginary part. 

As
\eqref{fpt} represents the interference pattern between the reference wave $\cF(\mu)$ and the masked object wave $-\im\cF( \mu\int f dz')/(2\kappa)$, reconstruction based on the second term on the right hand side of \eqref{fpt} can be performed by conventional holographic techniques \cite{bright-field,Wolf69,Wolf70}. 

We take the diffraction patterns of the scattered waves
\beq
\label{born-pattern}
 |\cF (\mu\odot f_\bt) |^2,
 \eeq
 as measurement data, which is reminiscent of dark-field imaging in light and electron microscopies where the unscattered wave (i.e. $\cF \mu$) is removed from view \cite{Frank06,dark-field}.  Dark-field imaging mode arises naturally in
 X-ray coherent diffractive imaging due to the use of a beam stop for blocking the direct beam in order to protect the detector and enhance the measurement of weakly scattered intensities.

The next two theorems  are analogous to Theorem  \ref{thm1} and \ref{thm2}. A notable effect of the dark-field imaging is the appearance of
an undetermined phase factor absent in Theorem \ref{thm1} and \ref{thm2}. 
\begin{thm}\label{thm3}Let $f, g \in O_n$  and assume the Mask Assumption. 
 Suppose that $\supp(f_\bt)$ is not a subset of a line and  that 
\beq
\label{2.39} |\cF(  g^*_\bt)|^2&=&|\cF(  f^*_\bt)|^2\\
|\cF(\mu\odot  g^*_\bt)|^2&=&|\cF(\mu\odot  f^*_\bt)|^2\label{2.39'}
\eeq
Then almost surely $
{g_\bt}= e^{\im \theta_\bt} f_\bt$ for some constant $\theta_\bt\in \IR$. 
 \end{thm}
 The proof of Theorem \ref{thm3} is given in Appendix \ref{appC}. 

  \begin{thm}\label{thm4} 
    Let $f, g\in O_n$ and suppose that $\widehat f(0)\neq 0$. Let the Mask Assumption hold true.  Suppose that neither $\supp(f_\bt)$ nor $\supp(f_{\bt'})$ is a subset of a line  and that 
\beq
\label{43'}
|\cF(\mu\odot g^*_{\bt'})|^2&=&|\cF(\mu\odot  f^*_{\bt'})|^2\\
|\cF(\mu\odot g^*_\bt)|^2&=&|\cF(\mu\odot f^*_\bt)|^2\label{44'}
\eeq
where $\bt$ and $\bt'$ are not parallel to each other. 
Then almost surely
\beq
\label{alt}
\hbox{either}\quad \lt({g_\bt}= e^{\im \theta_0} f_\bt\quad\&\quad 
g_{\bt'}= e^{\im \theta_{0}} f_{\bt'}\rt)\quad\hbox{or}\quad
f^*_{\bt}
=f^*_{\bt'}, 
\eeq
 for some constant $\theta_{0}\in \IR$ (the two in \eqref{alt} may both be true).\footnote{ The condition $\widehat f(0)\neq 0$ is missing in the statement of  the theorem in \cite{Born-tomo}.  The proof  is corrected and further elaborated in Appendix \ref{appD}.} 
 \end{thm}

\begin{cor}\label{cor4}
Let the assumptions of Theorem \ref{thm4} hold for any two non-parallel  $\bt,  \bt'\in \cT$.  Then 
\beq
\label{alt2}
\hbox{either}\quad \lt({g_\bt}= e^{\im \theta_0} f_\bt, \quad\forall \bt\in \cT\rt)\quad\hbox{or}\quad
(f^*_{\bt}
= f^*_{\bt'}, \quad\forall\bt,\bt'\in \cT)
\eeq
where $\theta_0$ is independent of $\bt\in \cT$. 
\end{cor}
\begin{proof} First, let us make the following observation. Suppose ${g_\bt}= e^{\im \theta_\bt} f_\bt$ and ${g_{\bt'}}= e^{\im \theta_{\bt'}} f_{\bt'}$ for $\bt\neq \bt'$. By Proposition \ref{thm:slice}
$
\widehat f_{\bt}(0)=\widehat f_{\bt'}(0)=\widehat f(0)\neq 0,$
it follows from 
$\widehat g_{\bt}(0)=\widehat g_{\bt'}(0)$ that  $\theta_{\bt}=\theta_{\bt'}$. 

Consequently, 
let  $\cT_1\subseteq \cT$ be the {\em maximum} set of  all $\bt\in \cT$ for which ${g_\bt}= e^{\im \theta_0} f_\bt$ for $\theta_0\in \IR$ independent of $\bt\in \cT$. Note that the value of  $\theta_0$ is arbitrary. Since $\cT_1$ is maximal,  it follows that  $g_{\bt'}\neq e^{\im\theta_0}f_{\bt'}$ for any $\bt'\not\in \cT_1$.

Suppose the first alternative in \eqref{alt2} is not true, i.e. $\cT_1\neq \cT$. 
Consider any $\bt'\not\in \cT_1$ and $\bt\in \cT_1$. By Theorem \ref{thm4}, $f^*_{\bt}
= f^*_{\bt'}, $ 
 implying the second alternative in \eqref{alt2}.

\end{proof}  

\subsection{Pairwise measurement for single-particle imaging}\label{sec:schemes}
 \begin{figure}
\centering
\subfigure[Beam splitter with coded and uncoded apertures]{\includegraphics[width=16cm]{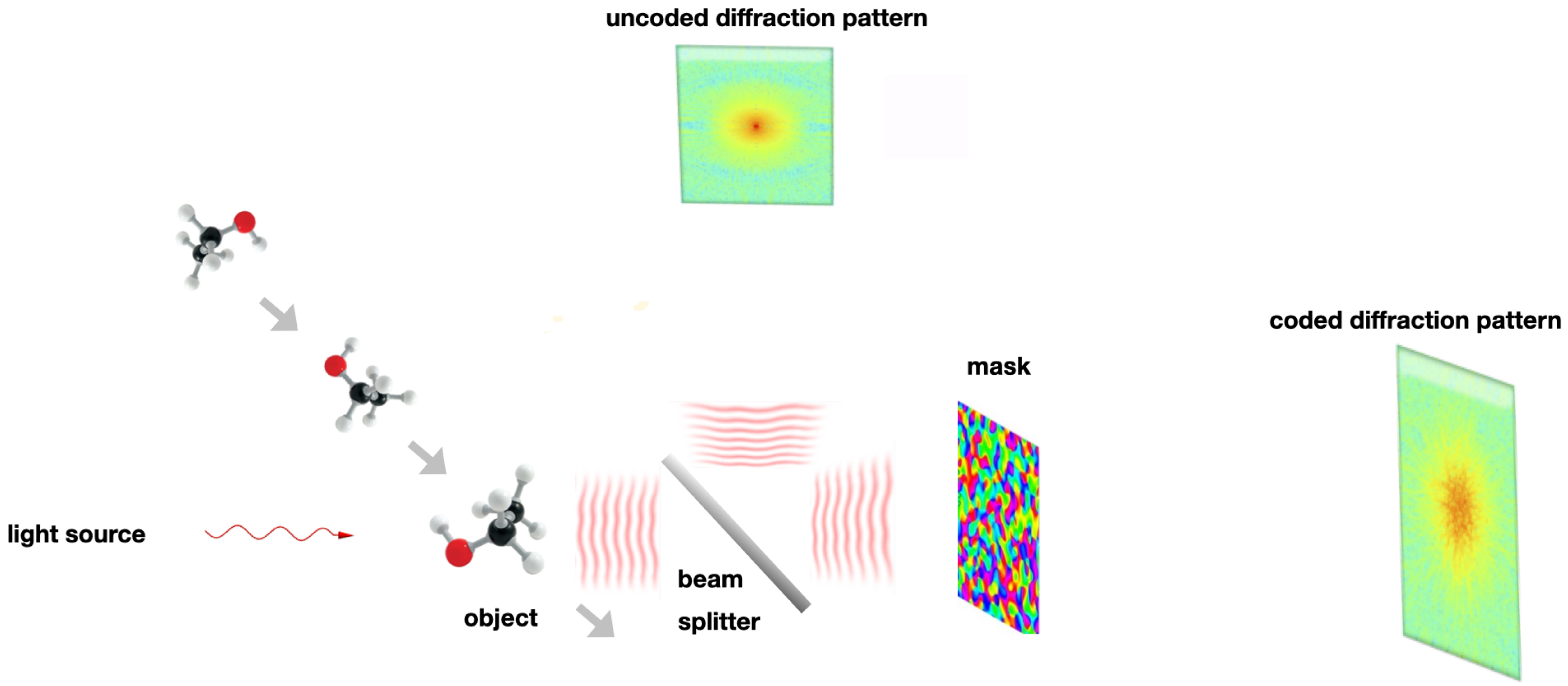}}
\vspace{1cm} \\
%{\includegraphics[width=15cm]{figs/2-incidence2}}
\subfigure[Two coded apertures in a known relative orientation]{\includegraphics[width=15cm]{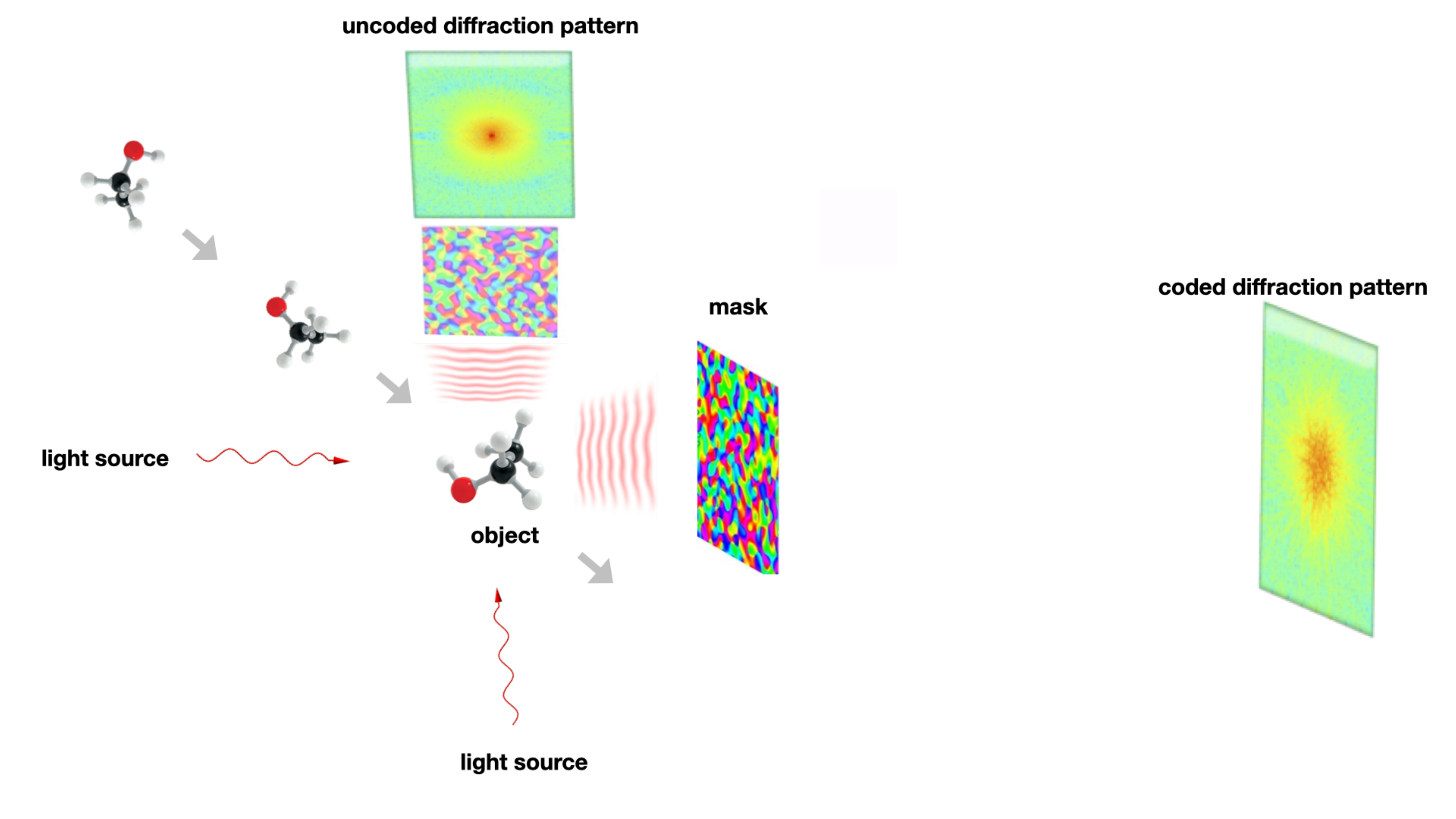}}
\caption{(a) Simultaneous measurement of one coded and one uncoded diffraction patterns with a beam splitter; (b) Simultaneous illumination of the object with two coded apertures in a known relative orientation.  %The knowledge about the relative orientation of the two beams enables reconstruction of the phase projection data by   information overlap in the common line. 
}
\label{fig3}
\end{figure}

In this section, we  introduce pairwise measurement schemes  
for single-particle imaging where each particle 
is destroyed after one illumination  \cite{XFEL18,serial-X11,XFEL17}. 
How can two diffraction patterns be measured, as assumed by Theorem \ref{thm1}, \ref{thm2}, \ref{thm3} and \ref{thm4},  if the particle is illuminated only once?

\subsubsection{Beam splitter}

For  Theorem \ref{thm1} \& \ref{thm3} with a unknown $\bt$, how can we be sure that the coded and uncoded diffraction patterns are
measured in the same direction? 

 Consider the measurement scheme stylized in Figure \ref{fig3}(a) where a beam splitter is inserted behind the object and  the mask placed  in only one of two light paths behind the splitter. Ideally, the beam splitter produces two identical beams to facilitate two snapshots of the same exit wave. The reader is referred to \cite{splitter13, splitter17, splitter20, splitter22} for recent advances in  X-ray splitters.

\subsubsection{Dual illuminations}

For Theorem \ref{thm2} \& \ref{thm4} we need a physical set-up that can render a pair of diffraction patterns in two different directions. 
This can be achieved by simultaneous illuminations  by two beams with both exit waves masked by the coded aperture as depicted in Figure \ref{fig3}(b).  Note that the two beams need not be coherent with each other since no interference between the two is called for.

%\subsubsection{Random conical tilt and orthogonal tilt}\label{sec:RCT}

 \begin{figure}
%\subfigure[Coincidence  sampling]{\includegraphics[width=6cm]{figs/OT}}\hspace{0.5cm}\hbox{$\Longrightarrow$}\hspace{0.5cm}
\subfigure[Tilt geometry for sequential illuminations]{\includegraphics[width=7.3cm]{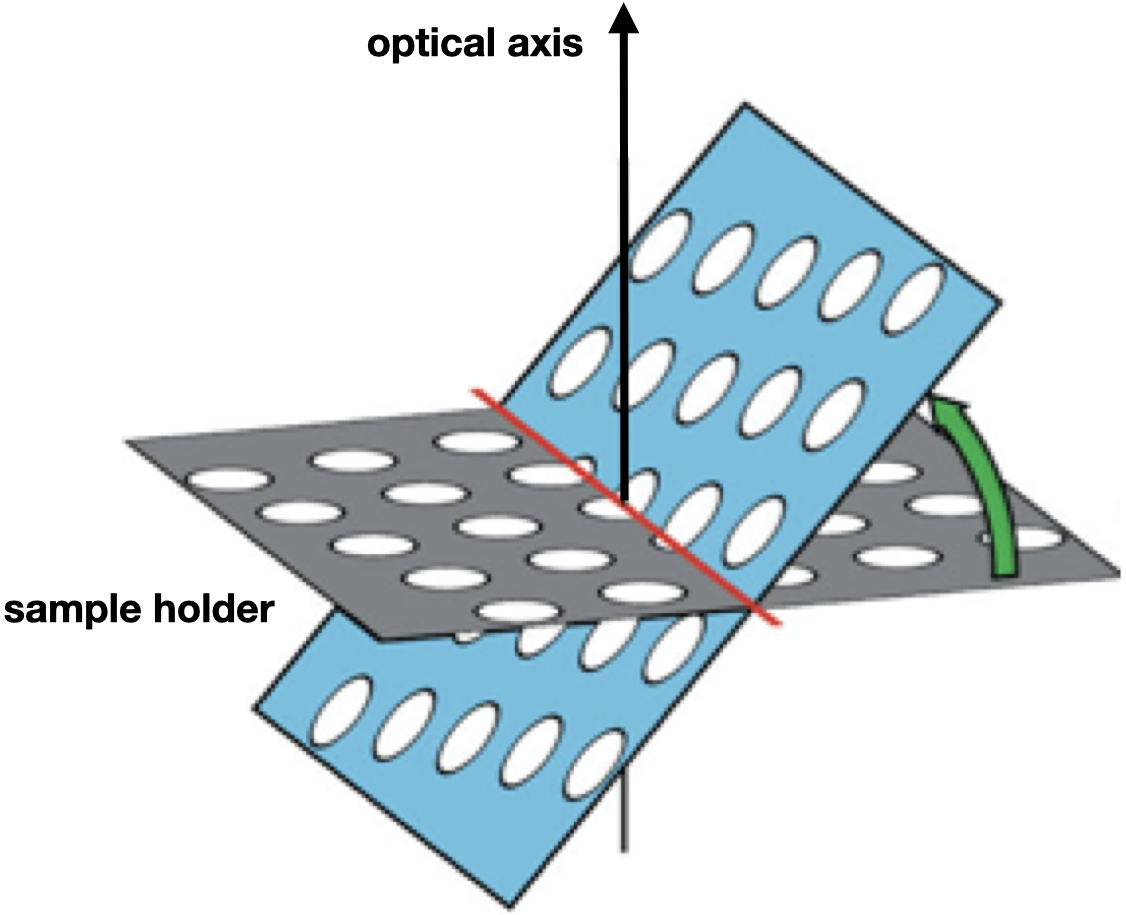}}
\caption{{\em Serial} data collection implemented by the random conical tilt and orthogonal tilt in cryo-EM both of which collect {\bf pairs} of measurement data in a fixed relative orientation corresponding to the angle about $50 \deg$ and $90\deg$, respectively,  between the two beams\cite{Frank06,OT10}.}
\label{fig:OT}
\end{figure}

If the particles can endure more than one dose of radiation then a fixed-target sample delivery can be implemented by 
the cryo-EM scheme  {\em random conical tilt} (RCT) or {\em orthogonal tilt} (OT)   \cite{Frank06,OT10}.
As shown in Figure \ref{fig:OT}, many identical particles are randomly located and oriented on a grid which can be precisely tilted about a tilt axis by a goniometer.
With dose-fractionated beams, 
the diffraction patterns of the identical particles in the two orientations are measured with the {\em coded} aperture in correspondence  with Figure \ref{fig3}(b).

Because only one exit wave is aimed at in Figure \ref{fig3}(a) 
instead of two in Figure \ref{fig3}(b), the exit wave reconstruction as guaranteed by Theorem \ref{thm1} and \ref{thm3} would be much more effective and robust
than that guaranteed by Theorem \ref{thm2} and \ref{thm4}. Indeed, the exit wave reconstruction from the two diffraction patterns collected in Figure \ref{fig3}(a) is equivalent to the phase retrieval problem well studied previously  \cite{sector,AP-phasing,DR-phasing}.

\subsection{Sector constraint}
The X-ray spectrum generally lies to the high-frequency side of various resonances associated with the binding of electrons so the complex refractive index can be written as 
\beq
\label{weak}
n=1-\delta+\im \beta,\quad 0<\delta, \beta \ll 1,
\eeq
where $\delta$ and $\beta$, respectively,  describe the dispersive and absorptive aspects of the wave-matter interaction. The component $\beta$ is usually much smaller than $\delta$ which is  often of the order of $10^{-5}$ for X-rays \cite{Jacobsen, Paganin}.

By \eqref{wrap} and \eqref{weak}, 
\beq
f=\half (n^2-1)\approx -\delta+\im \beta
\eeq
 and hence $f$ satisfies the so called sector condition introduced in \cite{unique}, i.e. the phase angle $\angle f (\bn)$   of $ f_\bt (\bn)$ for each $\bn$ satisfies 
\beq\label{sector}
\angle f(\bn)\in [a,b],\quad |a-b|< 2\pi,
\eeq
where $a$ and $b$ are two constants independent of $\bn$.
For example, for $\beta\ge 0$, $a=0 $ and $ b=\pi$. In particular,  if $\beta\ll \delta>0$, then $a\approx \pi$ and $b=\pi$. The sector condition is
a generalization of the constraint of positivity (of electron density) which is the cornerstone of the  ``direct methods"  in X-ray crystallography \cite{Hauptman86}. 

In view of \eqref{large}, the continuous interpolation $\widetilde f$ in  \eqref{1.5} satisfies
the sector condition 
\beq\label{sector'}
\angle \widetilde f(\bn)\in [\tilde a,\tilde b]\quad (\tilde a \approx a, \,\,\tilde b\approx b,\,\, p\gg 1). 
\eeq 

If $|\tilde b-\tilde a|\le \pi,$ the sector \eqref{sector'} is a convex set and hence the discrete projections
\eqref{2.8}-\eqref{2.10} also satisfy the section condition \eqref{sector'}. However, we can not expect the phase projection 
 $e^{\im \kappa f_\bt}$ to satisfy the sector condition regardless of $|\tilde b-\tilde a|$.

The sector condition \eqref{sector} enables reduction from a single coded diffraction pattern for a weak-phase object as stated below. 
\begin{thm}\cite{unique}\label{thm:born1} Let $f\in O_n$ with the singleton  $\cT=\{\bt\}$ for any $\bt$ such that the sector condition \eqref{sector'} is convex (i.e. $|\tilde b-\tilde a|\le \pi$).  Assume the mask function $\mu(\bn)=\exp[\im \phi(\bn)]$ to be uniformly distributed over the unit circle. Suppose that $\supp(f_\bt)$ is not a subset of a line and that for $g\in O_n$,  $ g^*_\bt $ produces the same coded  diffraction pattern as $ f^*_\bt$.  Then with  probability at least 
\beq
\label{sector2}
1-n^2\lt|{\tilde b-\tilde a\over 2\pi}\rt|^{\lfloor S_\bt/2\rfloor}\ge 1-n^2 2^{-\lfloor S_\bt/2\rfloor}, 
\eeq
 where $S_\bt$ is the number of nonzero pixels of $f_\bt$, 
we have $ {g_\bt}= e^{\im \theta_\bt} f_\bt$ for some constant $\theta_\bt\in \IR.$ 
 \end{thm}
 
If $|\cT|>1$ and if the mask functions for different $\bt\in \cT$ are independently distributed,  then the probability for  Theorem \ref{thm:born1} to hold for all $\bt \in \cT$ is at least 
\[
\prod_{\bt\in \cT}\lt(1-n^2\lt|{\tilde b-\tilde a\over 2\pi}\rt|^{\lfloor S_\bt/2\rfloor}\rt). 
\]
For the sake of simplicity in measurement,  however, let $\mu$ be the same mask for all $\bt\in \cT$.
The probability for  Theorem \ref{thm:born1} to hold for all $\bt \in \cT$ can be roughly estimated as follows.

First note that for any two events $A$ and $B$, 
\beq
\label{64}
P(A\cap B)=P(A)+P(B)-P(A\cap B)\ge P(A)+P(B)-1,
\eeq
where $P(\cdot)$ is the probability of the respective event. Let $\cT=\{\bt_j:j=1,...,m\}$, $E_j$ be the event that
Theorem \ref{thm:born1} to hold for $\bt_j$ and $p_j=P(E_j)$. By Theorem \ref{thm:born1}, $p_j\ge 1-c_j$
where
\[
c_j=n^2\lt|{\tilde b-\tilde a\over 2\pi}\rt|^{\lfloor S_{\bt_j}/2\rfloor}, 
\]
and, hence by \eqref{64}, 
\beq
\label{64'}
P(E_1\cap E_2)\ge p_1+p_2-1\ge 1-2c,\quad c=\max_j c_j.
\eeq
Iterating the bound \eqref{64} inductively with $E_j, j=1,...,m$, we obtain
\[
P(\cap_{i=1}^{m} E_i)=P(\cap_{i=1}^{m-1} E_i\cap E_m)\ge 1-(m-1)c-c=1-mc. 
\]

 \begin{cor}\label{cor5} Suppose $\widehat f(0)\neq 0$. 
Theorem \ref{thm:born1} holds true for $\bt$ with the same constant $\theta_\bt=\theta_0\in \IR$ independent of  $\bt$ in any $ \cT$ with probability at least
\beq
\label{sector2'}
1- |\cT|n^2\lt|{\tilde b-\tilde a\over 2\pi}\rt|^{s/2},\quad s:=\min_j S_{\bt_j},
\eeq
where $s$ is the minimum sparsity (the least number of nonzero pixels) in all directions in $\cT$. 

 \end{cor}
 
 \begin{proof}
 By Proposition \ref{thm:slice}
$
\widehat f_{\bt}(0)=\widehat f_{\bt'}(0)=\widehat f(0)\neq 0,$
it follows from 
$\widehat g_{\bt}(0)=\widehat g_{\bt'}(0)$ that  $\theta_{\bt}=\theta_{\bt'}$. 
Namely, $
g_\bt=e^{\im\theta_0} f_\bt$ for some constant $\theta_0$ independent of $\bt\in \cT$. 

 The proof is complete. 
 \end{proof}
 The bound \eqref{sector2'} is meaningful only if
 \beq
 \label{55'}
 |\cT|< n^{-2}\lt|{\tilde b-\tilde a\over 2\pi}\rt|^{-s/2}.
 \eeq
Usually $s$ is at least a multiple of $n$ (often $\cO(n^2)$),  \eqref{55'}  allows nearly exponentially large number of projections. As we will see in Corollary \ref{cor7} (ii), a far smaller number $m=n$  of projections suffices for unique determination of a weak phase object.

 Note that Theorem \ref{thm3}, \ref{thm4}, \ref{thm:born1}, Corollary \ref{cor2} and \ref{cor5} do not hold
 for a uniform  mask ($\mu=$ cost.) because the chiral ambiguity and the shift ambiguity are present, i.e. both 
 $g(\cdot) =f(-\cdot)$ and $g(\cdot)=f(\cdot+\bl),\bl\in \IZ^3, $ satisfy all the assumptions therein but $ g_\bt\not \equiv e^{\im\theta_\bt}  f_\bt$ in general.

\section{Phase unwrapping}
\label{sec:unwrap}

For a strong phase object, 
 \eqref{1.18'} naturally  leads  to the problem of phase unwrapping:
 \beq\label{37'}
 g_\bt(\bn)= f_\bt(\bn) \quad\mod 2\pi/\kappa,\quad \bn\in \IZ_p^2, 
\eeq
which may have infinitely many solutions. We seek conditions on  $\cT$ that can uniquely determine the 3D object in the sense that 
$
 g\equiv f.$

A basic approach appeals to the continuity of the projection's dependence on the direction $\bt$, which, in turn, is the consequence of continuous interpolation \eqref{1.5}. 

Let $\cT_\ep$ denote the
graph with the nodes given by  $\bt\in \cT$ and the edges defined between any two nodes $\bt_1,\bt_2\in \cT$ with $|\angle\bt_1\bt_2|\le \ep$ (such edges are called $\ep$-edges) where $\angle\bt_1\bt_2$ is the angle between $\bt_1$ and $\bt_2$. We call $\cT$ is $\ep$-connected if $\cT_\ep$ is a connected graph. We say that two nodes $\bt_1,\bt_2$ are $\ep$-connected if there is an $\ep$-edge between them. 

Suppose that $\cT$ is $\ep$-connected for certain $\ep$ (to be determined later). The continuous dependence of $g_\bt, f_\bt$ on $\bt$ implies that $|g_{\bt_1}-g_{ \bt_2}|$
and $|f_{\bt_1}-f_{ \bt_2}|$ are arbitrarily small if $\angle\bt_1\bt_2$ is sufficiently small. On the other hand,
$h_{\bt_1}-h_{ \bt_2}$ is an integer multiple of $2\pi/\kappa$ where $h_\bt := g_\bt- f_\bt$.  Then for  sufficiently small $\ep$, $h_\bt(\bn)$ is a constant for each $\bn$ and hence $ g_\bt- f_\bt$ is  independent of $\bt$.

We can give a rough estimate for the required closeness $\ep$ of two adjacent projections as follows.
In general, the gradient of the continuous extension $\widetilde f$ is $\cO(1)$ and hence the gradient of $f_\bt$ (being a sum of $n$ values of $f$) with respect to $\bt$ is $\cO(n)$.  Consequently,  $|f_{\bt_1}-f_{\bt_2}|$ can be made sufficiently small with $\ep=\cO(1/n)$ (with a sufficiently small constant).  

As pointed out at the beginning of Section \ref{sec:discrete}, if we make use of the property that the fractional variation of $f$ between adjacent grids is negligible, then  the gradient of the continuous extension $\widetilde f$ is $o(1)$ and  $|f_{\bt_1}-f_{\bt_2}|$ can be made sufficiently small with $\ep$ that may be much larger than $ 1/n.$

 The main result of this section will need the diversity condition:
 \beq
\label{100}\label{101}
%&&(\alpha_0,\beta_0)\neq (0,0)\quad \& \\
&&\#\{\alpha_l \xi+\beta_l\eta: |\alpha_l|,|\beta_l|<1,\,\, l=1,\dots, n\} =n \hbox{ for each $\xi,\eta\in \IZ_p, $ $(\xi,\eta)\neq 0$}.
\eeq
In other words, the difference $(\alpha_l-\alpha_{k},\beta_l-\beta_{k}), l\neq k,$ is not expressible as $Z_p$-fractions.
\begin{thm}\label{tom2}
Let $\cT$ be an $\ep$-connected set of directions {\bf containing any} of the following three sets:
\beq
\label{x}
&&\{(1, \alpha_l,\beta_l): l=1,\dots,n\}\cup \{(0, \alpha_0,\beta_0),  (0,0,1)\}\\
\label{y}&&\{(\beta_l,1, \alpha_l): l=1,\dots,n \}\cup \{(\beta_0,0, \alpha_0), (1,0,0)\}\\
\label{z}&&\{(\alpha_l,\beta_l, 1): l=1,\dots,n \}\cup \{(\alpha_0,\beta_0,0), (0,1,0)\}
\eeq
under the assumption \eqref{100} and that $\alpha_0\neq 0, |\beta_0|<1$.  

Suppose that the maximum variation of the object $f$ between two adjacent  grid points is less than $\pi/\kappa$ (The 3D Itoh condition) and   \eqref{1.18'}  holds for a sufficiently small $\ep=\cO(1/n)$. Then 
$g=f$. \end{thm}

\begin{rmk}
As per the discussion in Section \ref{sec:discrete}, with $\lambda/2$ as the unit of length,  $\pi/\kappa=1$.  

The projection $f_\bt$ in a direction $\bt$, however, usually violates the 2D Itoh condition.  Hence 2D phase unwrapping for $f_\bt$ may not have a unique solution \cite{Itoh,2D-unwrap}. 

\end{rmk}

\begin{rmk}\label{rmk1:tom2}

As shown in the following proof, the $x,y$ and $z$ axes  in \eqref{x}, \eqref{y} and \eqref{z}, respectively,  show up in the analysis 
because  they are ``privileged" w.r.t.  $\IZ_n^3$ which is not isotropic. On the other hand, due to arbitrariness in choosing the orientation of the object frame,
we can always designate  one of the projection directions in $\cT$ as exactly one of the coordinate axes, say, (0,0,1), and discretize  the object domain into $\IZ_n^3$ accordingly. 

\end{rmk}

\begin{proof} It suffices to consider the case that $\cT$ contains the set \eqref{x}. 

Let \eqref{37'} hold true.  Then the $\ep$-connected schemes with sufficiently small $\ep$ ensure
\[
h_\bt(\bn):= g_\bt(\bn)- f_\bt(\bn) \quad  \hbox{is   independent of $\bt\in \cT$}.
\]
Intuitively, with  sufficiently diverse views in $\cT$, $h:=g-f$ must be  a multiple of Kronecker's delta function
as shown in the following analysis.  

Let $c(\cdot,\cdot)$ be  independent of $\alpha_l,\beta_l,$ such that
\beq\label{3.33}
\widehat h_{(1,\alpha_l, \beta_l)}(\eta,\zeta)&= &c(\eta,\zeta)%a_{jk}
%\label{3.33'}\widehat{g}_{y(\alpha, \beta)}(j,k)  &= & c(-j,k)\\% b_{jk}\\
%\label{3.33''}\widehat{g}_{z(\alpha, \beta)}(j,k)&=& c(-k,j)
\eeq
and hence by Fourier Slice Theorem
\beq
\label{3.45}
\widehat h(- \eta,\zeta\in \IZ\alpha_l \eta-\beta_l \zeta, \eta,\zeta)&=&c(\eta,\zeta),\quad  \eta,\zeta\in \IZ_p.%a_{jk}\\
%\widehat g(j,-\alpha j-\beta k, k)&=& c(-j,k)\\%b_{jk}\\
%\widehat g(j,k,-\alpha j-\beta k)&=& c(-k,j)
\eeq
We want to show $c\equiv 0$.

%With $\cT$ given by \eqref{T} we consider \eqref{3.45} with $\beta=0$: 
Define the notation:
\beq
\label{3.35}
\widehat{h}_{\eta}(m,l)&=&\sum_{k} h(m,k,l) e^{-2\pi \im k\eta/{p}}\\
\widehat{h}_{\eta\zeta}(m)&=&\sum_{l} \widehat h_\eta (m,l) e^{-2\pi \im l\zeta/{p}}.\label{3.36}
\eeq
Clearly
\beq
\label{xg}
\widehat h(\xi,\eta,\zeta)&= &\sum_{m}\widehat h_{\eta\zeta}(m) e^{-2\pi \im m\xi /{p}}.
\eeq

 % For $\eta\neq 0$ and  with $n$ distinct sample points of the form
%\beq
%\xi_i=-\alpha_i \eta \in (-p/2, p/2), \quad \alpha_i=s_i,\label{3.50}
%\eeq
By the support constraint $\supp(h)\subseteq \IZ_n^3$, \eqref{3.45} \& \eqref{xg} become the $n\times n$ Vandermonde  system 
\beq
\label{van'}
V \widehat h_{\eta\zeta}=c(\eta, \zeta) \II
\eeq
with   the all-one vector $\II$ and
\beq
V&=&[V_{ij}],\quad V_{ij}= e^{-2\pi \im \xi_ij/{p}}, \quad \xi_i=-\alpha_{i}\eta-\beta_{i}\zeta %\quad \theta=[\theta_{x(\alpha_i,\beta_i)}]_{i=1}^n,
%X&=&\diag\big[e^{\im \theta_{x(\alpha_i,0)}}\big]. 
\eeq
for $\{\alpha_{i},\beta_{i}:i=1,\dots,n\}$.
The Vandermonde system is nonsingular if and only if $\{\xi_i: i=1, \dots,n\}$ has $n$ distinct members. 

Since the system \eqref{van'} has a unique solution,  we identify 
$\widehat h_{\eta\zeta}(\cdot)$ as 
\[
\widehat h_{\eta\zeta}(\cdot)=c(\eta,\zeta)\delta(\cdot).
\]
For $m\neq 0$, $\widehat{h}_{\eta\zeta}(m)=0$ for all $\eta, \zeta\in \IZ$ and hence $\widehat h_{\eta}(m,l)=0 $ for all $l$ and $m\neq 0$. Likewise for \eqref{3.35}, we select $n$ distinct values of $\eta$ to perform inversion of the Vandermonde system and obtain   
\beq
\label{amb2}
h(m, k, l)=0, \quad m\neq 0.
\eeq
 In other words, $h$ is supported on the $(y,z)$ plane. 
Consequently  the projection of $h$ in the direction of $(0,\alpha_0,\beta_0),$ with $\alpha_0\neq 0, $ would be part of  a line segment
and, hence by the assumption of $h_\bt$'s independence of  $\bt\in \cT$,
$h_\bt$ is also a line object for all $\bt\in \cT$. 

That is to say, $h$ is supported on the $z$-axis. 
Now that $(0,0,1)\in \cT$, the projection of $h$ in $(0,0,1)$ is Kronecker's delta function $\delta$ at the origin, $h_\bt$'s  independence of $\bt$  implies that  for some $q\in \IZ$, 
\beq
g(\bn)-f(\bn)={2\pi \over \kappa} q\delta(\bn)\label{amb}
\eeq
where $\delta$ is Kronecker's  delta function on $\IZ^3$. 

The ambiguity on the right hand side of \eqref{amb} 
can be further eliminated by limiting the maximum variation of the object between two adjacent grid points to
less than $\pi/\kappa$, the so called Itoh condition \cite{Itoh}. This can be seen as follows: If both $g$ and $f$ satisfy Itoh's condition as well as $
g(\bn)-f(\bn)=0$ for $\bn\neq 0$, 
then $|g-f|< 2\pi/\kappa$ at the origin, implying $q=0$ in \eqref{amb}. The proof is complete. 

 \end{proof}

%\section{Uniqueness with strong phase objects}\label{sec:strong}
\begin{figure}
\centering
%{\includegraphics[width=10cm]{figs/cut-sphere}}
%{\includegraphics[width=10cm]{figs/circular-triangle}}
%{\includegraphics[width=7cm]{figs/circular-triangle2}}
{\includegraphics[width=7cm]{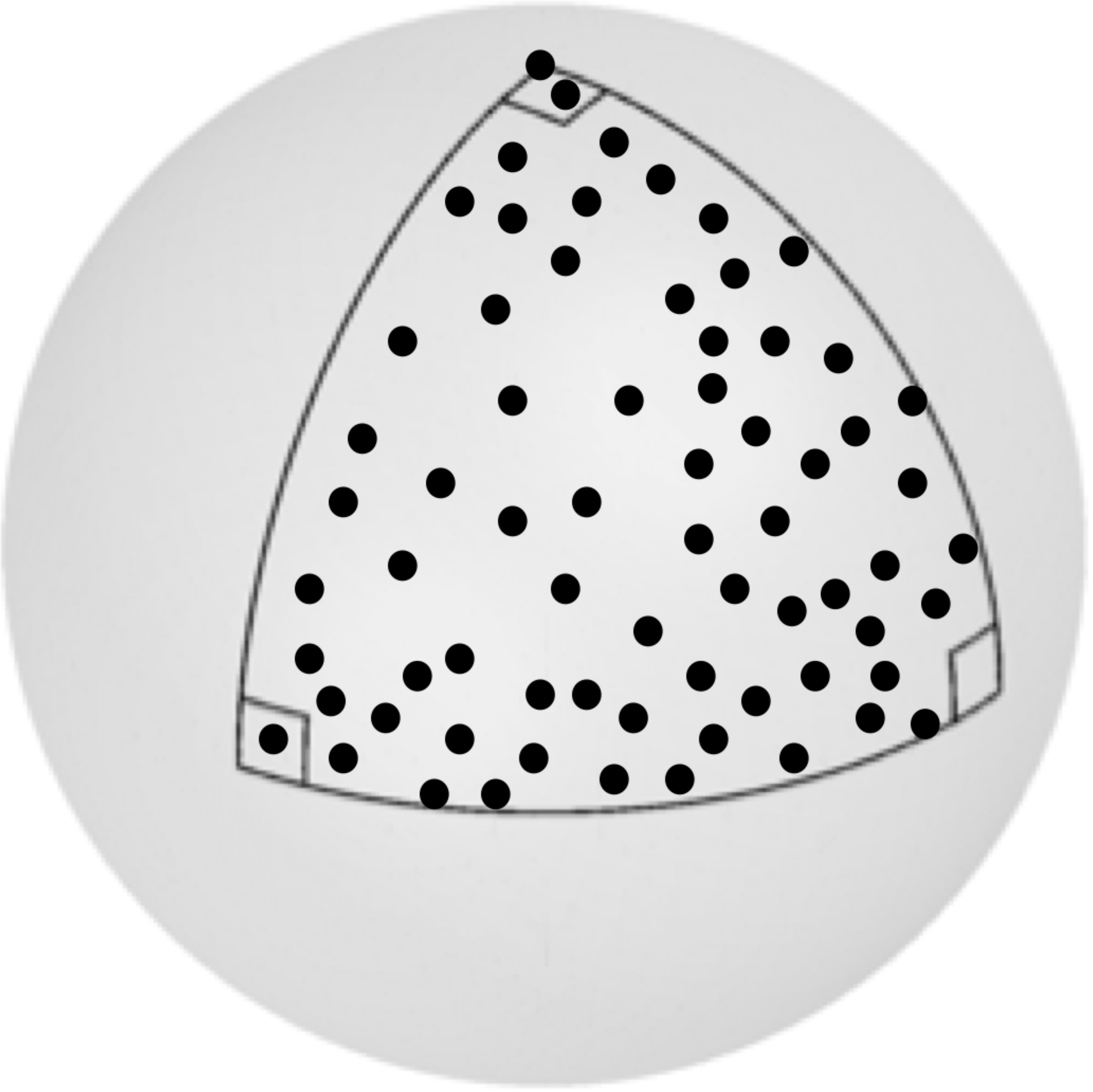}}
\caption{Unit sphere  representing all directions in the object frame. A sufficiently large set of randomly selected points from the spherical triangle (or a larger one) contain the scheme \eqref{x} and satisfy the conditions in Theorem \ref{tom2}.  Since $\ep=\cO(1/n)$ with a small constant, $|\cT|$ is 
at least a large multiple of $n$. 
}
\label{fig:sphere}
\end{figure}

In view of Theorem \ref{thm1}, \ref{thm2} and \ref{tom2}, we have  the following uniqueness results for 3D phase retrieval  with a strong phase object.

\begin{cor}\label{cor6} Let $\cT$ be a $\ep$-connected set of directions in Theorem \ref{tom2}  for a sufficiently small $\ep=\cO(1/n)$ satisfying \eqref{101}. Consider the class of objects in $O_n$  with the maximum variation between two adjacent  grid points less than $\pi/\kappa$. 

(i) Under the setting of Theorem \ref{thm1}, the pairs of coded and uncoded diffraction patterns corresponding to $\bt\in \cT$  uniquely determine the strong phase object almost surely .

(ii) Under the setting of Corollary \ref{cor2},  the  coded diffraction patterns corresponding to $\cT$ uniquely determine the strong phase object   almost surely.

\end{cor}

{
\subsection{The hybrid approximation} \label{sec:hybrid} Let us sketch the extension of Theorem \ref{tom2} to the hybrid approximation \eqref{hybrid} with  $q>1$
for which the phase unwrapping problem is finding conditions on $\cT$ such that the relation
\beq
\Big(1+{\im\kappa f_\bt\over q}\Big)^q=\Big(1+{\im\kappa g_\bt\over q}\Big)^q,\quad \forall\bt \in \cT,\label{hybrid-unwrap}
\eeq
for $f,g\in O_n$ implies $g=f$. Eq.  \eqref{hybrid-unwrap} is equivalent to 
\beq
g_\bt-\om f_\bt={q\over \im\kappa} (\om-1)\label{hybrid-unwrap2}
\eeq
where $\om$ is a $q$-th root of unity. 

By the above analysis  an $\ep$-connected $\cT$ with a sufficiently small $\ep=\cO(1/n)$ implies that $\om$ in \eqref{hybrid-unwrap2} is independent of $\bt$. With a slight modification of the argument for Theorem \ref{tom2}, we arrive the conclusion that the right hand side of \eqref{hybrid-unwrap2} is zero and hence $\om=1$, implying $f=g$. 

}
\subsection{Tilt schemes for phase unwrapping}\label{sec:tilt}
In this section,  we consider a few examples as applications of  Theorem \ref{tom2} and Corollary \ref{cor6}. 

\subsubsection{Random tilt}
\label{sec:random-tilt}
We can satisfy condition \eqref{101} with overwhelming probability by randomly and independently selecting $n$ pairs of $(\alpha_l,\beta_l)$ that are distributed with probability density function bounded away from 0 and $\infty$  over any square contained in $[0,1]^2$ (see \cite{rand}). 

In the case of \eqref{x} with $\alpha_l,\beta_l\in [0,1)$, for instance, this random tilt series is distributed over the spherical rectangle of azimuth range $[0,\pi/4)$ and polar angle range 
$(\pi/4,\pi/2]$. We can enlarge the  random sampling area from this spherical rectangle to  the spherical triangle shown in Figure \ref{fig:sphere}. If the sampling is sufficiently dense, then the whole scheme 
would include the direction $(0,1,\beta_0)$, for some $\beta_0\in (0,1)$,  and $\ep$-connect to $(0,0,1)$ (which is included by assumption),

More generally,  the conditions of Theorem \ref{tom2} are satisfied by any tilt series of sufficiently dense sampling from a spherical triangle with vertexes  in three orthogonal directions (cf.  Remark \ref{rmk1:tom2}).   

Random schemes arise naturally in single-particle imaging. On the other hand, deterministic tilt schemes are
often employed in tomography.

\subsubsection{Deterministic tilt} 
\label{sec:deterministic-tilt}
\commentout{
 \begin{figure}
\centering
%{\includegraphics[width=10cm]{figs/cut-sphere}}
%{\includegraphics[width=10cm]{figs/circular-triangle}}
%{\includegraphics[width=7cm]{figs/circular-triangle2}}
\caption{A sphere  representing all directions in the object frame. Any {\em two} of the three great circular arcs, when sufficiently sampled, provide enough information for unwrapping the object phases (cf. \eqref{1.60} with $\alpha=0$). This is an example of dual-axis tilting, each of range $\pi/2$. }
\label{fig:sphere}
\end{figure}
}

First, the single-axis tilting (with the conical angle $\pi/2$) is not covered by Theorem \ref{tom2} and contains  certain blindspot as exhibited  in the proof, i.e. it can not  completely resolve the ambiguities in phase unwrapping.

{Second, certain combinations of  \eqref{x}, \eqref{y}, \eqref{z}, can be made $\ep$-connected 
 (for sufficiently small $\ep$) in the following scheme:
 \beq\label{1.60}
\cT&=&\{(1,{l\over q},\alpha): l=0,\dots,q \}\cup\{ ({l\over q}, 1, \alpha):l=0,\dots,q \} \\
&&\cup \{(0,1,{l\over q} ):l=0,\dots,q \}
\cup\{ (0,{l\over q},1): l=0,\dots,q \},\quad q\in \IN\nn
%&&\{(\beta_l,1, 0): l=1,\dots,m \}\cup \{(0,0, 1), (1,0,0)\}\\
%&&\{(0,\beta_l, 1): l=1,\dots,m \}\cup \{(1,0,0), (0,1,0)\}
\eeq
with a fixed $\alpha\in [0,1)$, where the first subset is from \eqref{x}, the second and third from \eqref{y}
and the  fourth from \eqref{z}. %By the symmetry of projection, we may assume $\alpha\ge 0$.
In the limit of $q\to\infty$, the scheme \eqref{1.60} has a continuous limit which can be illustrated more concretely in terms of  the spherical polar coordinates as in the following example. 

\begin{figure}
\centering
%\subfigure[Conical tilting around the $z$-axis]{\includegraphics[width=6cm]{figs/parallels}}
%\vspace{-0.5cm}
\subfigure[]{\includegraphics[width=8cm]{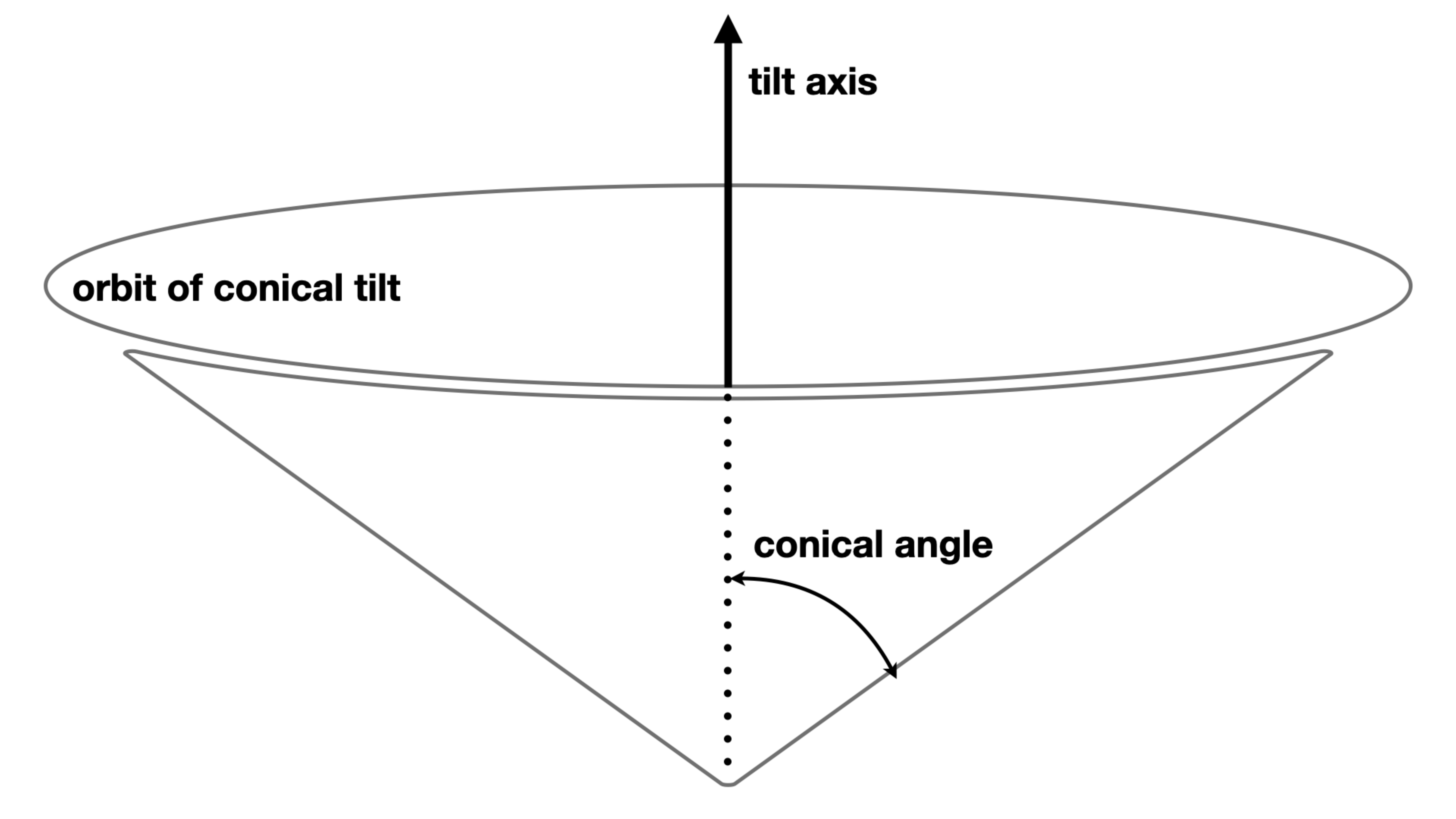}}\hspace{0.5cm}
\subfigure[]{\includegraphics[width=5cm]{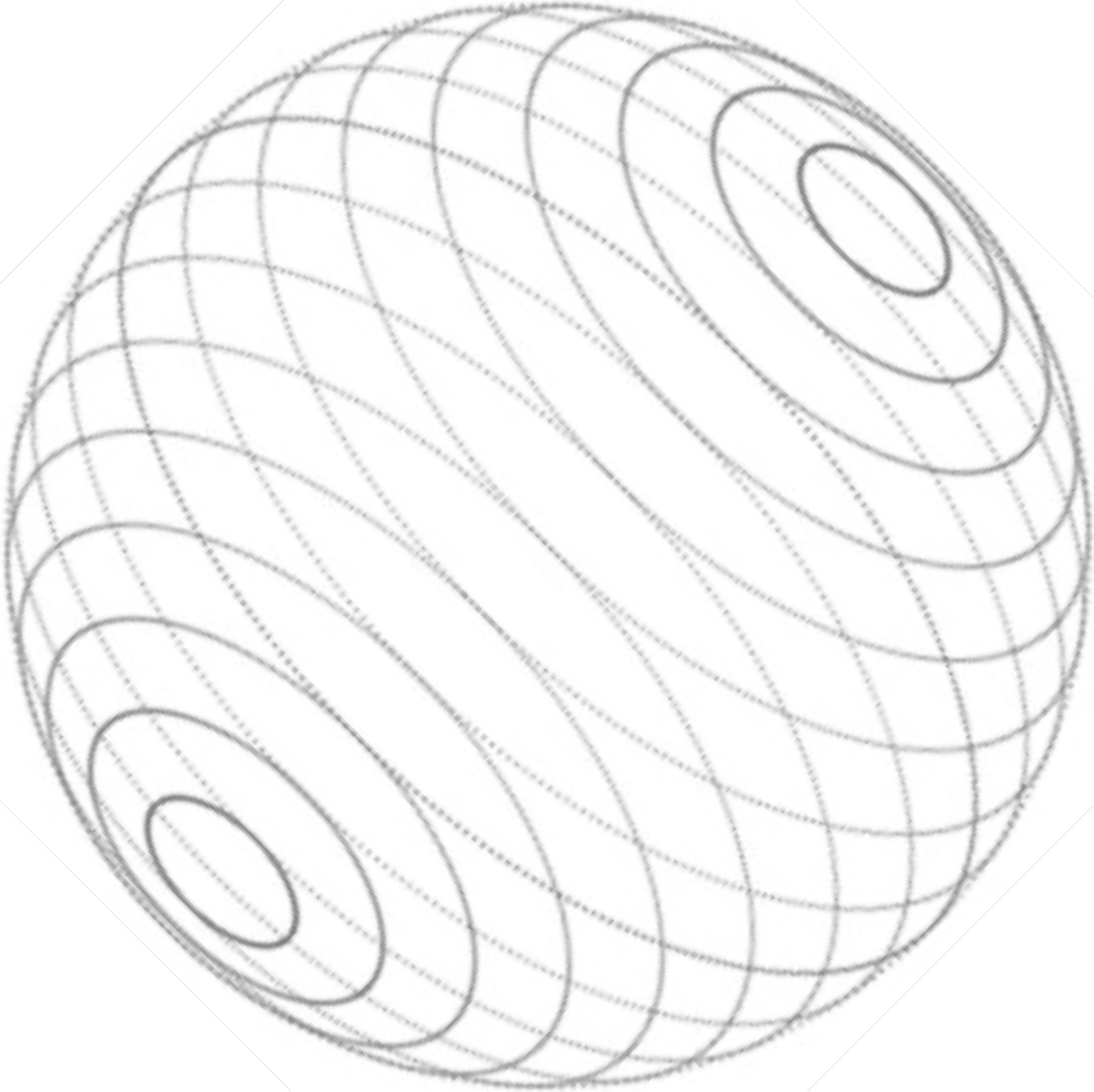}}
\caption{(a) Conical tilt geometry; (b) Conical tilt orbits  
of various conical angles about an axis of obliquity. A single-axis tilt orbit is a great circle, corresponding to a conical angle $\pi/2$, which can not uniquely unwrap phase. }
\commentout{
\caption{ (a) Lines of latitude representing data collection by conical tilting of different conical angles around the $z$-axis; (b) Conical tilts around an axis of obliquity 
which may be parametrized as in Theorem \ref{tom2}.  }
}
\label{fig:conical}
\end{figure}

 \begin{ex}\label{ex1}
{\rm The continuous limit of \eqref{1.60} consists of two circular arcs. 
The first arc, the limit of the first and second subsets in \eqref{1.60}, going from $(1,0,\alpha)$ to $(1,1,\alpha)$ and to $(0,1,\alpha)$,  is parametrized by the azimuthal angle $ \phi_z\in [0,\pi/2],$ at the polar angle  $\theta_z=\hbox{arc}\cot(\alpha)>\pi/4$ (since $\alpha\in [0,1)$) w.r.t. the polar axis $z$. 
The second arc, the limit of the third and fourth subsets in \eqref{1.60}, going from $(0,1,\alpha)$ to $ (0,1,1)$ and to $(0,0,1)$, is parametrized by the azimuthal angle $ \phi_x\in [\arctan (\alpha), \pi/2],$ at the polar angle $\theta_x=\pi/2$ w.r.t. the polar axis $x$. 
\commentout{The total length of the orbit is ${\pi\over 2\sqrt{1+\alpha^2}}+\arccos[ \alpha/\sqrt{1+\alpha^2}]$ which has the minimum value $ (\half+{1\over \sqrt{2}}){\pi\over 2} \approx 1.90$ at $\alpha=1$. }

In other words, the continuous limit of \eqref{1.60} is an union of a conical tilting (the first arc) of range $\pi/2$ at the conical angle $\hbox{arc}\cot(\alpha)$ and an orthogonal single-axis tilting (the second arc) of range $\hbox{arc}\cot(\alpha)$. 

\commentout{Single-axis tilting (maximal range $\pi$ due to the symmetry of projection) and conical tilting (maximal range  $2\pi$) are two simple data collection schemes in cryo-EM \cite{Frank06}. }

For  $\alpha=0$, the scheme is an orthogonal dual-axis tilting of a tilt range $\pi/2$ for each axis  \cite{dual-axis95}. The total length of the orbit is $\pi$.
\hfill $\square$. 
}
\end{ex}

Note that  the total radiation dose is proportional to the number of projections, which is $\cO(n)$ with a large constant (since $\ep=\cO(1/n)$ with a small constant), and, as $n\to \infty$, proportional to the orbit's total length on the unit sphere.

More conveniently,  instead of being split into a conical tilting and a single-axis tilting as in \eqref{1.60}, the schemes in Theorem \ref{tom2}  can be implemented as a single conical tilting (Figure \ref{fig:conical}) which  has a smooth circular orbit, instead of  a broken one. 

\begin{ex}\label{ex2}
{\rm 
Let $(1,0,0)$ and $(0,0,1)$, respectively, be the start and the end  of the orbit, with the midpoint $(0,1,0)$.  Any three directions of the conical tilt uniquely determine 
the direction of the tilt axis, $(1,1,1),$ with the conical angle, $\arccos(1/\sqrt{3})\approx 54.7^\circ,$ and  the tilt range,  $4\pi/3$.  The total length of the orbit is ${4\pi\over 9} \sqrt{10-2\sqrt{3}}\approx 3.57$ which is slightly larger than the
length  $\pi$ in Example \ref{ex1} for $\alpha=0$. 

The conical tilt going through $(1,0,\alpha)$, $(0,1,\alpha)$ and $(0,0,1)$
can be similarly constructed. We leave the details to the interested reader.   \hfill $\square$
}
\end{ex}

\section{Uniqueness with weak phase objects}
\label{sec:weak}

In this section we show that  for weak phase objects,  much less restrictive schemes than those of Theorem \ref{tom2} guarantee uniqueness of solution to 3D phase retrieval.

The following  is the uniqueness  result for discrete  projection tomography. 
\begin{thm}
\label{tom2-weak}
 Let  $\cT$ be {\bf any} one of the following sets:
\beq
\label{x'}&&\{(1, \alpha_l,\beta_l): l=1,\dots,n \}\\
\label{y'}&&\{(\alpha_l,1, \beta_l): l=1,\dots,n \}\\
\label{z'}&&\{(\alpha_l,\beta_l, 1): l=1,\dots,n \}
\eeq
under the condition \eqref{101}. 
\commentout{
with the property that 
\beq
\label{100}
&&(\alpha_0,\beta_0)\neq (0,0)\quad \& \\
&&\{\alpha_l \xi+\beta_l\eta: l=1,\dots, n\} \hbox{ has $n$ distinct members for each fixed $(\xi,\eta)\neq (0,0)$}. \label{101}
\eeq
} 
\commentout{Suppose that $\widehat f(0)\neq 0$ and 
that $g_\bt=e^{\im \theta_\bt} f_\bt, \bt\in \cT, $ for some constants $\theta_\bt\in \IR$.  Then 
$g=e^{\im \theta_0} f$  for some constant $\theta_0\in \IR$ independent of $\bt\in \cT$. 
}
Suppose that $g_\bt=e^{\im \theta_0} f_\bt,$ for some constant $\theta_0\in \IR$,
 independent of $\bt\in \cT$.  Then 
$g=e^{\im \theta_0} f$. 

\end{thm}

\begin{rmk}
Theorem \ref{tom2-weak} is the finite, discrete counterpart of
the classical result that a compactly supported function  is uniquely determined by the projections in any infinite set of directions (\cite{Helgason}, proposition 7.8). 
\end{rmk}

\begin{proof} 
\commentout{By Proposition \ref{thm:slice}
$
\widehat f_{\bt}(0)=\widehat f_{\bt'}(0)=\widehat f(0)\neq 0,$
it follows from 
$\widehat g_{\bt}(0)=\widehat g_{\bt'}(0)$ that  $\theta_{\bt}=\theta_{\bt'}$. 
Namely, $
g_\bt=e^{\im\theta_0} f_\bt$ for some constant $\theta_0$ independent of $\bt\in \cT$. 
}

To fix the idea,  consider the case  \eqref{z'} for $\cT$.
 By the discrete Fourier slice theorem, we have
\beq
\label{1.14'}
\widehat{g}(\xi,\eta, -\alpha_l \xi-\beta_l \eta) &= &e^{\im \theta_0}\widehat{f}(\xi,\eta, -\alpha_l \xi-\beta_l \eta),\quad l=1,\dots, n,\quad  \xi,\eta\in \IZ. 
%\widehat{g}(j, -\alpha j-\beta k,k) &=&e^{\im \theta_0}\widehat{f}(j, -\alpha j-\beta k,k)\\
%\widehat{g}(j,k, -\alpha j-\beta k)&=&e^{\im \theta_0}\widehat{f}(j,k, -\alpha j-\beta k)
\eeq
In other words, for each $\xi,\eta\in \IZ,$ the corresponding partial Fourier transforms defined in \eqref{xg} satisfy 
\beq
\label{110}
\sum_{m\in \IZ_n}(\widehat{g}_{\xi\eta}(m)- e^{\im\theta_0}\widehat{f}_{\xi\eta}(m))e^{-\im 2\pi m( -\alpha_l \xi-\beta_l \eta)/p }=0,\quad l=1,...,n
\eeq
in terms of the notation for  partial Fourier transforms in the proof of Theorem \ref{tom2}. 
For each $\xi,\eta, $ \eqref{110} is a Vandermonde system which is nonsingular if and only if \eqref{101} holds.
This implies that
\[
\widehat{g}_{\xi\eta}(m)= e^{\im\theta_0}\widehat{f}_{\xi\eta}(m),\quad m\in \IZ_p,\quad\forall \xi,\eta\in \IZ. 
\]
Therefore, $g=e^{\im \theta_0}f$ as asserted.
\end{proof}

It may be interesting to compare  Theorem \ref{tom2-weak}  with Crowther's rough estimate
\beq
\label{123}
N={\pi \over 2}n 
\eeq
 for the number $N$ of projections needed for projection tomography with  a single-axis tilting of tilt range $\pi$
 (\cite{Jacobsen}, eq. (8.3)).

We have the following uniqueness results for 3D  phase retrieval  with a weak phase object.
\begin{cor}\label{cor7}  Let  $\cT$ be {\bf any} one of the direction sets in Theorem \ref{tom2-weak}. 
\begin{itemize} 

\item[(i)] Under the setting of Theorem \ref{thm3}, the 
$n$ pairs of coded and uncoded diffraction patterns corresponding to $\cT$  uniquely determine the weak phase object almost surely.

%\item[(ii)] 3D phase retrieval with coded diffraction patterns corresponding to $\cT$ has a unique solution in $O_n$ almost surely.

\item[(ii)] Under the setting of Corollary \ref{cor5},  the $n$ coded diffraction patterns corresponding to $\cT$  uniquely determine  the weak phase object with high probability (for $n\gg 1$). 
\end{itemize}
\end{cor} 

In contrast to Corollary \ref{cor7} (ii), the setting of Corollary \ref{cor4} requires an extra coded diffraction pattern to remove
the isotropy ambiguity.
\begin{thm}\cite{Born-tomo}\label{thm:born}  Let  $\cT$ be {\bf any} one of the following direction sets 
\beq
\label{x''}&&\{(1, \alpha_l,\beta_l): l=1,\dots,n \}\cup \{(0,\alpha_0,\beta_0)\}\\
\label{y''}&&\{(\alpha_l,1, \beta_l): l=1,\dots,n \}\cup \{(\alpha_0,0,\beta_0)\}\\
\label{z''}&&\{(\alpha_l,\beta_l, 1): l=1,\dots,n \}\cup \{(\alpha_0,\beta_0,0)\}
\eeq
under the condition \eqref{101} and $(\alpha_0,\beta_0)\neq (0,0)$. Then in the setting of Corollary \ref{cor4},
$g=e^{\im\theta_0} f$  for some constant $\theta_0\in \IR$ almost surely.
\end{thm} 

\begin{proof}
To  rule out the second alternative in Corollary \ref{cor4} that $f^*_{\bt}
=f^*_{\bt'}, \forall\bt,\bt'\in \cT,
$
define $h_\bt:=f^*_\bt$ which is independent of $\bt\in \cT$.
Now applying the analysis in the proof of Theorem \ref{tom2} to this $h_\bt$ for the scheme, e.g. \eqref{x''}.  The argument up to \eqref{amb2}
leads to the conclusion that
the projection of $h$ in the direction of $(0,\alpha_0,\beta_0),$ with $\alpha_0\neq 0, $ is part of  a line segment
and hence  $h_\bt$ is a line object for all $\bt\in \cT$. This violates
the assumption in Corollary \ref{cor4} that no projection is
part of a line. This implies that the first alternative of Corollary \ref{cor4} holds, i.e. ${g_\bt}= e^{\im \theta_0} f_\bt, \quad\forall \bt\in \cT.$
\end{proof}

\section{Noise robustness}\label{sec:num}
    \begin{figure}
  \centering
\subfigure[2D image]{\includegraphics[width=4.5cm]{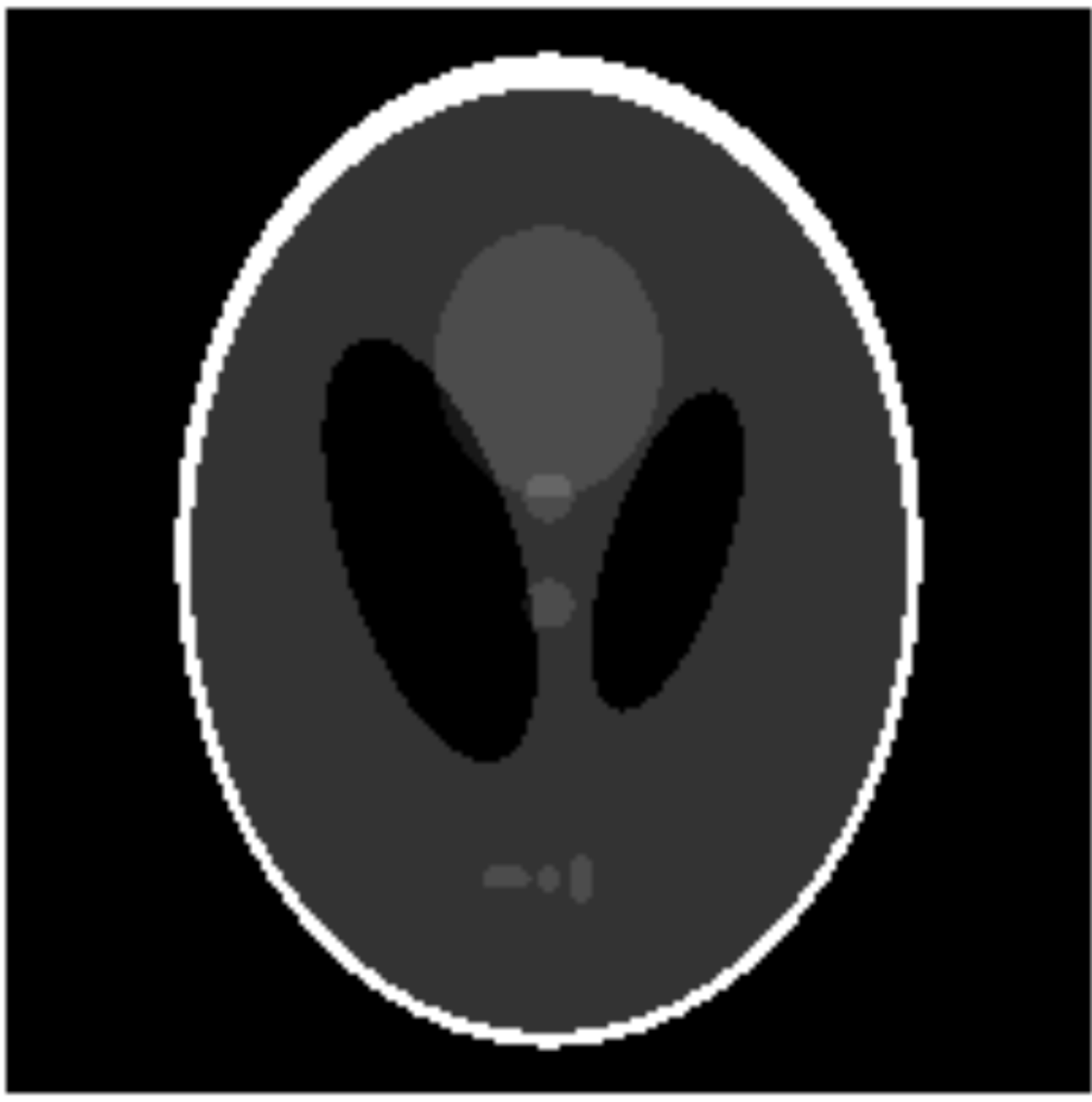}}\hspace{1cm}\hspace{1cm}
\subfigure[3D representation]{\includegraphics[width=3.5cm]{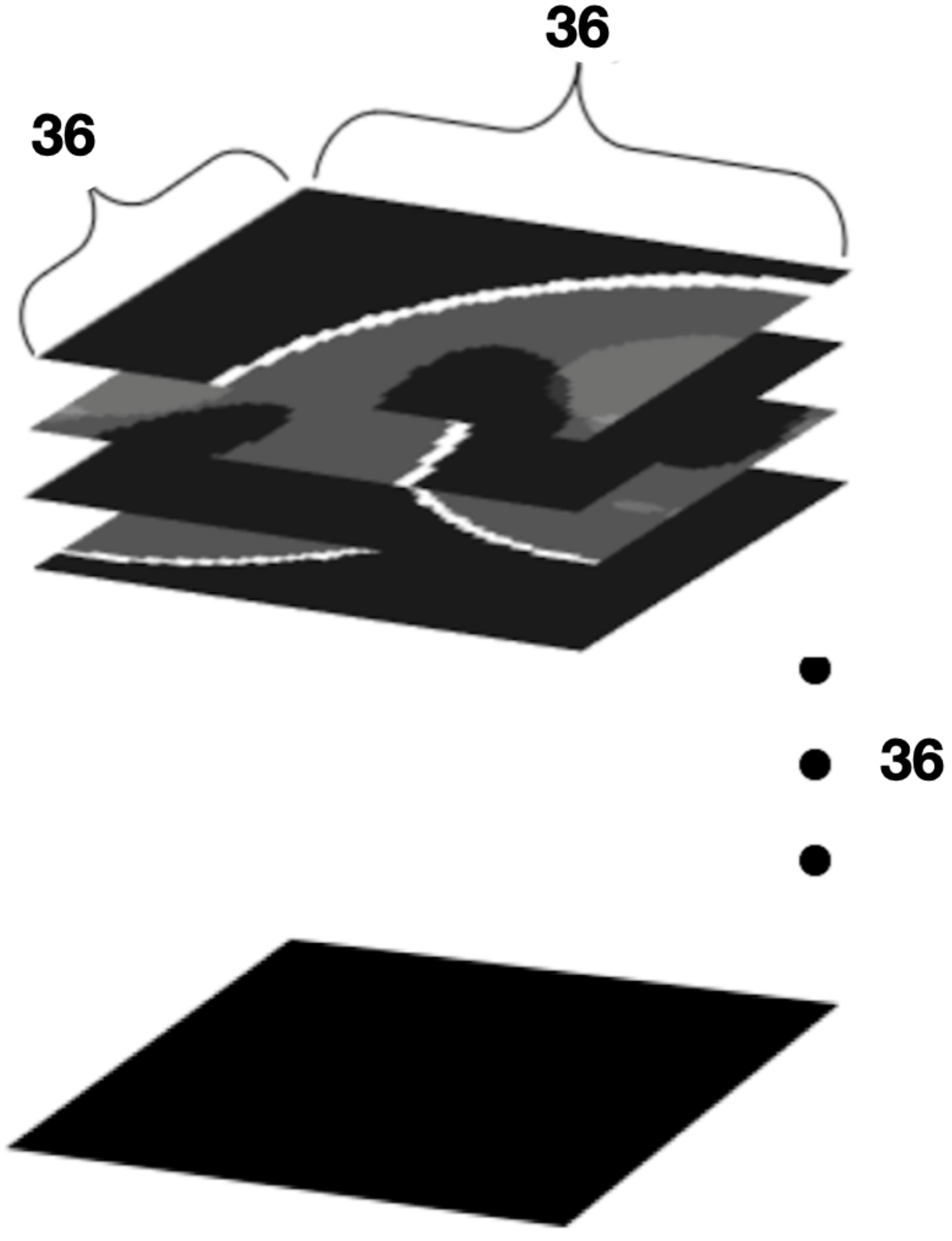}}
\caption{$216\times 216$ image  sliced and stacked unto a $36\times 36\times 36$ object. }
\label{fig:3D}
\end{figure}

Let us turn to the shot noise issue not addressed by the preceding uniqueness results.
At present, there are few theoretical results on noise robustness in phase retrieval except for
 simplified models \cite{Elser09}.
 
In practice, noise stability has as much to do with the reconstruction method as the information content of the given dataset. However, assessing and optimizing algorithms for 3D reconstruction from a large number of snapshots is by itself a challenging ongoing task  \cite{SSNR02}. Herein, 
we limit ourselves to testing the noise robustness of alternating projection  \cite{AP-phasing} (also known as Error-Reduction \cite{Fie82} or  Gerchberg-Saxton \cite{GS72} algorithm) in the case of weak phase objects to avoid the phase unwrapping problem altogether. 

{First  the ideal, noiseless detection process with a weak-phase object can be written as $b^2=|\cA f_*|^2$ in terms of a measurement matrix $\cA$ representing the process $
 f\stackrel{\cA}{\longrightarrow} 
(\widehat  f_\bt\star \widehat \mu)_{\bt\in \cT}. 
 $

In alternating projection (AP), the reconstruction procedure alternates between   the object-constraint projection $\cP_1=\cA\cA^\dagger$, where $A^\dagger$ is the pseudo-inverse of  $\cA$,  and the data-constrained  projection $
 \cP_2 h=b\odot \hbox{sgn}(h),$ 
 where $\hbox{sgn}(h)$ is the phase factor vector of $h$. 
 AP can be represented as $(\cP_1\cP_2)^kf_0, k\in \IN,$ with a random initialization $f_0$ \cite{AP-phasing}. 
 
Notably  AP is the gradient descent with unit step size 
for 
\beq
\label{105}
 \min_{u}\|b-|\cA u|\|
\eeq
which is a simplified asymptotic surrogate for maximizing  the Poisson  log-likelihood function \cite{DR-phasing}. 
 \commentout{
 \begin{algorithm}
\SetKwFunction{Round}{Round}
\textbf{Input:} $f^1,  \cT$ and  $b$.
\\
\textbf{Loop:}\\
\For{$k=1:k_{\textup{max}}-1$}
{
Compute  $\cA f^{k}=\big(\widehat\mu \star\widehat f_\bt^{k}\big)_{\bt\in \cT}$.\\
Set $h^{k}=\big(b\odot\sgn(\widehat\mu \star\widehat f_\bt^{k})\big)_{\bt\in \cT}$ and $g^{k}=\cQ^\dagger h^{k}$.\\
Compute $f^{k+1}=\cR^\dagger g^{k}$.  \\
 }
{\bf Output:} {$  \mbf^{k_{\textup{max}}}$.}

\caption{\textbf{AP method}}
\label{AP-alg}
\end{algorithm}
 }

\subsection{Noise-to-signal ratio (NSR)}

To introduce the Poisson noise into
 our set-up, let $\widetilde b^2$ be the Poisson random vectors with the mean $sb^2$ where
the adjustable scale factor $s>0$ represents the overall strength of object-radiation interaction. 

The noise level is measured by the  noise-to-signal ratio (NSR)
\beq
\hbox{NSR}:= {\hbox{\# total average non-signal photons}\over \hbox{ \# total signal photons}}. \label{55}
\eeq
Let $z=(z_j):=\widetilde b^2-sb^2$. The total average noise photon count is proportional to 
  \beq\label{L1}
\sum_j  \IE |z_j|&\hbox{or more conveniently}& \sum_j \sqrt{\IE|z_j|^2}=\Big\| \sqrt{\hbox{var}(\widetilde b^2)}\Big\|_1
\eeq
 where $\|\cdot\|_1$ denotes the L1-norm. 
 In other words, the NSR \eqref{55} can be conveniently calculated as 
\commentout{
     \beqn\label{nsr3}
  \hbox{NSR}_j:=  {\| (\IE(\widetilde b_j^2-{s} b_j^2)^2)^{1/2}\|_1\over s \|b_j\|^2}, \quad j=1,\cdots, m,
   \eeqn
 for each diffraction pattern and as 
 }
        \beq\label{nsr}
  \hbox{NSR}&:=& {\Big\| \sqrt{\hbox{var}(\widetilde b^2)}\Big\|_1\over  \|\IE(\widetilde b^2)\|_1}= {\|{b}\|_1\over \sqrt{s}\|b^2\|_1}.\label{56}
 % {\Big\| \sqrt{\IE(\widetilde b^2-{s} b^2)^2}\Big\|_1\over s \|b^2\|_1}= 
   \eeq

\subsection{Numerical results}
    \begin{figure}
  \centering
\subfigure[$R$ vs NSR $\in$ (0,1)]{\includegraphics[width=5cm]{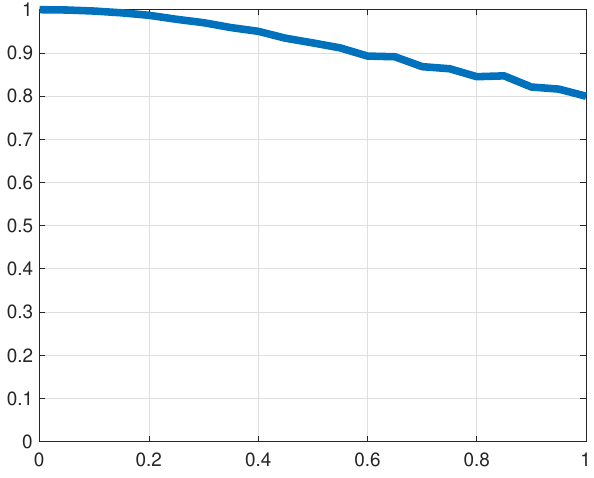}}\quad
\subfigure[NSR=0.5]{\includegraphics[width=4cm]{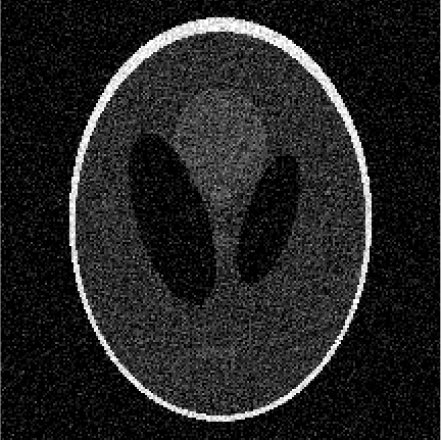}}\quad
\subfigure[NSR=1]{\includegraphics[width=4cm]{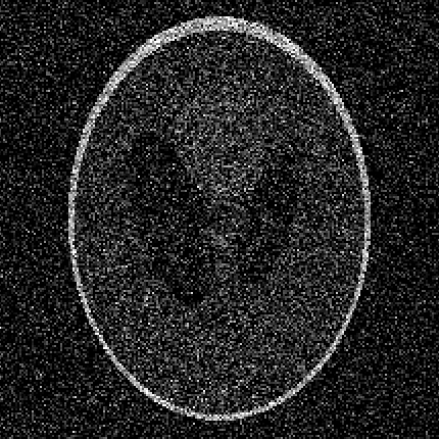}}
\caption{Randomly initialized AP reconstruction: (a) correlation vs NSR;  (b) flattened reconstruction at NSR =0.5 \& (c) NSR = 1. }
\label{fig:AP}
\end{figure}

In our simulations, we take the mask phase $ \phi$ to be an independent  uniform random variable over $[0,2\pi)$ for each pixel.

 We construct the real-valued 3D object from the $216\times 216$ phantom (Fig. \ref{fig:3D} (a)) by partitioning the phantom into 36 pieces, each of which is $36\times 36$ and stacking them into a $36\times 36\times 36$ cube (Fig. \ref{fig:3D}(b)).  
 %As indicated by the condition \ref{sector2} in Theorem \ref{thm:born1}, the vanishing voxels in the object shown in Figure \ref{fig:3D}
 This  is to facilitate the ``eye-ball" metric
for  qualitative evaluation of the reconstruction. For a quantitative metric,  we use the  absolute correlation
\[
R(f,f_*):={|\bar f\cdot f_*|\over \|f\| \|f_*\|}
\]
 between the true object $f_*$ and the reconstruction $f$.
 \commentout{ which is related to the relative error as \beqn
 \|f_*f_*^* -ff^*\|_{\rm F}
=\sqrt{2(\|f_*\|^4- |f^*f_*|^2)}
\eeqn
}

To avoid  the missing cone problem in tomography \cite{Frank06}, we use the random tilt scheme 
comprised  of the union of \eqref{x''}, \eqref{y''} and \eqref{z''}, in the form
 \beq
\label{tset}\cT&=&\{ \bt_i=(1,\alpha_i,\beta_i) \}_{i=1}^{n}\cup
\{ \bt_i=(\alpha_i,1, \beta_i) \}_{i=1+ n}^{2 n}\cup
\{ \bt_i=(\alpha_i, \beta_i,1) \}_{i=2 n+1}^{3n}\\
&&\hbox{with $\alpha_i,\beta_i,i=1,\cdots,3n,$ randomly chosen from
$(-1,1)$}.\nn
\eeq
which has  a significantly larger range of possibility than that portrayed in Figure \ref{fig:sphere}. 
  To reduce the burden of computer memory, we do not oversample the Fourier transform in our numerical simulation, i.e.
$b,\widetilde b \in \IR^{mp^2}$ with $m=3n$. Consequently, the total amount of measurement data is about $3/4$ of that assumed in Theorem \ref{tom2-weak}. 

To take advantage of the prior information that the object is real-valued,  we use $\cP_1=\cA\Re\cA^\dagger$ in AP reconstruction  where $\Re$ is the projection unto the real-part.  As shown in Figure \ref{fig:AP}, randomly initialized AP with the data set $\eqref{tset}$ is capable of handling high levels of noise.
}

The convergence rate, however, can be further improved by using more sophisticated algorithms  (see \cite{DR-phasing},  \cite{acta-phase} and references therein). 
When applicable, various sparsity priors  can enhance numerical reconstruction's  robustness to noise  \cite{sparsePR, GESPAR,sparse-Fienup, TV-tomo,AMP15}. 
Finally, the coded aperture itself needs not be known in advance and can be simultaneously calibrated by effective algorithms with a sufficiently large set $\cT$   \cite{blind-ptycho,phase-uncertain}.

%See, e.g. \cite{strong-phase11} and \cite{strong-phase22} for some recent studies of diffractive imaging with strong phase objects. 

\section{Discussion and extension}\label{sec:final}

\commentout{
Let us recap what has been accomplished above:
\begin{itemize}
\item Weak-phase objects: uniqueness for known orientations 

\begin{itemize}
\item 1 random mask: $n+1$ illuminations
\item 1 random mask + 1 uniform mask: $n$ illuminations.
\end{itemize}

\item Strong-phase objects $\rightarrow$ phase projection tomography $\sim$ phase unwrapping. 

\item Phase unwrapping uniqueness: $\cO(n)$ illuminations under Itoh condition. 

%\item Unknown orientations: 
%\begin{itemize}
%\item 2 coded apertures (coincidence measurement) $\rightarrow$ (phase) projections of partially known relative orientations;
%\item 1 coded +1 uncoded + beam splitter: (phase) projections of unknown orientations. 
%\end{itemize}
\end{itemize}
}

While our results are primarily aimed at diffraction tomography with known orientations, the idea of reducing phase retrieval to (phase) projection tomography by pair-wise measurements (Theorem \ref{thm1}, \ref{thm2}, \ref{thm3} and \ref{thm4}) have  potential applications to single-particle imaging which typically is subject to higher level of measurement uncertainties. 

Specifically  the proposed measurement schemes  embodied in Figure \ref{fig3} enable the reconstruction of (phase) projections
 in various unknown orientations $\bt$. 
Since it is easier to mitigate measurement uncertainties with projection data than diffraction data,  the measurement uncertainties in X-ray and electron experiments with small-sized objects such as nano-crystals and macromolecules may be handled by projection-based methods  (see  \cite{ variable-size14,variable-form21,serial15, fixed17,translational-disorder19}).  

\commentout{
While the schemes  in Theorem \ref{thm1}, \ref{thm3} and \ref{thm:born1} are numerically effective and stable (Section \ref{sec:num}), the schemes in Theorem \ref{thm2} and \ref{thm4} are expected to be less so
in view of the fact that the information overlap between projections in two directions is 
just a line in the Fourier domain (the common line). The remedy would be to include
more than two directions  in the reconstruction of (phase) projections.  How many directions are necessary? Theorem \ref{thm:born} suggests that the answer is $n+1$ directions. 
}

For example, for a weak phase object, classification and alignment can then be carried out with  single-particle cryo-EM methods such as  the common-line methods \cite{Crowther,common-line87, common-line11}, the maximum-likelihood  methods \cite{max-like98,max-like13} and the Bayesian methods \cite{Bayes02,Bayes12-1} based on projection data instead of diffraction patterns \cite{crypto-tomo03, common-line08, crypto-tomo09, Bayes-diffract09,manifold12}.

For a strong phase object, however, there remain several hurdles. The foremost  is developing effective numerical algorithms for 3D phase unwrapping which is not as well studied as 2D phase unwrapping \cite{2D-unwrap}.  

After the alignment of  the phase projections, it is tempting to apply the 2D phase unwrapping methods and try to recover the projection from the phase projection  in each direction. The projection in each direction, however, usually violate the 2D Itoh condition even when the 3D Itch condition holds.  If the 2D Itoh condition fails to hold at a large number of  pixels, then 2D phase unwrapping becomes more complicated, requiring additional prior information. Moreover, the unwrapped phases for different directions must be consistent with one another. Hence the 3D phase unwrapping problem should be approached  with all phase projections together instead of one direction at a time.

A similar approach called {\em ptychographic tomography}  has been proposed and studied numerically \cite{ptycho-tomo10, strong-phase11,ptycho-tomo16,ptycho-tomo17,ptycho-tomo18}.  The difference from  the present work is that in ptychography, instead of simultaneous pairwise measurements of the whole object,  multiple significantly overlapped diffraction patterns are measured in each direction by sequentially shifting the aperture over different parts of the object. As a consequence, ptychographic tomography  is limited to sizable objects capable of sustaining  multiple intense illuminations and hence not suitable for, e.g. single particle imaging.

\section*{Acknowledgments}
The research is supported by the Simons Foundation grant FDN 2019-24 and the NSF grant CCF-1934568. 
I thank Prof. Pengwen Chen of National Chung-Hsing University, Taiwan,  for producing Figure \ref{fig:AP}.

\begin{appendix}

\section{Proof of Theorem \ref{thm1}}\label{appA}

The following result is our basic tool. 
\begin{lem}\cite{unique}\label{prop1} Let $\mu$ be the phase mask (i.e. $\mu(\bn)=\exp[\im \phi(\bn)], \phi(\bn)\in \IR, \forall\bn$) with independent, continuous random variables $\phi(\bn)\in \IR$. If $e^{\im \kappa g_\bt}\odot \nu$   produces
the same diffraction pattern as $e^{\im \kappa f_\bt}\odot \mu$, then  for  some $\mbm_\bt \in \IZ^2, \theta_\bt\in \IR$
\beq
\label{1.70}
e^{\im \kappa g_\bt(\bn)}\nu(\bn)&=&\hbox{\rm either} \quad 
 e^{\im \theta_\bt} e^{\im\kappa f_\bt(\bn + \mbm_\bt  )}\mu(\bn + \mbm_\bt )\\
 && \quad
 \hbox{\rm or} \quad 
 e^{\im \theta_\bt} e^{-\im\kappa \overline{ f_\bt(-\bn +\mbm_\bt  )}}\overline{\mu(-\bn +\mbm_\bt  )}\nn
\eeq
for all $ \bn$.
\end{lem}
 Lemma  \ref{prop1} is a special case of the more general result in \cite{unique} which is not limited to
phase masks. Note that the statement holds for {\em any} real-valued continuous random variable $\phi(\bn)$.  
By more advanced techniques from algebraic geometry and probability, one can relax the conditions of continuity and independence on $\phi(\bn)$. 

After taking logarithm, \eqref{1.70} becomes 
\beq
\label{1.8'}
\kappa g_\bt(\bn)-\im \ln \nu(\bn)&=&\hbox{\rm either}\quad 
\theta_\bt+\kappa f_\bt(\bn + \mbm_\bt  )-\im \ln \mu(\bn + \mbm_\bt ) \\
&&\quad\hbox{or}\quad  \theta_\bt-\kappa  \overline{ f_\bt(-\bn +\mbm_\bt  )}-\im \ln \overline{\mu(-\bn +\mbm_\bt) }\nn
\eeq
modulo $2\pi.$

If $\mu$ is completely known, i.e. $\nu=\mu$, then \eqref{1.8'} becomes 
\beq
\label{1.8}
\kappa g_\bt(\bn)-\im \ln \mu(\bn)&=&\hbox{\rm either}\quad 
\theta_\bt+\kappa f_\bt(\bn + \mbm_\bt  )-\im \ln \mu(\bn + \mbm_\bt ) \\
&&\quad\hbox{\rm or}\quad  \theta_\bt-\kappa  \overline{f_\bt(-\bn +\mbm_\bt  )}-\im \ln \overline{\mu(-\bn +\mbm_\bt) }\nn
\eeq
modulo $2\pi$.

Since both diffraction patterns are from the same snapshot, we can reset the object frame so that 
$\bl_\bt=0$. 

Suppose the first alternative in \eqref{1.8} holds with $\mbm_\bt\not = 0$. By \eqref{43},  $e^{\im\kappa f_\bt}$ and  $e^{\im\kappa g_\bt}$ have  the same  autocorrelation function and hence
\beqn
\sum_\bn e^{\im \kappa (f_\bt(\bn+\bk)-\overline{f_\bt(\bn)})}&=& \sum_\bn e^{\im \kappa (f_\bt (\bn+\mbm_\bt+\bk)-\overline{f_\bt(\bn+\mbm_\bt)})}e^{\im (\phi(\bn+\mbm_\bt+\bk)-\phi(\bn+\bk))} e^{-\im (\phi(\bn+\mbm_\bt)-\phi(\bn))}
\eeqn
for all $\bk$, or equivalently
\beq
\label{2.1}
\lefteqn{\sum_\bn e^{\im \kappa (f_\bt(\bn+\mbm_\bt+\bk)-\overline{f_\bt(\bn+\mbm_\bt)})}}\\
&=& \sum_\bn e^{\im \kappa (f_\bt (\bn+\mbm_\bt+\bk)-\overline{f_\bt(\bn+\mbm_\bt)})}e^{\im (\phi(\bn+\mbm_\bt+\bk)-\phi(\bn+\mbm_\bt)-\phi(\bn+\bk)+\phi(\bn))}\nn
\eeq
by change of index, $\bn\to \bn+\mbm_\bt$, on the left hand side of equation. 
Define
\[
\Delta_\bk f_\bt(\bn+\mbm_\bt):=f_\bt(\bn+\mbm_\bt+\bk)-\overline{f_\bt(\bn+\mbm_\bt)}
\]
and rewrite  \eqref{2.1} as
\beq\label{2.2}
0&=&\sum_{\bn} \lt[e^{\im (\phi(\bn+\mbm_\bt+\bk)-\phi(\bn+\mbm_\bt)-\phi(\bn+\bk)+\phi(\bn))}-1\rt]e^{\im \kappa \Delta_\bk f_\bt (\bn+\mbm_\bt)},
\eeq
for all $\bk$. 
We want to show that the probability of the event \eqref{2.2} is zero. 

Let us consider those summands in \eqref{2.2}, for any fixed $\bk$,  that share
a common $\phi(\bl)$, for any fixed $\bl$,  in the expression. Clearly, there are at most four such terms:
\beq
&& \lt[e^{\im (\phi(\bl)-\phi(\bl-\bk)-\phi(\bl-\mbm_\bt)+\phi(\bl-\bk-\mbm_\bt))}-1\rt]e^{\im \kappa \Delta_\bk f_\bt (\bl-\bk)}\label{2.5}\\
&+& \lt[e^{\im (\phi(\bl+\bk)-\phi(\bl)-\phi(\bl+\bk-\mbm_\bt)+\phi(\bl-\mbm_\bt))}-1\rt]e^{\im \kappa \Delta_\bk f_\bt (\bl)}\nn\\
&+&\lt[e^{\im (\phi(\bl+\mbm_\bt)-\phi(\bl-\bk+\mbm_\bt)-\phi(\bl)+\phi(\bl-\bk))}-1\rt]e^{\im \kappa \Delta_\bk f_\bt (\bl-\bk+\mbm_\bt)} \nn\\
&+& \lt[e^{\im (\phi(\bl+\bk+\mbm_\bt)-\phi(\bl+\mbm_\bt)-\phi(\bl+\bk)+\phi(\bl))}-1\rt]e^{\im \kappa \Delta_\bk f_\bt (\bl+\mbm_\bt)}. \nn
\eeq
Since the continuous random variable $\phi(\bl)$ does not appear in other summands and hence is independent of them, \eqref{2.2} implies that \eqref{2.5} (and the rest of \eqref{2.2})  must vanish almost surely.

For $\bk$ that are linearly independent of $\mbm_\bt$, the four independent random variables
\beq
\label{51-2}
\phi(\bl-\bk-\mbm_\bt),\quad \phi(\bl+\bk-\mbm_\bt),\quad \phi(\bl-\bk+\mbm_\bt),\quad \phi(\bl+\bk+\mbm_\bt)
\eeq
appear separately in  exactly one summand in  \eqref{2.5}. Consequently, \eqref{2.5} (and hence \eqref{2.2}) almost surely does not vanishes 
for $\bk$ that are linearly independent of $\mbm_\bt$.

On the other hand, if  $\bk$ is parallel to $\mbm_{\bt}\neq 0$, then for  any 
\beqn
 \bk\not \in \{ \pm \mbm_{\bt},\pm \half  \mbm_{\bt},\pm 2  \mbm_{\bt}\}
 \eeqn
  the four terms in \eqref{51-2}
appear  separately in exactly one summand in \eqref{2.5}.  Consequently, \eqref{2.5} almost surely does not vanishes.

Thus whenever the first alternative in \eqref{1.8} holds true,  we have  $\mbm_\bt=0$ and $e^{\im \kappa g_\bt}=e^{\im\theta_\bt} e^{\im\kappa f_\bt}$ for some constant $\theta_\bt$ independent of the grid point. 

Next, we show that the second alternative in \eqref{1.70} is false.
By \eqref{43}, we have
\beqn
\lefteqn{\sum_\bn e^{\im \kappa (f_\bt(\bn+\bk)-\overline{f_\bt(\bn)})}}\\&=& \sum_\bn e^{-\im \kappa (\overline{f_\bt (-\bn+\mbm_\bt+\bk)}-f_\bt(-\bn+\mbm_\bt))}e^{-\im (\phi(-\bn+\mbm_\bt+\bk)+\phi(\bn+\bk))} e^{\im (\phi(-\bn+\mbm_\bt)+\phi(\bn))}
\eeqn
for all $\bk$, or equivalently
\beq
\label{2.1'}\lefteqn{\sum_\bn e^{\im \kappa (f_\bt(\bn+\mbm_\bt+\bk)-\overline{f_\bt(\bn+\mbm_\bt)})}}\\
&=&\sum_\bn e^{-\im \kappa (\overline{f_\bt (-\bn+\mbm_\bt+\bk)}-f_\bt(-\bn+\mbm_\bt))}e^{-\im (\phi(-\bn+\mbm_\bt+\bk)+\phi(\bn+\bk))} e^{\im (\phi(-\bn+\mbm_\bt)+\phi(\bn))}
\nn
\eeq
by change of index, $\bn\to -\bn+\mbm_\bt$, on the left hand side of equation. 
With 
\[
\Delta_\bk \overline{f_\bt(-\bn+\mbm_\bt)}:=\overline{f_\bt(-\bn+\mbm_\bt+\bk)}-f_\bt(-\bn+\mbm_\bt)
\]
we rewrite  \eqref{2.1'} as
\beq\label{2.2'}
0&=&\sum_\bn \lt[e^{\im (-\phi(-\bn+\mbm_\bt+\bk)+\phi(-\bn+\mbm_\bt)-\phi(\bn+\bk)+\phi(\bn))}-1\rt]e^{-\im \kappa \Delta_\bk \overline{f_\bt (-\bn+\mbm_\bt)}}
\eeq
for all $\bk$. 
We want to show that the  right hand side of \eqref{2.2'} almost surely does not vanish for any $\bk$. 

As before, consider those summands in \eqref{2.2'}, for any fixed $\bk$,  that share
a common $\phi(\bl)$, for any fixed $\bl$,  in the expression. Clearly, there are at most four such terms:
\beq
&& \lt[e^{\im (-\phi(\bl)+\phi(\bl+\bk)-\phi(-\bl+\mbm_\bt)+\phi(-\bl-\bk+\mbm_\bt))}-1\rt] e^{-\im \kappa \Delta_\bk \overline{ f_\bt (\bl)}}\label{2.5'}\\
&& +\lt[e^{\im (-\phi(\bl-\bk)+\phi(\bl)-\phi(-\bl+\bk+\mbm_\bt)+\phi(-\bl+\mbm_\bt))}-1\rt] e^{-\im \kappa \Delta_\bk \overline{ f_\bt (\bl-\bk)}}\nn\\
&&+\lt[e^{\im (-\phi(-\bl+\mbm_\bt)+\phi(-\bl+\bk+\mbm_\bt)-\phi(\bl)+\phi(\bl-\bk))}-1\rt]e^{-\im \kappa \Delta_\bk \overline{ f_\bt (-\bl+\mbm_\bt)}}\nn\\
&& +\lt[e^{\im (-\phi(-\bl-\bk+\mbm_\bt)+\phi(-\bl+\mbm_\bt)-\phi(\bl+\bk)+\phi(\bl))}-1\rt]e^{-\im \kappa \Delta_\bk \overline{ f_\bt (-\bl-\bk+\mbm_\bt)}}\nn
\eeq
With $\phi(\bl)$ appearing in no other terms, \eqref{2.2'} implies that \eqref{2.5'} must vanish almost surely.

Some observations are in order. First, both $\phi(\bl)$ and $\phi(-\bl+\mbm_\bt)$ appear exactly once in each summand in \eqref{87}. Second, the following pairings of the other phases
\beqn
&&\{\phi(\bl+\bk), \phi(-\bl-\bk+\mbm_\bt)\},\quad \{\phi(\bl-\bk), \phi(-\bl+\bk+\mbm_\bt)\}
%&&\{\phi(-\bl+\bk+\mbm_\bt), \phi(\bl-\bk)\},\quad\{\phi(-\bl-\bk+\mbm_\bt), \phi(\bl+\bk)\}
\eeqn
also appear exactly twice in \eqref{87}. As long as $\bk\not=0$ and $2\bl\neq \mbm_\bt$, these two pairs are not identical and hence 
\beq
0&=& \lt[e^{\im (-\phi(\bl)+\phi(\bl+\bk)-\phi(-\bl+\mbm_\bt)+\phi(-\bl-\bk+\mbm_\bt))}-1\rt]e^{-\im \kappa \Delta_\bk \overline{ f_\bt (\bl)}}\label{88'}\\
&& +\overline{\lt[e^{\im (\phi(-\bl-\bk+\mbm_\bt)-\phi(-\bl+\mbm_\bt)+\phi(\bl+\bk)-\phi(\bl))}-1\rt]}e^{-\im \kappa \Delta_\bk \overline{ f_\bt (-\bl-\bk+\mbm_\bt)}}\nn\\
\label{89'}0&=&\lt[e^{\im (-\phi(\bl-\bk)+\phi(\bl)-\phi(-\bl+\bk+\mbm_\bt)+\phi(-\bl+\mbm_\bt))}-1\rt] e^{-\im \kappa \Delta_\bk \overline{ f_\bt (\bl-\bk)}}\nn\\
&&+\overline{\lt[e^{\im (\phi(-\bl+\mbm_\bt)-\phi(-\bl+\bk+\mbm_\bt)+\phi(\bl)-\phi(\bl-\bk))}-1\rt]}e^{-\im \kappa \Delta_\bk \overline{ f_\bt (-\bl+\mbm_\bt)}}\nn
\eeq
both of which are almost surely false because the two factors
\[
 \lt[e^{\im (-\phi(\bl)+\phi(\bl+\bk)-\phi(-\bl+\mbm_\bt)+\phi(-\bl-\bk+\mbm_\bt))}-1\rt],\quad \lt[e^{\im (-\phi(\bl-\bk)+\phi(\bl)-\phi(-\bl+\bk+\mbm_\bt)+\phi(-\bl+\mbm_\bt))}-1\rt]
 \]
 differ with their complex conjugates  in a random manner independently from $f_\bt$. 

Therefore,  the second alternative in \eqref{1.8} almost surely does not hold true.

In summary, the first alternative in \eqref{1.8} holds with $\mbm_\bt= 0$, namely
\beqn
 \kappa g_\bt(\bn)&=& \theta_\bt+ \kappa f_\bt(\bn)\quad\mod 2\pi
 \eeqn
 almost surely. 
 
% \subsection{Projection support constraint}\label{sec:tight}
The actual support of the projections \eqref{2.8}-\eqref{2.10}  for $0\le \alpha, \beta \le 1$ and odd integer $n$, for example, is contained in
\beq\label{support}
\bigcup_{i\in \IZ_n}(\IZ_n-\lfloor\alpha i\rfloor) \times (\IZ_n-\lfloor\beta i\rfloor)
\eeq 
where $\lfloor\cdot\rfloor$ denotes the floor function. In turn, the set in \eqref{support} is a subset of
\beq
\label{support2}\lt\{\bigcup_{i\in \IZ_n}(\IZ_n-\lfloor\alpha i\rfloor) \rt\}\times \lt\{\bigcup_{i\in \IZ_n} (\IZ_n-\lfloor\beta i\rfloor)\rt\}&=&\IZ_{\ell_\alpha}\times \IZ_{\ell_\beta}
\eeq
where 
\beq\label{zero4}
\ell_\alpha=2 \cdot \lfloor\half (1+|\alpha|) {(n-1)}\rfloor+1,\quad \ell_\beta=2\cdot \lfloor\half (1+|\beta|) {(n-1)}\rfloor+1.
\eeq
\commentout{
\beq
\label{zero}
f_{(1,\alpha,\beta)}(c_1,c_2)=f_{(\alpha,1,\beta)}(c_1,c_2)= f_{(\alpha,\beta,1)}(c_1,c_2)=0
\eeq
if $c_1, c_2$ are {\em not}  elements
of the sub-lattice
\beq
\label{zero2}
&& \lb -\lfloor \half(1+|\alpha|) {(n-1)}\rfloor, \lfloor\half (1+|\alpha|) {(n-1)}\rfloor\rb\\ 
& \times & \lb -\lfloor \half(1+|\beta|) {(n-1)}\rfloor, \lfloor\half (1+|\beta|){(n-1)}\rfloor\rb\nn
\eeq
 In other words, 
\beq
\label{zero3}
\supp[f_{(1,\alpha,\beta)}], \,\,\supp[f_{(\alpha,1,\beta)}], \,\,\supp[f_{(\alpha,\beta,1)}]\subset\IZ_{\ell_\alpha}\times \IZ_{\ell_\beta}
\eeq
}
The same support constraint $\IZ_{\ell_\alpha}\times \IZ_{\ell_\beta}$ with \eqref{zero4} applies to the case with $|\alpha|,|\beta|\le 1$ and odd integer $n$.

Now that $g_\bt(\bn)=f_\bt(\bn)=0$ for $\bn\in \IZ_p^2\setminus \IZ_{\ell_\alpha}\times \IZ_{\ell_\beta} $ for $\bt=(1,\alpha,\beta)$ due to the support constraint \eqref{support2}-\eqref{zero4}, so
$\theta_\bt$ must be an integer multiple of $2\pi$, i.e. $e^{\im \kappa g_\bt}=e^{\im\kappa f_\bt}$ almost surely. The proof is complete.

\section{Proof of Theorem \ref{thm2}}\label{appB}
Let
\beqn
f^*_{\bt'}(\bn)=f_{\bt'} (\bn+\bl_{\bt'}),&&f^*_\bt(\bn)=f_\bt (\bn+\bl_\bt)\\
g^*_{\bt'}(\bn)=g_{\bt'} (\bn+\bl_{\bt'}), && g^*_\bt(\bn)=g_\bt (\bn+\bl_\bt)
\eeqn
for some $\bl_{\bt'},\bl_\bt$. 
 
  Suppose that  the first alternative in \eqref{1.8} holds true for $\bt'$,  i.e.
\beq
 \kappa g^*_{\bt'}(\bn)-\im \ln \mu(\bn)&=&
\theta_{\bt'}+\kappa f^*_{\bt'}(\bn + \mbm_{\bt'} )-\im \ln \mu(\bn + \mbm_{\bt'} ) 
\eeq
modulo $2\pi$, or equivalently
\beq\label{1.10}
\kappa g^*_{\bt'}(\bn)+\phi(\bn)&=&
\theta_{\bt'}+\kappa f^*_{\bt'}(\bn + \mbm_{\bt'}  )+\phi(\bn + \mbm_{\bt '}) + \kappa
h_{\bt'}(\bn), 
\eeq
where $
h_{\bt'}(\bn)
$ is an integer multiple of $2\pi/\kappa$ for every $\bn$, 
implying
\beq\label{2.20'}
\kappa\widehat g^*_{\bt'}(\bk)+\widehat\phi(\bk)= e^{\im 2\pi\mbm_{\bt'}\cdot\bk/p} (\kappa\widehat f^*_{\bt'}(\bk)+\widehat\phi(\bk))+p^2\theta_{\bt'}D^2_p(\bk)+\kappa \widehat h_{\bt'}(\bk). 
\eeq

First we show that the second alternative in \eqref{1.8} can not hold  for $\bt\not=\bt'$. 
Suppose otherwise, i.e.
\beq\label{101-2}
\kappa g^*_\bt(\bn)+\phi(\bn)&=&
 \theta_\bt+\kappa  \overline{ f_\bt(-\bn +\mbm_\bt  )}-\phi(-\bn +\mbm_\bt)  \quad  \hbox{\rm mod\,\,} 2\pi .
\eeq
implying  
\beq\label{2.20}
\kappa\widehat g^*_\bt(\bk)+\widehat\phi(\bk)&=& e^{-\im 2\pi \mbm_\bt\cdot\bk/p}(\kappa \overline{\widehat f^*_{\bt}(\bk)}-\widehat\phi(-\bk))+p^2\theta_\bt D^2_p(\bk)+\kappa\widehat h_\bt (\bk)
\eeq
where $
h_{\bt}(\bn)
$ is an integer multiple of $2\pi$ for every $\bn$.  

Since  \[
\widehat g_{\bt'}(\bk')=\widehat g_\bt(\bk),\quad \widehat f_{\bt'}(\bk')=\widehat f_\bt(\bk),\quad (\bk,\bk')\in  C_{\bt,\bt'},\]
we have 
\beq
\label{70'}
\widehat g^*_{\bt'}(\bk') e^{-\im 2\pi\bk'\cdot\bl_{\bt'}/p}&=&\widehat g^*_{\bt}(\bk) e^{-\im 2\pi\bk\cdot\bl_{\bt}/p} \\
\widehat f^*_{\bt'}(\bk') e^{-\im 2\pi\bk'\cdot\bl_{\bt'}/p}&=&\widehat f^*_{\bt}(\bk) e^{-\im 2\pi\bk\cdot\bl_{\bt}/p}. 
\eeq

Eq. \eqref{70'}, together with \eqref{2.20'} and \eqref{2.20}, imply that for  $ (\bk,\bk')\in C_{\bt,\bt'}$
\beqn
&& e^{-\im 2\pi\bk'\cdot\bl_{\bt'}/p}\lt[e^{\im 2\pi\mbm_{\bt'}\cdot\bk'/p} (\kappa\widehat f^*_{\bt'}(\bk')+\widehat\phi(\bk'))+p^2\theta_{\bt'}D^2_p(\bk')+\kappa \widehat h_{\bt'}(\bk')-\widehat\phi(\bk')\rt]\\
&=&e^{-\im 2\pi\bk\cdot\bl_{\bt}/p} \lt[e^{-\im 2\pi \mbm_\bt\cdot\bk/p}(\kappa \overline{\widehat f^*_{\bt}(\bk)}-\widehat\phi(-\bk))+p^2\theta_\bt D^2_p(\bk)+\kappa\widehat h_\bt (\bk)-\widehat\phi(\bk)\rt], \nn
\eeqn
and hence 
\beq
 \label{a.45}&&\lt[e^{\im 2\pi(\mbm_{\bt'}-\bl_{\bt'})\cdot\bk'/p}-e^{-\im 2\pi\bk'\cdot\bl_{\bt'}/p}\rt] \widehat\phi(\bk')+ e^{-\im 2\pi (\mbm_\bt+\bl_\bt)\cdot\bk/p} \widehat\phi(-\bk)+e^{-\im 2\pi\bl_\bt\cdot\bk/p}\widehat\phi(\bk)\\
&=&-\kappa e^{\im 2\pi(\mbm_{\bt'}-\bl_{\bt'})\cdot\bk'/p} \widehat f^*_{\bt'}(\bk')+\kappa e^{-\im 2\pi (\mbm_\bt+\bl_\bt)\cdot\bk/p} \overline{\widehat f^*_{\bt}(\bk)}+\kappa e^{-\im 2\pi\bk\cdot\bl_{\bt}/p}\widehat h_\bt (\bk)\nn\\
&& -\kappa e^{-\im 2\pi\bk'\cdot\bl_{\bt'}/p}\widehat h_{\bt'}(\bk')-p^2e^{-\im 2\pi\bk'\cdot\bl_{\bt'}/p}\theta_{\bt'}D^2_p(\bk')+p^2e^{-\im 2\pi\bk\cdot\bl_{\bt}/p}\theta_\bt D^2_p(\bk).  \nn
\eeq 
The left hand side of \eqref{a.45} is a sum of independent, continuous random variables while  the right hand side is a discrete random variable for a fixed $f$. Therefore, \eqref{a.45} is false almost surely. 

This leaves  the first alternative of \eqref{1.8} the only viable alternative for $\bt$, 
i.e. 
\beq\label{70}
\kappa g^*_\bt(\bn)+\phi(\bn)&=&
\theta_\bt+\kappa f^*_\bt(\bn + \mbm_\bt  )+\phi(\bn + \mbm_\bt )\quad \hbox{\rm mod}\,\, 2\pi,
\eeq
for some $\mbm_\bt$, and hence 
\beqn
&&{e^{\im 2\pi(\mbm_{\bt}-\bl_{\bt})\cdot\bk/p} (\kappa\widehat f^*_{\bt}(\bk)+\widehat\phi(\bk))-e^{-\im 2\pi \bl_{\bt}\cdot\bk/p} \widehat\phi(\bk)+p^2\theta_{\bt}e^{-\im 2\pi\bk\cdot\bl_{\bt}/p}D^2_p(\bk)+\kappa e^{-\im 2\pi\bk\cdot\bl_{\bt}/p}\widehat h_{\bt}(\bk)}\\
&=&e^{\im 2\pi(\mbm_{\bt'}-\bl_{\bt'})\cdot\bk'/p} (\kappa\widehat f^*_{\bt'}(\bk')+\widehat\phi(\bk'))-e^{-\im 2\pi \bl_{\bt'}\cdot\bk'/p} \widehat\phi(\bk')\nn\\
&&+p^2\theta_{\bt'}e^{-\im 2\pi\bk'\cdot\bl_{\bt'}/p}D^2_p(\bk')+\kappa e^{-\im 2\pi\bk'\cdot\bl_{\bt'}/p}\widehat h_{\bt'}(\bk')\nn
\eeqn
for $ (\bk,\bk')\in C_{\bt,\bt'}$. 
Reorganizing the above equation, we have
\beq
&&(e^{\im 2\pi\mbm_{\bt}\cdot\bk/p}-1)e^{-\im 2\pi \bl_{\bt}\cdot\bk/p} \widehat \phi(\bk)+(1-e^{\im 2\pi\mbm_{\bt'}\cdot\bk'/p} )e^{-\im 2\pi \bl_{\bt'}\cdot\bk'/p} \widehat\phi(\bk')\label{2.28}\\
&=&e^{\im 2\pi(\mbm_{\bt'}-\bl_{\bt'})\cdot\bk'/p} \kappa\widehat f^*_{\bt'}(\bk')-e^{\im 2\pi(\mbm_{\bt}-\bl_{\bt})\cdot\bk/p}\kappa \widehat f^*_{\bt}(\bk)+p^2\theta_{\bt'}e^{-\im 2\pi\bk'\cdot\bl_{\bt'}/p}D^2_p(\bk')\nn\\
&&-p^2\theta_{\bt}e^{-\im 2\pi\bk\cdot\bl_{\bt}/p}D^2_p(\bk)+\kappa e^{-\im 2\pi\bk'\cdot\bl_{\bt'}/p}\widehat h_{\bt'}(\bk')-\kappa e^{-\im 2\pi\bk\cdot\bl_{\bt}/p}\widehat h_{\bt}(\bk)\nn
\eeq
for $ (\bk,\bk')\in C_{\bt,\bt'}$. 
If, for some $ (\bk,\bk')\in C_{\bt,\bt'}$, 
\beq
\label{line}
(e^{\im 2\pi\mbm_{\bt}\cdot\bk/p}-1)e^{-\im 2\pi \bl_{\bt}\cdot\bk/p} \neq 0\quad\hbox{or}\quad \quad e^{-\im 2\pi \bl_{\bt'}\cdot\bk/p}(1-e^{\im 2\pi\mbm_{\bt'}\cdot\bk'/p} )\neq 0,
\eeq
then the left hand side of \eqref{2.28} is a continuous random variable  while  the right hand side takes value in a discrete set for a given $f$. This is a contradiction with probability one,
implying for all $ (\bk,\bk')\in C_{\bt,\bt'}$
\beqn
(e^{\im 2\pi\mbm_{\bt}\cdot\bk/p}-1)e^{-\im 2\pi \bl_{\bt}\cdot\bk/p} = e^{-\im 2\pi \bl_{\bt'}\cdot\bk'/p}(e^{\im 2\pi\mbm_{\bt'}\cdot\bk'/p}-1)=0
\eeqn
and consequently, \beq
\label{line2}
\mbm_{\bt}\cdot \bk= \mbm_{\bt'}\cdot \bk'=0 \quad \mod p.\eeq
By assumption, some  $(\bk,\bk')\in C_{\bt,\bt'}$ have components whose ratio is not a fraction over $\IZ_p$, \eqref{line2} can not hold true for $\mbm_\bt, \mbm_{\bt'}\in \IZ_p^2$. 

By \eqref{70} and \eqref{1.10}, 
\beqn
\kappa g^*_\bt(\bn)&=&
\theta_\bt+\kappa f^*_\bt(\bn)\quad \hbox{\rm mod}\,\, 2\pi,\\
\kappa g^*_{\bt'}(\bn)&=&
\theta_{\bt'}+\kappa f^*_{\bt'}(\bn ) \quad \hbox{\rm mod}\,\, 2\pi,
\eeqn
which imply, by the set-up of zero-padding, $\theta_\bt=\theta_{\bt'}=0$ and hence  
\beqn
\kappa g^*_\bt(\bn)=
\kappa f^*_\bt(\bn),&&
\kappa g^*_{\bt'}(\bn)=
\kappa f^*_{\bt'}(\bn ), 
\eeqn
for all $\bn\in \IZ_p^2$.

Let us rule out  the remaining undesirable alternative: For all $\bn\in \IZ_p^2$,
\beq
\kappa g^*_\bt(\bn)+\phi(\bn)&=&
 \theta_\bt+\kappa  \overline{ f^*_\bt(-\bn +\mbm_\bt  )}-\phi(-\bn +\mbm_\bt)  \quad  \hbox{\rm mod\,\,} 2\pi\\
\kappa g^*_{\bt'}(\bn)+\phi(\bn)&=&
 \theta_{\bt'}+\kappa  \overline{ f^*_{\bt'}(-\bn +\mbm_{\bt'} )}-\phi(-\bn +\mbm_{\bt'})  \quad  \hbox{\rm mod\,\,} 2\pi .
\eeq
For  $ (\bk,\bk')\in C_{\bt,\bt'}, $ we have
\beqn
&& e^{-\im 2\pi\bk'\cdot\bl_{\bt'}/p}\lt[e^{-\im 2\pi\mbm_{\bt'}\cdot\bk'/p} (\kappa\overline{\widehat f^*_{\bt'}(\bk')}-\widehat\phi(-\bk'))+p^2\theta_{\bt'}D^2_p(\bk')-\widehat\phi(\bk')+\kappa\widehat h_{\bt'} (\bk')\rt]\\
&=&e^{-\im 2\pi\bk\cdot\bl_{\bt}/p} \lt[e^{-\im 2\pi \mbm_\bt\cdot\bk/p}(\kappa \overline{\widehat f^*_{\bt}(\bk)}-\widehat\phi(-\bk))+p^2\theta_\bt D^2_p(\bk)-\widehat\phi(\bk)+\kappa\widehat h_{\bt} (\bk)\rt], \nn
\eeqn
and hence
\beqn&&e^{-\im 2\pi (\mbm_\bt+\bl_\bt)\cdot\bk/p} \widehat\phi(-\bk)+e^{-\im 2\pi\bl_\bt\cdot\bk/p}\widehat\phi(\bk)-e^{-\im 2\pi (\mbm_{\bt'}+\bl_{\bt'})\cdot\bk'/p} \widehat\phi(-\bk')-e^{-\im 2\pi\bl_{\bt'}\cdot\bk'/p}\widehat\phi(\bk') \\
&=&-\kappa e^{-\im 2\pi(\mbm_{\bt'}+\bl_{\bt'})\cdot\bk'/p} \overline{\widehat f^*_{\bt'}(\bk')}+\kappa e^{-\im 2\pi (\mbm_\bt+\bl_\bt)\cdot\bk/p} \overline{\widehat f^*_{\bt}(\bk)}+\kappa e^{-\im 2\pi\bk\cdot\bl_{\bt}/p}\widehat h_\bt (\bk)\nn\\
&& -\kappa e^{-\im 2\pi\bk'\cdot\bl_{\bt'}/p}\widehat h_{\bt'}(\bk')-p^2e^{-\im 2\pi\bk'\cdot\bl_{\bt'}/p}\theta_{\bt'}D^2_p(\bk')+p^2e^{-\im 2\pi\bk\cdot\bl_{\bt}/p}\theta_\bt D^2_p(\bk).  \nn
\eeqn
The left hand side is a continuous random variable while the right hand side is a discrete random variable for a fixed $f$. This is impossible and hence the undesirable alternative is ruled out. The proof is complete.

\section{Proof of Theorem \ref{thm3}}\label{appC}

The proof of Theorem \ref{thm3} is analogous to that of Theorem \ref{thm1}, except with the additional complication of possible vanishing of the object function under the Born approximation. 

Similar to Lemma \ref{prop1}, for the diffraction pattern given by \eqref{born-pattern} we have the following
characterization. 
\begin{lem}\cite{unique}\label{prop2}  Let $\mu=e^{\im \phi(\bn)}$ with independent, continuous random variables $\phi(\bn)\in \IR$.  Suppose that $\supp(f_\bt)$ is not a subset of a line and another masked object projection $\nu \odot g_\bt$  produces
the same diffraction pattern as $ \mu \odot f_\bt$. Then for some $\bp$ and $\theta_\bt,$
\beq
\label{r1}
g_\bt(\bn)\nu(\bn)&=&\hbox{\rm either} \quad 
 e^{\im \theta_\bt} f_\bt(\bn + \mbm_\bt  )\mu(\bn + \mbm_\bt )\\
 && \quad
 \hbox{\rm or} \quad 
 e^{\im \theta_\bt} \overline{ f_\bt(-\bn +\mbm_\bt  )}\overline{\mu(-\bn +\mbm_\bt  )}\nn
\eeq
for all $ \bn$.
\end{lem}

If $\mu$ is completely known, then $\nu=\mu$ and \eqref{r1} becomes 
\beq
\label{1.8-2}
g_\bt(\bn)\mu(\bn)&=&\hbox{\rm either}\quad 
e^{\im \theta_\bt} f_\bt(\bn + \mbm_\bt  )\mu(\bn + \mbm_\bt ) \\
&&\quad\hbox{\rm or}\quad  e^{\im \theta_\bt} \overline{f_\bt(-\bn +\mbm_\bt  )\mu(-\bn +\mbm_\bt  )}. \nn
\eeq

First suppose that the first alternative in \eqref{1.8-2} and we want to show that $\mbm_\bt=0$, which then implies that $g_\bt(\cdot)=e^{\im\theta_\bt} f_\bt(\cdot)$.

Equality of the uncoded diffraction \eqref{44} implies that the autocorrelation of $g_\bt$ equals that of $f_\bt$ and hence by \eqref{1.8-2}
\beqn
\sum_{\bn\in\IZ_p^2} f_\bt(\bn+\bk)\overline{f_\bt (\bn)}&=& \sum_{\bn\in \IZ_p^2} f_\bt(\bn +\bk+ \mbm_\bt  )\overline{f_\bt(\bn + \mbm_\bt  )}\mu(\bn +\bk+ \mbm_\bt )\mu(\bn)\overline{\mu(\bn+\bk)\mu(\bn + \mbm_\bt ) }
\eeqn
which, after change of index $\bn\to \bn+\mbm_\bt$ on the left hand side,  becomes
\beq
\label{83}0&=& \sum_{\bn\in \IZ_p^2} f_\bt(\bn +\bk+ \mbm_\bt  )\overline{f_\bt(\bn + \mbm_\bt  )}\lt[ e^{\im(\phi(\bn+\bk+\mbm_\bt)-\phi(\bn+\mbm_\bt)+\phi(\bn)-\phi(\bn+\bk))}-1\rt]
\eeq
for all $\bk\in \IZ_{2p-1}^2$. It is convenient to consider the autocorrelation function as $(2p-1)$-periodic function
and endow $\IZ_{2p-1}^2$ with the periodic boundary condition. 

Let us consider those summands on the right side of  \eqref{83}, for any fixed $\bk$,  that share
a common $\phi(\bl)$, for any fixed $\bl$. Clearly, there are at most four such terms:
\beq
&& \lt[e^{\im (\phi(\bl)-\phi(\bl-\bk)-\phi(\bl-\mbm_\bt)+\phi(\bl-\bk-\mbm_\bt))}-1\rt]f_\bt(\bl )\overline{f_\bt(\bl-\bk)}\label{84}\\
&+& \lt[e^{\im (\phi(\bl+\bk)-\phi(\bl)-\phi(\bl+\bk-\mbm_\bt)+\phi(\bl-\mbm_\bt))}-1\rt]f_\bt(\bl+\bk )\overline{f_\bt(\bl)}\nn\\
&+&\lt[e^{\im (\phi(\bl+\mbm_\bt)-\phi(\bl-\bk+\mbm_\bt)-\phi(\bl)+\phi(\bl-\bk))}-1\rt]f_\bt(\bl +\mbm_\bt)\overline{f_\bt(\bl-\bk+\mbm_\bt)}\nn\\
&+& \lt[e^{\im (\phi(\bl+\bk+\mbm_\bt)-\phi(\bl+\mbm_\bt)-\phi(\bl+\bk)+\phi(\bl))}-1\rt]f_\bt(\bl+\bk+\mbm_\bt )\overline{f_\bt(\bl+\mbm_\bt)}. \nn
\eeq
Since the continuous random variable $\phi(\bl)$ does not appear in other summands and hence is independent of them, \eqref{83} implies that \eqref{84} (and the rest of \eqref{83}) vanishes almost surely.

Suppose $\mbm_\bt\not=0$ and consider  any $\bk$ that is linearly independent of $\mbm_\bt$.  The four independent random variables
\beq
\label{85}
\phi(\bl-\bk-\mbm_\bt),\quad \phi(\bl+\bk-\mbm_\bt),\quad \phi(\bl-\bk+\mbm_\bt),\quad \phi(\bl+\bk+\mbm_\bt)
\eeq
appear separately in  exactly one summand in  \eqref{84}. 
Consequently, \eqref{84} can not vanish,  unless 
\beq\label{85-1}
f_\bt(\bl )\overline{f_\bt(\bl-\bk)}=0,&&\overline{f_\bt(\bl)}f_\bt(\bl+\bk )=0\\
\label{85-2}f_\bt(\bl +\mbm_\bt)\overline{f_\bt(\bl-\bk+\mbm_\bt)}=0, &&\overline{f_\bt(\bl+\mbm_\bt)} f_\bt(\bl+\bk+\mbm_\bt )=0
\eeq
in \eqref{84}.

On the other hand, if  $\bk$ is parallel to $\mbm_{\bt}\neq 0$, then for  any 
\beq
\label{85-3}
 \bk\not \in \{ \pm \mbm_{\bt_0},\pm \half  \mbm_{\bt_0},\pm 2  \mbm_{\bt_0}\}
 \eeq
  the four terms in \eqref{85}
appear  separately in exactly one summand in \eqref{84}.  Consequently, \eqref{84} (and hence \eqref{83}) almost surely does not vanishes unless \eqref{85-1} and \eqref{85-2} hold.

Consider $\bk=0$ which satisfies \eqref{85-3} if $\mbm_\bt\neq 0$. Clearly  \eqref{85-1}-\eqref{85-2}  with $\bk=0$ implies 
 that $f_\bt\equiv 0$, which violate the assumption that $\supp(f_\bt)$ is not a subset of a line. Thus $\mbm_\bt=0$ in the first alternative in \eqref{1.8-2}.

Next we prove  that the second alternative in \eqref{1.8-2} is false for all $\mbm_\bt$. 
Otherwise, by  \eqref{44} we have
\beqn
\lefteqn{\sum_{\bn\in\IZ_p^2} f_\bt(\bn+\bk)\overline{f_\bt (\bn)}}\\
&=& \sum_{\bn\in \IZ_p^2} \overline{f_\bt(-\bn-\bk+ \mbm_\bt  )}{f_\bt(-\bn + \mbm_\bt  )}\overline{\mu(-\bn -\bk+ \mbm_\bt )\mu(\bn+\bk)}{\mu(\bn)\mu(-\bn + \mbm_\bt ) }
\eeqn
which, after change of index $\bn\to -\bn-\bk+\mbm_\bt$ on the left hand side,  becomes
\beq
\label{86}0&=&\sum_{\bn\in \IZ_p^2}  \overline{f_\bt(-\bn-\bk+ \mbm_\bt  )}{f_\bt(-\bn + \mbm_\bt  )} \\
&&\quad \quad\cdot \lt[ e^{\im(-\phi(-\bn-\bk+\mbm_\bt)+\phi(-\bn+\mbm_\bt)-\phi(\bn+\bk)+\phi(\bn))}-1\rt]\nn
%&=&\sum_{\bn\in \IZ_p^2}  \overline{f_\bt(\bn-\bk+ \mbm_\bt  )}{f_\bt(\bn + \mbm_\bt  )}\lt[ e^{\im(-\phi(\bn-\bk+\mbm_\bt)+\phi(\bn+\mbm_\bt)-\phi(\bn)+\phi(\bn-\bk))}-1\rt]\nn
\eeq
for all $\bk\in \IZ_{2p-1}^2$. 

Consider those summands in \eqref{86}, for any fixed $\bk$,  that share
a common $\phi(\bl)$, for any fixed $\bl$. Clearly, there are at most four such terms:
\beq
&& \lt[e^{\im (-\phi(\bl)+\phi(\bl+\bk)-\phi(-\bl+\mbm_\bt)+\phi(-\bl-\bk+\mbm_\bt))}-1\rt]\overline{f_\bt(\bl)}f_\bt(\bl+\bk) \label{87}\\
&& +\lt[e^{\im (-\phi(\bl-\bk)+\phi(\bl)-\phi(-\bl+\bk+\mbm_\bt)+\phi(-\bl+\mbm_\bt))}-1\rt]\overline{f_\bt(\bl-\bk)}f_\bt(\bl) \nn\\
&&+\lt[e^{\im (-\phi(-\bl+\mbm_\bt)+\phi(-\bl+\bk+\mbm_\bt)-\phi(\bl)+\phi(\bl-\bk))}-1\rt]\overline{f_\bt(-\bl+\mbm_\bt)}f_\bt(-\bl+\bk+\mbm_\bt) \nn\\
&& +\lt[e^{\im (-\phi(-\bl-\bk+\mbm_\bt)+\phi(-\bl+\mbm_\bt)-\phi(\bl+\bk)+\phi(\bl))}-1\rt]\overline{f_\bt(-\bl-\bk+\mbm_\bt)}f_\bt(-\bl+\mbm_\bt)\nn
\eeq
which must vanish under  \eqref{86}. 

Some observations are in order. First, both $\phi(\bl)$ and $\phi(-\bl+\mbm_\bt)$ appear exactly once in each summand in \eqref{87}. Second, the following pairings of the other phases
\beq
\label{99}
&&\{\phi(\bl+\bk), \phi(-\bl-\bk+\mbm_\bt)\},\quad \{\phi(\bl-\bk), \phi(-\bl+\bk+\mbm_\bt)\}
%&&\{\phi(-\bl+\bk+\mbm_\bt), \phi(\bl-\bk)\},\quad\{\phi(-\bl-\bk+\mbm_\bt), \phi(\bl+\bk)\}
\eeq
also appear exactly twice in \eqref{87}. As long as 
\beq
\label{87'}
\bk&\not =&0\\
\&\quad  \bl&\neq& \mbm_\bt/2, \nn
\eeq
the two sets in \eqref{99}  are not identical and, since each contains at least one element that is independent of the other, we have
\beq
0&=& \lt[e^{\im (-\phi(\bl)+\phi(\bl+\bk)-\phi(-\bl+\mbm_\bt)+\phi(-\bl-\bk+\mbm_\bt))}-1\rt]\overline{f_\bt(\bl)}f_\bt(\bl+\bk) \label{88}\\
&& +\overline{\lt[e^{\im (\phi(-\bl-\bk+\mbm_\bt)-\phi(-\bl+\mbm_\bt)+\phi(\bl+\bk)-\phi(\bl))}-1\rt]}\overline{f_\bt(-\bl-\bk+\mbm_\bt)}f_\bt(-\bl+\mbm_\bt)\nn\\
\label{89}0&=&\lt[e^{\im (-\phi(\bl-\bk)+\phi(\bl)-\phi(-\bl+\bk+\mbm_\bt)+\phi(-\bl+\mbm_\bt))}-1\rt]\overline{f_\bt(\bl-\bk)}f_\bt(\bl) \\
&&+\overline{\lt[e^{\im (\phi(-\bl+\mbm_\bt)-\phi(-\bl+\bk+\mbm_\bt)+\phi(\bl)-\phi(\bl-\bk))}-1\rt]}\overline{f_\bt(-\bl+\mbm_\bt)}f_\bt(-\bl+\bk+\mbm_\bt).\nn
\eeq
Because the two factors
\[
 \lt[e^{\im (-\phi(\bl)+\phi(\bl+\bk)-\phi(-\bl+\mbm_\bt)+\phi(-\bl-\bk+\mbm_\bt))}-1\rt],\quad \lt[e^{\im (-\phi(\bl-\bk)+\phi(\bl)-\phi(-\bl+\bk+\mbm_\bt)+\phi(-\bl+\mbm_\bt))}-1\rt]
 \]
 differ with their complex conjugates  in a random manner independently from $f_\bt$, both \eqref{88} and \eqref{89} are almost surely false unless
 \beq\label{90}
 \overline{f_\bt(\bl)}f_\bt(\bl+\bk)=0,&&f_\bt(\bl) \overline{f_\bt(\bl-\bk)}=0,\\
f_\bt(-\bl+\mbm_\bt) \overline{f_\bt(-\bl-\bk+\mbm_\bt)}=0,
 &&\overline{f_\bt(-\bl+\mbm_\bt)}f_\bt(-\bl+\bk+\mbm_\bt)=0.\label{91}
 \eeq
 
 On the other hand, if $\bl=\mbm_\bt/2$ but $\bk\neq 0$, then
 \beq
 \label{92}\bl+\bk=-\bl+\bk+\mbm_\bt \neq  -\bl-\bk+\mbm_\bt=\bl-\bk,
 \eeq
 and hence \eqref{87} = 0  becomes
 \beq
0&=& \lt[e^{\im (-2\phi(\bl)+\phi(\bl+\bk)+\phi(\bl-\bk))}-1\rt]\overline{f_\bt(\bl)}f_\bt(\bl+\bk)\label{95} + \overline{\lt[e^{\im (-2\phi(\bl)+\phi(\bl+\bk)+\phi(\bl-\bk))}-1\rt] }\overline{f_\bt(\bl-\bk)}f_\bt(\bl)\nn
\eeq
implying \eqref{90}. In other words, \eqref{90} holds true for $\bk\neq 0$.

 Now we show that \eqref{90}  for $\bk\neq 0$  implies 
 that $f_\bt$ has at most one nonzero pixel.
 
Suppose that $f_\bt(\bl)\not = 0$ for some $\bl$. Then by \eqref{90}, $f_\bt(\bn)=0$ for all other $\bn\neq \bl$, i.e.
$f_\bt$ is a singleton  which contradicts the assumption that  $\supp(f_\bt)$ is not a subset of a line.

Consequently the second alternative in \eqref{1.8-2} is false almost surely. 
 The proof  is complete.

\section{Proof of Theorem \ref{thm4}} \label{appD}

The argument is a more detailed, corrected exposition of that for Theorem 5.1 in \cite{Born-tomo} where
the condition $\widehat f(0)\neq 0$ is missing.

Recall that  for  $ (\bk,\bk')\in C_{\bt, \bt'}:=L_{\bt,\bt'}(f)\cap L_{\bt,\bt'}(g)$
\beqn
\widehat g^*_{\bt'}(\bk') e^{-\im 2\pi\bk'\cdot\bl_{\bt'}/p}&=&\widehat g^*_{\bt}(\bk) e^{-\im 2\pi\bk\cdot\bl_{\bt}/p} \\
\widehat f^*_{\bt'}(\bk') e^{-\im 2\pi\bk'\cdot\bl_{\bt'}/p}&=&\widehat f^*_{\bt}(\bk) e^{-\im 2\pi\bk\cdot\bl_{\bt}/p}. 
\eeqn

Suppose that  the first alternative in \eqref{1.8-2} holds true for $\bt'$,  i.e.
\beq
\label{100-2}
g^*_{\bt'}(\bn)&=&
e^{\im \theta_{\bt'}} f^*_{\bt'}(\bn+\mbm_{\bt'})\lamb_{\bt'}(\bn+\mbm_{\bt'})
\eeq
with 
\[\lamb_{\bt'}(\bn)=\mu(\bn)/\mu(\bn-\mbm_{t'}),
\]
implying
\[
\widehat{g^*_{\bt'}}(\bk')=e^{\im \theta_{\bt'}}e^{\im 2\pi \mbm_{\bt'} \cdot\bk'/p} (\widehat f^*_{\bt'}\star \widehat\lamb_{\bt'})(\bk')
\]
where $\star$ denotes the discrete convolution over $\IZ_p^2$.

We now prove that the second alternative in \eqref{1.8-2} can not hold for $\bt$. Otherwise, for some $\mbm_\bt$  \beq
\label{1.27}
g^*_\bt(\bn)&= & e^{\im \theta_\bt} \overline{ f^*_\bt(-\bn +\mbm_\bt)\nu_\bt(-\bn+\mbm_\bt)}
\eeq
with
\[
\nu_\bt(\bn)={ \mu(\bn)}/\overline{\mu(-\bn+\mbm_{t})}. 
\]
implying  
\[
\widehat{g^*_\bt}(\bk)
=\overline{(\widehat f^*_\bt\star \widehat \nu_\bt)(\bk )} e^{-\im 2\pi \mbm_\bt \cdot\bk/p}.
\]

For $(\bk,\bk')\in C_{\bt,\bt'}$, 
\[
e^{\im \theta_{\bt'}}e^{\im 2\pi (\mbm_{\bt'}-\bl_{\bt'}) \cdot\bk'/p} \widehat f^*_{\bt'}\star \widehat\lamb_{\bt'}(\bk')e^{-\im 2\pi\bk'\cdot\bl_{\bt'}/p}= e^{\im \theta_\bt} e^{-\im 2\pi (\mbm_\bt+\bl_\bt) \cdot\bk/p}\overline{\widehat f^*_\bt\star \widehat \nu_\bt }(\bk),\]
implying 
\beq\label{1.30}
0&=&e^{\im \theta_{\bt'}}e^{\im 2\pi (\mbm_{\bt'}-\bl_{\bt'}) \cdot\bk'/p}\sum_{\bn\in \IZ^2_n} e^{\im \phi(\bn)}
e^{-\im \phi(\bn-\mbm_{\bt'})}f^*_{\bt'}(\bn)e^{-\im2\pi \bn\cdot\bk'/p}\\
&-&e^{\im \theta_\bt} e^{-\im 2\pi (\mbm_\bt+\bl_\bt) \cdot\bk/p} \sum_{\bn\in \IZ_n^2} e^{-\im \phi(\bn)}e^{-\im\phi(-\bn+\mbm_\bt)}\overline{f^*_\bt(\bn)}e^{\im 2\pi \bn\cdot\bk/p}. \nn
\eeq

We now show that eq. \eqref{1.30} can not hold for any $\mbm_{\bt'}, \mbm_\bt$. 

For fixed $\bl$,  only one term in \eqref{1.30} contains $e^{\im\phi(\bl)}$: 
\[
e^{\im \theta_{\bt'}}e^{\im 2\pi (\mbm_{\bt'}-\bl_{\bt'}) \cdot\bk'/p}e^{\im \phi(\bl)}
e^{-\im \phi(\bl-\mbm_{\bt'})}f^*_{\bt'}(\bl)e^{-\im2\pi \bl\cdot\bk'/p}
\]
which must vanish by itself following \eqref{1.30} unless the random factors cancel out, i.e.
$\mbm_{\bt'}=0$. 

If $\mbm_{\bt'}\neq 0$ then $f^*_{\bt'}(\bl)=0$ for all $\bl$, contrary to the assumption of a non-line $f^*_{\bt'}$. 
On the other hand, if $\mbm_{\bt'}=0$, the summands of the first sum in \eqref{1.30} are non-random (as $e^{\im\phi(\bn)}e^{-\im \phi(\bn-\mbm_{\bt'})}=1$) while those of the second sum are random (as $e^{-\im \phi(\bn)}e^{-\im\phi(-\bn+\mbm_\bt)}=e^{-\im \phi(\bn)-\im \phi(-\bn)}$) for some $\bn$. Consequently, both sums must vanish separately, in particular,
\[
\sum_{\bn\in \IZ^2_n} f^*_{\bt'}(\bn)e^{-\im2\pi \bn\cdot\bk'/p}=\widehat f^*_{\bt'} (\bk')=0,\quad (\bk,\bk')\in C_{\bt,\bt'}, 
\]
implying $\widehat f(0)=0$ which violates our assumption.

 Consequently,  the only viable alternative for $\bt$ under \eqref{100-2}  is
\beq
\label{1.13}
g^*_\bt(\bn)&=&
e^{\im \theta_\bt} f^*_\bt(\bn+\mbm_\bt )\lamb_\bt(\bn+\mbm_\bt),\quad \forall \bt \in \cT, 
\eeq
{for some} $\mbm_\bt$. 

For $ (\bk,\bk')\in C_{\bt, \bt'}$,
\beq
e^{\im\theta_{\bt}}e^{\im 2\pi (\mbm_{\bt}-\bl_\bt) \cdot\bk}\widehat f^*_{\bt} \star\widehat\lamb_{\bt}(\bk)&=&
e^{\im\theta_{\bt'}}e^{\im 2\pi (\mbm_{\bt'}-\bl_{\bt'}) \cdot\bk'}\widehat f^*_{\bt'} \star\widehat\lamb_{\bt'}(\bk'),\label{1.61}
\eeq
 implying 
\beq\label{1.30'}
0&=&e^{\im \theta_{\bt}}e^{\im 2\pi (\mbm_{\bt}-\bl_\bt) \cdot\bk/p}\sum_{\bn\in \IZ^2_n} e^{\im \phi(\bn)}
e^{-\im \phi(\bn-\mbm_{\bt})}f^*_{\bt}(\bn)e^{-\im2\pi \bn\cdot\bk/p}\\
&-&\nn
e^{\im \theta_{\bt'}}e^{\im 2\pi (\mbm_{\bt'}-\bl_{\bt'}) \cdot\bk'/p}\sum_{\bn\in \IZ^2_n} e^{\im \phi(\bn)}
e^{-\im \phi(\bn-\mbm_{\bt'})}f^*_{\bt'}(\bn)e^{-\im2\pi \bn\cdot\bk'/p}. 
\eeq

Given $\bl$, only the following two terms contain $e^{\im \phi(\bl)}$
\beqn
&&e^{\im \theta_{\bt}}e^{\im 2\pi (\mbm_{\bt}-\bl_\bt) \cdot\bk/p}e^{\im \phi(\bl)}
e^{-\im \phi(\bl-\mbm_{\bt})}f^*_{\bt}(\bl)e^{-\im2\pi \bl\cdot\bk/p}\\
&-&e^{\im \theta_{\bt'}}e^{\im 2\pi (\mbm_{\bt'}-\bl_{\bt'}) \cdot\bk'/p} e^{\im \phi(\bl)}
e^{-\im \phi(\bl-\mbm_{\bt'})}f^*_{\bt'}(\bl)e^{-\im2\pi \bl\cdot\bk'/p}
\eeqn
which must vanish by \eqref{1.30'} unless $\mbm_\bt=0$ or $\mbm_{\bt'}=0$.

If $\mbm_\bt=\mbm_{\bt'}=0,$ then $g^*_\bt(\bn)=e^{\im\theta_\bt} f^*_\bt(\bn)$ and  $g^*_{\bt'}(\bn)=e^{\im\theta_{\bt'} }f^*_{\bt'}(\bn)$ for all $\bn$. 

By Proposition \ref{thm:slice}
$
\widehat f_{\bt}(0)=\widehat f_{\bt'}(0)=\widehat f(0)\neq 0,$
it follows from 
$\widehat g_{\bt}(0)=\widehat g_{\bt'}(0)$ that  $\theta_{\bt}=\theta_{\bt'}$.

If only one of them vanishes, say $\mbm_\bt=0, \mbm_{\bt'}\neq 0,$ then the first sum in \eqref{1.30'} is non-random
while the second sum is random and hence both must vanish separately. In particular
\[
\sum_{\bn\in \IZ^2_n} f^*_{\bt}(\bn)e^{-\im2\pi \bn\cdot\bk/p}=\widehat f^*_\bt (\bk) = 0,\quad (\bk,\bk')\in C_{\bt,\bt'} 
\]
implying $\widehat f(0)=0$ which violates our assumption.

The remaining case, $\mbm_\bt\neq 0\,\,\&\,\,\mbm_{\bt'}\neq 0$ is further split into two sub-cases:
$\mbm_\bt\neq\mbm_{\bt'}$ and $\mbm_\bt= \mbm_{\bt'}.$

Suppose $\mbm_\bt\neq\mbm_{\bt'}$ and both are nonzero. Then the random factors in \eqref{1.30'}
\[
e^{\im \phi(\bn)}
e^{-\im \phi(\bn-\mbm_{\bt})}, \quad e^{\im \phi(\mbm)}
e^{-\im \phi(\mbm-\mbm_{\bt'})} 
\]
can not balance out to satisfy \eqref{1.30'}. 

Consider the remaining undesirable  possibility under \eqref{1.61}: $\mbm_\bt= \mbm_{\bt'}\neq 0.$ Let $\mbm_0:= \mbm_\bt= \mbm_{\bt'}$. Then \eqref{1.30'} becomes 
\beqn
0&=&e^{\im \theta_{\bt}}e^{\im 2\pi (\mbm_{0}-\bl_\bt) \cdot\bk/p}\sum_{\bn\in \IZ^2_n} e^{\im \phi(\bn)}
e^{-\im \phi(\bn-\mbm_{0})}f^*_{\bt}(\bn)e^{-\im2\pi \bn\cdot\bk/p}\\
&-&\nn
e^{\im \theta_{\bt'}}e^{\im 2\pi (\mbm_{0}-\bl_{\bt'}) \cdot\bk'/p}\sum_{\bn\in \IZ^2_n} e^{\im \phi(\bn)}
e^{-\im \phi(\bn-\mbm_{0})}f^*_{\bt'}(\bn)e^{-\im2\pi \bn\cdot\bk'/p}\\
&=& \sum_{\bn\in \IZ^2_n} e^{\im \phi(\bn)}
e^{-\im \phi(\bn-\mbm_{0})} \lt[e^{\im \theta_{\bt}}e^{\im 2\pi (\mbm_{0} -\bl_\bt)\cdot\bk/p}f^*_{\bt}(\bn)e^{-\im2\pi \bn\cdot\bk/p}\rt.\\
&&\hspace{3cm} \lt. -e^{\im \theta_{\bt'}}e^{\im 2\pi (\mbm_{0}-\bl_{\bt'}) \cdot\bk'/p} f^*_{\bt'}(\bn)e^{-\im2\pi \bn\cdot\bk'/p}\rt],
\eeqn
implying 
\beq
\label{135}
e^{\im \theta_{\bt}}e^{\im 2\pi (\mbm_{0}-\bl_\bt) \cdot\bk/p}f^*_{\bt}(\bn)e^{-\im2\pi \bn\cdot\bk/p}
=e^{\im \theta_{\bt'}}e^{\im 2\pi (\mbm_{0}-\bl_{\bt'}) \cdot\bk'/p} f^*_{\bt'}(\bn)e^{-\im2\pi \bn\cdot\bk'/p},\quad \forall \bn,
\eeq
for  $ (\bk,\bk')\in C_{\bt,\bt'}$. 

Rewriting \eqref{135} for $f^*_\bt (\bn)\neq 0$ (then $f^*_{\bt'}(\bn)\neq 0$), we have
\[
e^{\im 2\pi (\bl_{\bt'}\cdot\bk'-\bl_\bt\cdot\bk)/p} e^{\im 2\pi (\mbm_{0}-\bn) \cdot(\bk-\bk')/p}
=e^{\im (\theta_{\bt'}-\theta_\bt)} f^*_{\bt'}(\bn)/f^*_{\bt}(\bn)
\]
whose left hand side is a linear phase factor  and whose right hand side
is independent of  $ (\bk,\bk')\in C_{\bt,\bt'}$. Hence
\[
\bl_{\bt'}\cdot\bk'-\bl_\bt\cdot\bk+(\mbm_0-\bn)\cdot(\bk-\bk')=a\quad\mod p
\]
for some constant $a\in \IR$ for all $\bn$ such that $f^*_\bt(\bn)f^*_{\bt'}(\bn)\neq 0$, and consequently,    by \eqref{135} 
 \beq\label{143}
 e^{\im \theta_{\bt}}f^*_{\bt}(\bn)
=e^{\im a}e^{\im \theta_{\bt'}} f^*_{\bt'}(\bn),\quad \forall \bn.
\eeq
By the common line property, $\theta_\bt=\theta_{\bt'}+a$ and  $f^*_{\bt}
= f^*_{\bt'}$.

Let us turn to  the last undesirable alternative:
\beq
\label{1.14}
g^*_\bt(\bn)\mu(\bn)&=&
 e^{\im \theta_\bt} \overline{ f^*_\bt(-\bn +\mbm_\bt) \mu(-\bn+\mbm_\bt)},\\
 g^*_{\bt'}(\bn)\mu(\bn)&=&
 e^{\im \theta_{\bt'}} \overline{ f^*_{\bt'}(-\bn +\mbm_{\bt'}) \mu(-\bn+\mbm_{\bt'})}.
\eeq

For  $ (\bk,\bk') \in C_{\bt,\bt'}$,
\[
e^{\im \theta_{\bt'}}e^{-\im 2\pi (\mbm_{\bt'}+\bl_{\bt'}) \cdot\bk'/p} \overline{\widehat f^*_{\bt'}\star \widehat\nu_{\bt'}}(\bk')= e^{\im \theta_\bt} e^{-\im 2\pi (\mbm_\bt+\bl_\bt) \cdot\bk/p}\overline{\widehat f^*_\bt\star \widehat \nu_\bt }(\bk),\]
 implying 
\beq\label{1.30''}
0&=&e^{\im \theta_{\bt'}}e^{-\im 2\pi (\mbm_{\bt'}+\bl_{\bt'}) \cdot\bk'/p}\sum_{\bn\in \IZ^2_n} e^{-\im \phi(\bn)}
e^{-\im \phi(-\bn+\mbm_{\bt'})}\overline{f^*_{\bt'}(\bn)}e^{\im2\pi \bn\cdot\bk'/p}\\
&-&e^{\im \theta_\bt} e^{-\im 2\pi (\mbm_\bt+\bl_\bt) \cdot\bk/p} \sum_{\bn\in \IZ_n^2} e^{-\im \phi(\bn)}e^{-\im\phi(-\bn+\mbm_\bt)}\overline{f^*_\bt(\bn)}e^{\im 2\pi \bn\cdot\bk/p}. \nn
\eeq
For fixed $\bl$, only the following four terms contain $e^{-\im \phi(\bl)}$
\beq\label{140}
&&e^{\im \theta_{\bt'}}e^{-\im 2\pi (\mbm_{\bt'}+\bl_{\bt'}) \cdot\bk'/p}e^{-\im \phi(\bl)} \lt[
e^{-\im \phi(\bl-\mbm_{\bt'})}\overline{f^*_{\bt'}(\bl)}e^{\im2\pi \bl\cdot\bk'/p}+ e^{-\im\phi(\mbm_{\bt'}-\bl)}
 \overline{f^*_{\bt'}(\mbm_{\bt'}-\bl)}e^{\im (\mbm_{\bt'}-\bl)\cdot\bk'/p}\rt]\\
&&- e^{\im \theta_{\bt}}e^{-\im 2\pi (\mbm_{\bt}+\bl_\bt) \cdot\bk/p}e^{-\im\phi(\bl)}\lt[
e^{-\im \phi(\bl-\mbm_{\bt})}\overline{f^*_{\bt}(\bl)}e^{\im2\pi \bl\cdot\bk/p}+ e^{-\im\phi(\mbm_\bt-\bl)}
 \overline{f^*_\bt(\mbm_\bt-\bl)}e^{\im (\mbm_\bt-\bl)\cdot\bk/p}\rt]\nn
\eeq
which must sum to zero by \eqref{1.30''}. 

But  the expression in \eqref{140} can not be zero unless $\mbm_\bt=\mbm_{\bt'} (:=\mbm_0)$ and the following equations hold for  $ (\bk,\bk')\in C_{\bt,\bt'}$, 
\beq\label{145}
e^{\im \theta_{\bt'}}e^{-\im 2\pi (\mbm_{0}+\bl_{\bt'}) \cdot\bk'/p} 
\overline{f^*_{\bt'}(\bl)}e^{\im2\pi \bl\cdot\bk'/p}&=& e^{\im \theta_{\bt}}e^{-\im 2\pi (\mbm_{0}+\bl_\bt) \cdot\bk/p}
\overline{f^*_{\bt}(\bl)}e^{\im2\pi \bl\cdot\bk/p}\\
\quad {e^{\im \theta_{\bt'}}e^{-\im 2\pi (\mbm_{0}+\bl_{\bt'}) \cdot\bk'/p} 
\overline{f^*_{\bt'}(\mbm_{0}-\bl)}e^{\im2\pi (\mbm_{0}-\bl)\cdot\bk'/p}}
&=& e^{\im \theta_{\bt}}e^{-\im 2\pi (\mbm_{0}+\bl_\bt) \cdot\bk/p}
\overline{f^*_{\bt}(\mbm_0-\bl)}e^{\im2\pi (\mbm_0-\bl)\cdot\bk/p}. \label{146}
\eeq

For $f^*_{\bt'}(\bl)f^*_{\bt}(\bl)\neq 0$, \eqref{145}  implies that for  $ (\bk,\bk')\in C_{\bt,\bt'}$ and some constant $a\in \IR$,
\beq
\label{142}
a+ (\mbm_{0}+\bl_{\bt'}-\bl) \cdot\bk'= (\mbm_{0}+\bl_{\bt}-\bl) \cdot\bk\quad\mod p, 
 \eeq
and consequently, 
  \beq
 \label{147}
 e^{-\im \theta_{\bt'}}e^{-\im\theta_0}
{f^*_{\bt'}(\bl)}= e^{-\im \theta_{\bt}}{f^*_{\bt}(\bl)}.
\eeq
The same analysis for \eqref{146} leads to the equivalent equation   \eqref{147}.
By the common line property, $\theta_\bt=\theta_{\bt'}+a$ and  $f^*_{\bt}
= f^*_{\bt'}$. 

\commentout{% Interesting observation
Taking  the Fourier transform on \eqref{143}, we have
\[
e^{\im\theta_\bt} \widehat f_\bt(\bk)=e^{\im\theta_{\bt'}} \widehat f_{\bt'} (\bk).
\]
By Fourier slice theorem, we have $\widehat f_\bt(\bk)= \widehat f_{\bt'} (\bk'),\quad \bk,\bk'\in P_\bt\cap P_{\bt'}$. 
}

The two undesirable ambiguities \eqref{143} and \eqref{147} is  summarized by the second alternative in \eqref{alt}. 
The proof is complete.

\end{appendix}

\end{document}